\documentclass[11pt]{amsart}

\pdfoutput = 1

\usepackage[english]{babel}
\usepackage[autostyle]{csquotes}
\usepackage{amsmath,amssymb,amsthm}
\usepackage{amsfonts}
\usepackage{amsxtra}

\usepackage{array,dcolumn}
\usepackage{graphicx}
\usepackage[hyperfootnotes=true,
        pdffitwindow=true,
        plainpages=false,
        pdfpagelabels=true,
        pdfpagemode=UseOutlines,
        pdfpagelayout=SinglePage,
                hyperindex,]{hyperref}

\usepackage{lmodern}
\usepackage[T1]{fontenc}
\usepackage[utf8]{inputenc}
\usepackage[activate={true,nocompatibility},spacing,kerning]{microtype}
\usepackage{bm}
\usepackage{bbm}
\usepackage[sc,osf]{mathpazo}
\microtypecontext{spacing=nonfrench}

  \calclayout

\newcommand{\mathsym}[1]{{}}
\newcommand{\unicode}[1]{{}}

\theoremstyle{plain}
\newtheorem{theorem}{Theorem}

\newtheorem{proposition}[theorem]{Proposition}

\theoremstyle{definition}

\theoremstyle{remark}
\newtheorem{remark}[theorem]{Remark}

\numberwithin{equation}{section}

\numberwithin{theorem}{section}

\numberwithin{figure}{section}

\begin{document}


\title[]{A review of exact results for fluctuation formulas  in random matrix
theory}
\author{Peter J. Forrester}
\address{School of Mathematics and Statistics, 
ARC Centre of Excellence for Mathematical \& Statistical Frontiers,
University of Melbourne, Victoria 3010, Australia}
\email{pjforr@unimelb.edu.au}

\date{}


\begin{abstract}
Covariances and variances of linear statistics of a point process can be written as
integrals over the truncated two-point correlation function. When the point process
consists of the eigenvalues of a random matrix ensemble, there are often large $N$ universal
forms for this correlation after smoothing, which results in particularly simple limiting formulas
for the fluctuation of the linear statistics. We review these limiting formulas,
derived in the simplest cases as corollaries of explicit knowledge of the truncated two-point
correlation. One of the large $N$ limits is to scale the eigenvalues so that limiting
support is compact, and the linear statistics vary on the scale of the support. This is
a global scaling. The other, where a thermodynamic limit is first taken so that the
spacing between eigenvalues is of order unity, and then a scale imposed on
the test functions so they are slowly varying, is the bulk scaling. The latter was already
identified as a probe of random matrix characteristics for quantum spectra in the
pioneering work of Dyson and Mehta.
\end{abstract}


\maketitle

\tableofcontents

\section{Introduction}\label{S1}
The formulation of random matrix theory for applications to the spectra of complex quantum systems was laid out
in pioneering works of Wigner, Gaudin, Mehta and Dyson. Preprints of the original papers, dating from the late
1950's and early 1960's, are conveniently bundled in
a book edited by Porter \cite{Po65}, itself published in 1965, which contains too a review of this early literature.
During the 1980's, as a fundamental  contribution to the subject of quantum chaos, work of Bohigas
et al.~\cite{BGS84} made it clear that the correct meaning to give to a ``complex quantum system'' is any quantum
system for which the underlying classical dynamics is chaotic. Single particle quantum systems in this class,
such as kicked tops and irregular billiard domains, were subsequently studied intensely; see \cite{Ha00}.

To test against random matrix predictions two statistical quantities, both of which were prominent in the works
of the pioneering researchers cited above, were preferred. Assuming an unfolding of the energy levels so that their
mean spacing becomes unity, one is the distribution of the spacing between consecutive eigenvalues, while the
other is the so-called number variance, $\Sigma^2(L)$ say, corresponding to the fluctuation of the number of
eigenvalues in an interval of length $L$, assumed large. From a theoretical viewpoint, these statistics have distinct
characteristics.

Define
\begin{equation}\label{1}
N_L = \sum_l \chi_{\lambda_l \in [0,L]},
\end{equation}
where $\{ \lambda_l \}$ denotes the unfolded eigenvalues labelled from some origin in the bulk.
Then, by the unfolding assumption,
$$
\langle N_L \rangle = L,
$$
while the variance is precisely the number variance
\begin{equation}\label{2}
\langle (N_L - L)^2 \rangle = \Sigma^2(L).
\end{equation}
Generally a statistic $A = \sum_l a(\lambda_l)$ is referred to as a linear statistic. Thus
$\Sigma^2(L)$ is the variance of the particular linear statistic (\ref{1}). In contrast, the spacing distribution 
begin a function of consecutive eigenvalues, does
not relate to the structure of a linear statistic.

The present review focusses attention on fluctuation formulas associated with linear statistics.
In addition to the quantity $N_L$, there are now many linear statistics and random matrix ensembles for
which knowledge of the corresponding distributional properties is of applied interest. Moreover,
consideration of the calculation of these distributions involves rich mathematical structures, with
the scope for further research. At the same time the existing literature is vast, necessitating some
restriction to the scope of the review. Thus for the most part we consider only the formulas for
covariances and variances of linear statistics. Where possible we relate these formulas to the
corresponding smoothed truncated two-point correlation function.

Throughout Section \ref{S2} we identify models in random matrix theory for which the
truncated two-point correlation function is known explicitly and has a sufficiently simple
analytic structure to allow large $N$ analysis of the covariances and variances of linear statistics. 
There are two types of large $N$ limits which lead to structured results. One is when
the eigenvalues are scaled to have a limiting compact support, with the test functions
by way of the linear statistics varying on the scale of the support. The other is when
the test functions are first chosen to vary on the scale of the mean spacing between
eigenvalues, and a scaling is chosen so that in the large $N$ limit
the eigenvalues are on average a unit distance apart. Next a scale $L$ is introduced
into the test functions, and the limit $L \to \infty$ is taken. Some understanding of
the structures found can be given by adapting a log-gas viewpoint, which we do
in subsection \ref{S2.10b}.

In Section \ref{S3} we review fluctuation formulas which are generalisations of
those encountered in Section \ref{S2}, but which require a more challenging
analysis. First considered are the classical $\beta$-ensembles, where a loop
equation analysis suffices to obtain the fluctuation formulas in a global scaling.
Most prominent in this class are the Gaussian orthogonal and unitary ensemble
cases. They permit numerous generalisations, and we take note of the corresponding
fluctuation formulas of a number of them. One of these generalisations is to Wigner matrices, where the
independent Gaussian entries are replaced by a more general zero mean, fixed
standard deviation random variable. 
The classical Laguerre ensembles, realised in the case of orthogonal and symplectic symmetry
by the matrix structure $W^\dagger W$ for $W$ a rectangular Gaussian matrix, also
permit generalisations. One is to consider
the eigenvalues of $W_M^\dagger W_M$, where $W_M= G_M G_{M-1} \cdots G_1$,
with each $G_j$ an $N_j \times N_{j-1}$ a rectangular complex Gaussian matrix. In the
global scaling limit, there is a simple formula for the variance of a polynomial linear
statistic. The  predictions of this formula can be compared with that obtained from
a loop equation analysis. The final topic considered is that of variance formulas
associated with the eigenvalues of Ginibre matrices, i.e.~non-Hermitian square Gaussian 
random matrices, in the global scaling limit.

\section{Formalism and simple examples}\label{S2}
\subsection{Covariance, variance and correlation functions}\label{S2.1}
We take the viewpoint of the eigenvalues for a random matrix ensemble as an example of a continuous
point or particle process. For notational convenience we regard the points as confined to an interval
$I$ of the real line, although this in not necessary --- the domain may as well be in higher dimensions.
For $N$ particles we let $p_N(x_1,\dots, x_N)$ denote their joint probability density function. 
Integrating out all but one, respectively two, particles gives
the corresponding one and two point correlations
\begin{align}\label{2.1}
\rho_{(1),N}(x_1) & = N \int_I dx_2 \cdots \int_I d x_N \, p_N(x_1,\dots, x_N) \nonumber \\
\rho_{(2),N}(x_1,x_2) & = N (N - 1)  \int_I dx_3 \cdots \int_I d x_N \, p_N(x_1,\dots, x_N).
\end{align}
Equivalently
\begin{align}\label{2.2}
\rho_{(1),N}(x) & =  \Big \langle  \sum_{l=1}^N  \delta (x - x_l) \Big \rangle 
\nonumber \\
\rho_{(2),N}(x,x') & =  \Big \langle  \sum_{l,l'=1 \atop l \ne l'}^N  \delta (x - x_l)   \delta (x' - x_{l'}) \Big \rangle .
\end{align}
For large separation we expect $\rho_{(2),N}(x_1,x_2) \approx  \rho_{(1),N}(x_1)  \rho_{(1),N}(x_2) $ which
motives introducing the truncated (or connected) two point correlation
\begin{equation}\label{2.3}
\rho_{(2),N}^T(x_1,x_2) = \rho_{(2),N}(x_1,x_2) - \rho_{(1),N}(x_1) \rho_{(1),N}(x_1).
\end{equation}

A simple but fundamental result is that these correlations relate to
 the covariance between two linear statistics, defined by
\begin{equation}\label{3}
{\rm Cov} \, \Big ( \sum_{l=1}^N f(x_l),  \sum_{l=1}^N g(x_l) \Big ) :=
\Big \langle   \sum_{l,l'=1}^N f(x_l) g(x_{l'}) \Big \rangle - \Big \langle   \sum_{l=1}^N f(x_l)  \Big \rangle    \Big \langle  \sum_{l=1}^N g(x_l)  \Big \rangle.
\end{equation}

\begin{proposition}\label{p2.1}
We have
\begin{align}
{\rm Cov} \, \Big ( \sum_{l=1}^N f(x_l),  \sum_{l=1}^N g(x_l) \Big )  & =  \int_I dx   \int_I dx' \, f(x) g(x') \Big ( \rho_{(2),N}^T(x,x') + \rho_{(1),N}(x) \delta(x - x') \Big ) \label{3.1} \\
& = - {1 \over 2}  \int_I dx   \int_I dx' \, (f(x) - f(x'))   (g(x) - g(x'))  \rho_{(2),N}^T(x,x').    \label{3.2}
\end{align}
\end{proposition}

\begin{proof}
In relation to the first term on the RHS of (\ref{3}) we have
\begin{align}\label{3.3}
\Big \langle   \sum_{l,l'=1}^N f(x_l) g(x_{l'}) \Big \rangle  &  
=  \int_I dx   \int_I dx' \, f(x) g(x')  \Big \langle  \sum_{j,k=1 \atop j \ne k}^N  \delta (x - x_j)   \delta (x' - x_k) \Big \rangle \nonumber \\ & \quad +
  \int_I f(x) g(x)  \Big \langle  \sum_{j =1 }^N  \delta (x - x_j)  \Big \rangle \, dx   \nonumber \\
  & = \int_I dx   \int_I dx' \, f(x) g(x') \Big ( \rho_{(2),N}(x,x')  + \delta(x - x') \rho_{(1),N}(x)   \Big ),
  \end{align}
 while the second term allows the simple rewrite
\begin{equation}\label{3.4}
\Big \langle   \sum_{l=1}^N f(x_l)  \Big \rangle    \Big \langle  \sum_{l=1}^N g(x_l)  \Big \rangle  =
 \int_I dx   \int_I dx' \, f(x) g(x')  \rho_{(1),N}(x)   \rho_{(1),N}(x') .
\end{equation} 
Subtracting (\ref{3.4}) from (\ref{3.3}), substituting in (\ref{3}) and recalling the definition (\ref{2.3}) gives (\ref{3.1}).

In relation to (\ref{3.2}), comparing with (\ref{3.1}) we see that it suffices to show
$$
- \int_I dx \, f(x) g(x) \int_I d x' \,  \rho_{(2),N}^T(x,x') = \int_I f (x) g(x) \rho_{(1),N}(x) \, dx.
$$
This is true since we can check from the definitions that
\begin{equation}\label{3.4a}
 \int_I   \rho_{(2),N}^T(x,x') \, dx' = - \rho_{(1),N}(x).
\end{equation} 
\end{proof}

We remark that
a simple consequence of the first equation in (\ref{2.2}) is the formula for the mean
\begin{equation}\label{3.4b}
{\mathbb E} \, \Big ( \sum_{l=1}^N f(x_l) \Big ) := \Big \langle   \sum_{l=1}^N f(x_l)  \Big \rangle  =
 \int_I f(x)  \rho_{(1),N}(x)  \, dx,
 \end{equation} 
 as already used in (\ref{3.4}), and
 setting $f=g$ in Proposition \ref{p2.1} gives the corresponding result for the variance. 
 We note too that the quantity
 \begin{equation}\label{C.1}
 C_{(2),N}(x,x') :=  \rho_{(2),N}(x,x')  + \delta(x - x') \rho_{(1),N}(x)
  \end{equation} 
  has the interpretation of a response density which is induced by there being an eigenvalue at point $x'$.

 \subsection{Global scaling limit of CUE random matrices}
 
 Haar distributed random unitary matrices, or equivalently the circular unitary ensemble (CUE), provides
 the simplest example within random matrix theory of a significant simplification of the formulas in
 Proposition \ref{p2.1} for the covariances. First, with the eigenvalues of the CUE all being on the unit
 circle in the complex plane, we interpret the $x_l$ as the angle parametrising the eigenvalues, and so
 take $I = [0,2\pi)$. For the density we then have $\rho_{(1),N}(x) = N/2 \pi$ independent of the angle
 $x$, and $ C_{(2),N}(x,x') $ (recall the notation (\ref{C.1})) is a function of $x-x'$ which is periodic
 of period $2 \pi$ in $x-x'$.  The corresponding Fourier series is well known (see e.g.~\cite[\S 5.2]{WF15})
 to have the simple form
 \begin{equation}\label{3.4c}
  C_{(2),N}(x,x')   = {1 \over (2 \pi)^2} \sum_{l=-\infty}^\infty m_l^{\rm CUE} e^{i l (x-x')},
  \quad m_l^{\rm CUE} = \begin{cases} |l|, & |l| < N \\
  N, & |l| \ge  N.  \end{cases}
 \end{equation} 
 Substituting in (\ref{3.1}) allows for simplification to a single sum.

 \begin{proposition}\label{p2.1z}
We have
\begin{equation}\label{3.4d}
{\rm Cov}^{\rm CUE} \, \Big ( \sum_{l=1}^N f(x_l),  \sum_{l=1}^N g(x_l) \Big )   =
\sum_{l=-\infty}^\infty  m_l^{\rm CUE} f_l g_{-l}  ,
 \end{equation} 
 where  $f_l = (1/2 \pi) \int_0^{2 \pi} f(x) e^{i l x} \, dx$ and similarly the meaning of $g_{-l}$.
 Moreoever, if $f$ and $g$ are differentiable on $[0,2\pi)$ with $f', g'$ H\"older continuous
 of order $\alpha > 0$ then
\begin{equation}\label{3.4e}
\lim_{N \to \infty} {\rm Cov}^{\rm CUE} \, \Big ( \sum_{l=1}^N f(x_l),  \sum_{l=1}^N g(x_l) \Big )   =
\sum_{l=-\infty}^\infty    | l |    f_l g_{-l},
 \end{equation} 
while if $f = g =   \chi_{[0,L]}$ ($0 < L < 2 \pi$) and $N_L$ is specified by (\ref{1}) then
\begin{equation}\label{3.4f}
\lim_{N \to \infty}  {1 \over \log N} {\rm Var}^{\rm CUE} \,  ( N_L  ) = {1 \over \pi^2}.
 \end{equation} 
 \end{proposition} 
 
 \begin{proof}
 It remains to consider (\ref{3.4e}) and (\ref{3.4f}). To deduce  (\ref{3.4e}) from (\ref{3.4d}), the
 essential point is that the stated conditions on $f$ and $g$ imply that for large $l$ the decay of
 $f_l g_{-l}$ is $O(1/l^{2(1 + \alpha)})$. This allows the sum over $l$ on the
 RHS of  (\ref{3.4d}) to be truncated at $| l | < N$ in the large $N$ limit and tells us too that the sum on the RHS
 of (\ref{3.4e}) converges.
 
 In relation to (\ref{3.4f}), with $f = g = \chi_{[0,L]}$, $0 < L < 2 \pi$,
  we compute that for $l \ne 0$, $f_l = ( e^{i l L} - 1)/(2 \pi i l)$. This substituted in
 (\ref{3.4d}) implies the stated result.
 
 \end{proof}
 
 Let $0<L_1, L_2 < 2 \pi$ with $L_1 \ne L_2$. We see from (\ref{3.4d}) and the derivation of (\ref{3.4f}) that
 \begin{align}\label{3.4g}
 \lim_{N \to \infty} {\rm Cov}^{\rm CUE}(N_{L_1}, N_{L_2}) & = {1 \over \pi^2} \sum_{l=1}^\infty {1 \over l} \Big (
 \cos l(L_1 - L_2)/2 - \cos l(L_1 + L_2)/2 \Big ) \nonumber \\
 & = {1 \over \pi^2} \Big ( \log | \sin(L_1 + L_2)/2 | - \log  | \sin(L_1 - L_2)/2 | \Big ).
 \end{align}
 Note that this diverges for $L_1 = L_2$, which is in keeping with (\ref{3.4f}). We refer to 
 \cite{SLMS21} for the analogous result in the case of complex Gaussian Hermitian random
 matrices, and a discussion of further context.
 
 \begin{remark}\label{R1}
 1.~The characteristic function, $\hat{P}_{N,f}(t)$ say, for the distribution of the linear statistic
 $\sum_{l=1}^N f(x_l)$ in the setting of the first paragraph of \S \ref{S2.1} is given by
\begin{equation}\label{Pp} 
\hat{P}_N(t;f) = \int_I dx_1 \cdots  \int_I dx_N \, \Big ( \prod_{l=1}^N e^{ i t f (x_l)} \Big )
p_N(x_1,\dots,x_N).
\end{equation}
In the case of the CUE, $I = [0,2 \pi)$ and 
\begin{equation}\label{Pp1} 
p_N(x_1,\dots,x_N) = {1 \over (2 \pi )^N N!} \prod_{1 \le j < k \le N} | e^{i x_k} - e^{i x_j} |^2;
\end{equation}
see e.g.~\cite[Prop.~2.2.5 with $\beta = 2$]{Fo10}.
A well known variant of the Andr\'eief identity  (see  e.g.~\cite{Fo18}) allows
(\ref{Pp}) to be written as the Toeplitz determinant
\begin{equation}\label{Pp2} 
\hat{P}_N(t;f) =  \det \bigg [ {1 \over 2 \pi} \int_0^{2 \pi} e^{i t f(x)} e^{i (j - k ) x} \, dx \bigg ]_{j,k=1}^N.
\end{equation}
For $f$ satisfying the conditions of Proposition \ref{p2.1z}, the celebrated strong Szeg\"o theorem
(see e.g.~\cite{Ba12})  gives
\begin{equation}\label{Pp3} 
\lim_{N \to \infty} e^{- i t N f_0 } \hat{P}_{N}(t;f) =  \exp \Big ( -  t^2 \sum_{l=1}^\infty l f_l f_{-l} \Big ).
\end{equation}

On the other hand, according to the cumulant expansion, for small $t$
\begin{equation}\label{Pp4} 
\hat{P}_{N,f}(t) = e^{i t {\mathbb E}( \sum_{j=1}^N f (x_j)) - (t^2/2) {\rm Var} \, (\sum_{j=1}^N f (x_j)) + O(t^3)}.
\end{equation}
Comparing (\ref{Pp3}) to (\ref{Pp4}) shows consistency with (\ref{3.4d}) in the case $f=g$. Moreover (\ref{Pp3}) gives
that the limiting distribution of the centred linear statistic under the conditions of Proposition \ref{p2.1z} is a Gaussian.
This interpretation of the strong Szeg\"o theorem was first given by Johansson \cite{Jo88}. \\
2.~We see from (\ref{3.4c}) that $m_l^{\rm CUE}$ is independent of $N$ for all $|l| < N$. Closely related to this is the fact that the
distribution of $|{\rm Tr} \, U^k|^2 = |\sum_{j=1}^N e^{i x_j k }|^2$, for $k$ a positive integer less than or equal to $N$, is such that its first
$N$ moments coincide with the corresponding moments of $\sqrt{k}$ times
a standard complex Gaussian random variable for $k \le N$ \cite{Jo97}, \cite{DE01}. \\
3.~Consistent with (\ref{3.4c}) is the functional form \cite{Dy62a}
\begin{equation}\label{Pp5} 
\rho^T_{(2),N}(x,x') = - \bigg ( {\sin N ( x - x')/2 \over 2 \pi \sin (x - x')/2} \bigg )^2 = - {1 \over 8 \pi^2} \bigg ( {1 \over \sin^2(x - x')/2} - {\cos N (x - x') \over \sin^2(x - x')/2} \bigg ).
\end{equation}
Substituting in (\ref{3.2}) then shows
\begin{multline}\label{Pp6} 
\lim_{N \to \infty} {\rm Cov}^{\rm CUE}  \,  \, \Big ( \sum_{l=1}^N f(x_l),  \sum_{l=1}^N g(x_l) \Big ) =
{1 \over 16 \pi^2} \int_0^{2 \pi} dx  \int_0^{2 \pi} dx' \, {(f(x) - f(x')) (g(x) - g(x')) \over \sin^2 (x - x')/2 } \\
 = - {1 \over 2 \pi^2}  \int_0^{2 \pi} dx  \int_0^{2 \pi} dx' \,  \Big ( {d \over dx} f(x) \Big )  \Big ( {d \over dx'} g(x') \Big ) \log | \sin (x - x')/2)|,
\end{multline}
valid provided $f$ and $g$ are H\"older continuous
 of order $\alpha > 1/2$ in the first expression, and satisfy the conditions required for (\ref{3.4e})
 in the second
 (these conditions ensure the double integrals converge).
 The second formula follows from the first upon using the identity
 \begin{equation}\label{Pp7} 
 {1 \over 4} {1 \over \sin^2 ( x - x')/2} =  
 {\partial^2 \over \partial x \partial x'} \log \Big | \sin \Big ({x - x' \over 2} \Big )  \Big |.
 \end{equation}
 and integrating by parts. The Fourier expansion
 \begin{equation}\label{Pp7a} 
   \log \Big | \sin \Big ({ x - x' \over 2} \Big ) \Big | =    \sum_{p=-\infty}^\infty \alpha_p e^{i p ( x - x')}, \:
 \alpha_p = {1 \over 2 \pi} \int_0^{2 \pi} \log \Big | \sin \Big (  { x \over 2}  \Big ) \Big | e^{-i p x} \, dx = \! - {1 \over 2 |p|} \: (p \ne 0),
  \end{equation}
 provides a direct transformation from the second formula in (\ref{Pp6}), to the formula of
 (\ref{3.4d}), upon integration by parts.

\end{remark}

For a given $N$ it is straightforward to sample $N \times N$ CUE matrices --- indeed this is now an inbuilt function in
the Mathematica computer algebra package --- and to numerically compute the corresponding eigenangles $\{ x_j \}_{j=1}^N$.
This allows the result (\ref{3.4e}) to be illustrated for a particular linear statistic. We choose $f(x) = g(x) = \cos 2 x$.
The corresponding random function $\sum_{j=1}^N \cos 2 x_j$ then has mean zero, and according to (\ref{3.4e}) has
variance equal to $2$.
A plot of the value of a single realisation of this random function for $N=1,2,\dots,150$, and with successive values
joined by straight lines as a visual aid, is given in Figure \ref{Ffig1}. 
\begin{figure*}
\centering
\includegraphics[width=0.75\textwidth]{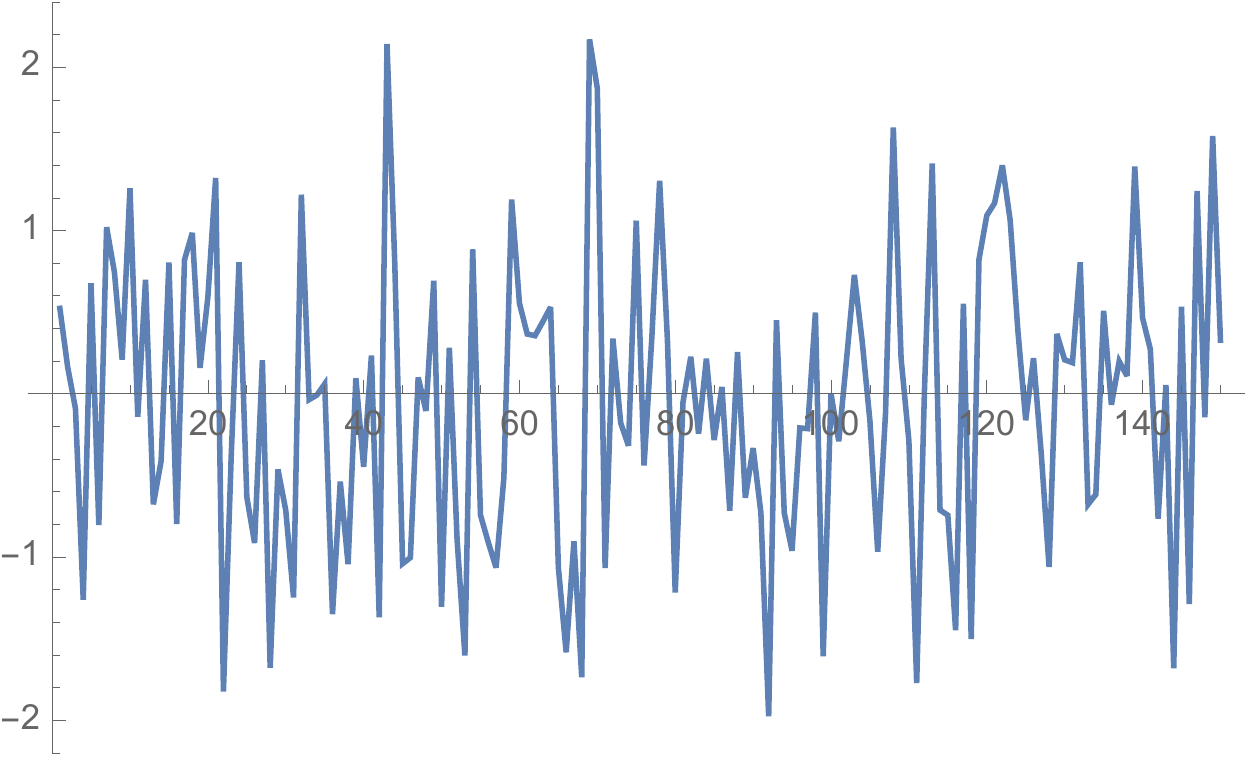}
\caption{Plot of values of a single realisation of the random function $\sum_{j=1}^N \cos 2 x_j$ for $\{x_j\}$ the eigenangles
of an $N \times N$ CUE matrix, for $N$ from $1$ up to $150$.}
\label{Ffig1}
\end{figure*}

\subsection{Global scaling limit for a deformation of the CUE}
The non-oscillatory term in (\ref{Pp5}) is independent of $N$, which according to (\ref{3.2})
is responsible for  the leading large $N$ covariance itself being independent of $N$, as seen in (\ref{Pp5}).
In contrast, for eigenvalues behaving as a perfect gas of non-interacting particles on $[0,2\pi)$,
the joint eigenvalue density is 
 \begin{equation}\label{q0} 
p_N(x_1,\dots,x_N) = {1 \over (2 \pi)^N}, 
 \end{equation}
and we see from (\ref{2.2}) and
(\ref{2.3}) that the truncated two-particle correlation is now proportional to $N$,
 \begin{equation}\label{q1} 
 \rho_{(2),N}^T(x_1, x_2) = - {N \over (2 \pi)^2}.
 \end{equation}
 Substituting in (\ref{3.2}) gives that the corresponding covariance is similarly proportional to $N$,
 \begin{equation}\label{q2}  
{\rm Cov} \, \Big ( \sum_{l=1}^N f(x_l),  \sum_{l=1}^N g(x_l) \Big )  =
{N \over 8 \pi^2} 
 \int_0^{2 \pi} dx  \int_0^{2 \pi} dx' \, (f(x) - f(x')) (g(x) - g(x')). 
 \end{equation}
 
 In the mid 1960's Gaudin \cite{Ga66} introduced into random matrix theory an eigenvalue
 PDF interpolating between the CUE and Poisson functional forms,
 (\ref{Pp1}) and (\ref{q0}) respectively,
  \begin{equation}\label{q3}  
  {\alpha^{N(N-1)/2} \over {Q}_N} \prod_{1 \le j < k \le N} \bigg | {e^{i x_j} - e^{i x_k} \over e^{i x_j}  - \alpha e^{i x_k} } \bigg |^2 =
  {1 \over {Q}_N}  \prod_{1 \le j < k \le N} \bigg  ( 1 + {\sinh^2 \gamma \over \sin^2 (x_j - x_k)/2} \bigg )^{-1},
  \end{equation}
  with $0 < x_l < 2 \pi$ ($l=1,\dots,N)$, $\alpha:= e^{-2 \gamma}$ and
  \begin{equation}\label{q4}  
  {Q}_N = (2 \pi )^N N! \alpha^{N(N-1)/2} \prod_{k=1}^N {1 - \alpha \over 1 - \alpha^k}.
  \end{equation} 
  Thus taking the limit $\alpha \to 1$, or equivalently $\gamma \to 0$ reclaims (\ref{q0}),
  while taking $\alpha \to 0$, or equivalently $\gamma \to \infty$ reclaims  (\ref{Pp1}).
  We will refer to this ensemble as the CUE${}_\alpha$.
  In \cite{Fo93a} the CUE${}_\alpha$ was related to the
  theory of parametric eigenvalue motion due to Pechukas \cite{Pe83} and Yukawa \cite{Yu86},
  and most recently it was placed within the theory of circulant $L$-ensembles \cite{Fo21b}.
  
  As an interpolation between (\ref{Pp5}) and (\ref{q1}) it was found in \cite[\S 4.3]{Ga66} that
   \begin{equation}\label{S1c}    
 \rho_{(2)}^T(x_1,x_2;\alpha) = - {N \over (2 \pi)^2} - {2 \over (2 \pi)^2} \, {\rm Re} \,
 \sum_{0  \le \mu_1 < \mu_2 \le N }    \prod_{k= \mu_1}^{\mu_2 - 1}
 {1 - \alpha^k \over e^{i (x_1 - x_2)} - \alpha^k}.
  \end{equation} 
  Thus from this we see that 
  \begin{equation}\label{S1b} 
   \lim_{\alpha \to 0}  \rho_{(2)}^T(x_1,x_2;\alpha) = - {N \over (2 \pi)^2} - {1 \over (2 \pi)^2}
   \sum_{\mu_1, \mu_2 = 0 \atop \mu_1 \ne \mu_2} e^{i (\mu_1 - \mu_2)(x_1 - x_2)} (N - | \mu_1 - \mu_2|),
 \end{equation}
 which is equivalent to the Fourier expansion (\ref{3.4c}). We read off too the limiting behaviour
   \begin{equation}\label{S1d}    
   \lim_{\alpha \to 1}  \rho_{(2)}^T(x_1,x_2;\alpha) = - {N \over (2 \pi)^2},
   \end{equation}
   in agreement with (\ref{q1}).  
   
   Of interest is the analogue of (\ref{3.4e}) for $0 < \alpha < 1$. In relation to this, and with $[w^{-p}] f(w)$ denoting
   the coefficient of $w^{-p}$ in the Laurent expansion of $f(w)$, for $p \in \mathbb Z^+$ set
    \begin{equation}\label{S0a}    
    m_p^{(\alpha)} = m_{-p}^{(\alpha)} = \lim_{N \to \infty} [w^{-p}] \bigg (
    N - \sum_{0  \le \mu_1 < \mu_2 \le N }    \prod_{k= \mu_1}^{\mu_2 - 1}
 {1 - \alpha^k \over w - \alpha^k} \bigg ).
  \end{equation}
  Consideration of the small $\alpha$ expansion of this quantity shows it to be well defined. While the Laurent expansion
  of the quantity in brackets in (\ref{S0a}) is complicated for finite $N$ and general $p$, the coefficients greatly simplify in the
  limit $N \to \infty$. The mechanism is that the terms for which it is difficult to predict their coefficients do not occur until the order of
  $\alpha^{O(N)}$ in the small $\alpha$ expansion, and thus vanish. The terms before that have a regular pattern, allowing us to conclude
     \begin{equation}\label{S0b}  
   m_p^{(\alpha)} =    |p| \sum_{j=0}^\infty \alpha^{|p|j} = { |p| \over 1 - \alpha^{|p|}}.
  \end{equation}  
  Consequently, with $f$ and $g$ in the class of functions specified in the statement of (\ref{3.4e}), the latter generalises
  to read 
  \begin{equation}\label{4.4e1}
\lim_{N \to \infty} {\rm Cov}^{{\rm CUE}_\alpha} \, \Big ( \sum_{l=1}^N f(x_l),  \sum_{l=1}^N g(x_l) \Big )   =
\sum_{l=-\infty \atop l \ne 0}^\infty     {| l |  \over 1 - \alpha^{|l|}}  f_l g_{-l}.
 \end{equation}
 
 Let $C_{(2),N}^{(\alpha)}$ be specified by (\ref{C.1}) as it applies to the CUE${}_\alpha$. It follows from 
 (\ref{S1c}), (\ref{S0a}) and (\ref{S0b}) that we have
   \begin{equation}\label{4.4f}
C_{(2),\infty}^{(\alpha)}(x,x') := \lim_{N \to \infty}    C_{(2), N}^{(\alpha)}(x,x') = {1 \over (2 \pi)^2} \sum_{p=-\infty \atop p \ne 0 }^\infty
{|p| \over 1 - \alpha^{|p|}} e^{i p ( x - x')}.
 \end{equation}
 This functional form permits the alternative expression
   \begin{align}\label{4.4g} 
 C_{(2),\infty}^{(\alpha)}(x,x') & = - {1 \over (2 \pi)^2} {\partial^2 \over \partial x^2} \log \Big ( \prod_{l=0}^\infty
 (1 - \alpha^l e^{i (x - x')})    (1 - \alpha^l e^{-i (x - x')}) \Big ) \nonumber \\
 & = - {1 \over (2 \pi )^2 }  {\partial^2 \over \partial x^2} \log  | \theta_1( x - x';\alpha^{1/2})  |,
 \end{align}
 where
$$
\theta_1(z;q) :=i  \sum_{n=-\infty}^\infty  (-1)^n q^{(n+1/2)^2} e^{(2n+1) i z}.
$$
Here the first equality can be seen to agree with (\ref{4.4f}) by a direct calculation, while the second equality
requires knowledge of the product formula for the Jacobi theta function $\theta_1$.
Hence, analogous to the second line of (\ref{Pp6}), in addition to (\ref{4.4e1}) we have
\begin{multline}\label{Pp6.1} 
\lim_{N \to \infty} {\rm Cov}^{{\rm CUE}_\alpha}  \,  \, \Big ( \sum_{l=1}^N f(x_l),  \sum_{l=1}^N g(x_l) \Big )\\
 = - {1 \over 2 \pi^2}  \int_0^{2 \pi} dx  \int_0^{2 \pi} dx' \,  \Big ( {d \over dx} f(x) \Big )  \Big ( {d \over dx'} g(x') \Big ) \log |  \theta_1( x - x';\alpha^{1/2}) |.
\end{multline}

 \subsection{Bulk scaling limit of CUE and CUE${}_\alpha$ random matrices}
 The bulk scaled limit is a rescaling of the coordinates of the particles in the point process so that they
 have an order unity density (taken to be unity for convenience). For the CUE, which has $N$ eigenvalues
 with coordinates $x_j$ between $0$ and $2 \pi$, this is done by changing variables $x_j - \pi \mapsto 2 \pi X_j/N$.
 Applying the corresponding change of variables to (\ref{3.1}) with $f(2 \pi X/N + \pi) = F(X)$,
 $g(2 \pi X/N + \pi) = G(X)$ shows
 \begin{equation}\label{3.1Ba}
\lim_{N \to \infty} {\rm Cov}^{\rm CUE} \, \Big ( \sum_{l=1}^N F(X_l),  \sum_{l=1}^N G(X_l) \Big )   =
 \int_{-\infty}^\infty dx   \int_{-\infty}^\infty dx' \, F(x) G(x') \Big ( \rho_{(2),\infty}^T(X,X') +  \delta(X - X') \Big ), 
 \end{equation} 
 where, from (\ref{Pp5}),
  \begin{equation}\label{3.1Ca}
   \rho_{(2),\infty}^T(X,X') = - {\sin^2 \pi (X - X') \over ( \pi (X - X'))^2}.
 \end{equation} 
 The convergence of the double integral in (\ref{3.1Ba}) requires that both $F(X)$ and $G(X)$ be
 integrable at infinity.
 Noting that (\ref{3.1Ca}) is a function of the difference $X - X'$ allows the double integral in
 (\ref{3.1Ba}) to be reduced to a single integral involving Fourier transforms. 
 And associating with
 $F(X)$ and $G(X)$ a length scale $L$, an analogue of (\ref{3.4e}) can be deduced, as first
 made explicit by Dyson and Mehta \cite{DM63}.   
 
  \begin{proposition}\label{p3.1}
  Introduce the structure function (also referred to as the spectral form factor)
    \begin{equation}\label{3.1D}
    S_\infty^{\rm CUE}(k) := \int_{-\infty}^\infty  C_{(2), \infty}(X,0) e^{i k X} \, dX = 
    \begin{cases} \displaystyle {|k| \over 2 \pi}, & 0 < k < 2 \pi \\
    1, & |k| \ge 2 \pi, \end{cases}
    \end{equation}
    where the equality follows from the definition (\ref{C.1}) of $ C_{(2), \infty}$ and the functional form (\ref{3.1Ca}).
    We have
\begin{equation}\label{3.4dB}
\lim_{N \to \infty} {\rm Cov}^{\rm CUE} \, \Big ( \sum_{l=1}^N F(X_l),  \sum_{l=1}^N G(X_l) \Big )   =
{1 \over 2 \pi}
\int_{-\infty}^\infty \hat{F}(k) \hat{G}(-k)  S_\infty^{\rm CUE}(k) \, dk,
 \end{equation} 
 where  $ \hat{F}(k) $ denotes the Fourier transform of $F(X)$, 
  and similarly the meaning of $ \hat{G}(-k) $.
  Furthermore, replacing $F(X)$ by $F_L(X) = F(X/L)$, and similarly replacing
  $G(X)$, we have that
  \begin{equation}\label{3.4dB1}
\lim_{L \to \infty } \lim_{N \to \infty} {\rm Cov}^{\rm CUE} \, \Big ( \sum_{l=1}^N F_L(X_l),  \sum_{l=1}^N G_L(X_l) \Big )   =
{1 \over (2 \pi)^2} \int_{-\infty}^\infty \hat{F}(k) \hat{G}(-k)  |k| \, d k,
 \end{equation} 
 assuming $ \hat{F}(k) \hat{G}(-k)$ decays sufficiently fast for the integral to converge.
 \end{proposition} 
 
 \begin{proof}
 It remains to justify (\ref{3.4dB1}), starting from (\ref{3.4dB}) with $F,G$ replaced by $F_L, G_L$. From
 the definitions we have $\hat{F}_L(k) = L \hat{F}(Lk)$, and similarly $\hat{G}_L(k) = L \hat{G}(Lk)$.
 Changing variables $k \mapsto k/L$ and taking into consideration (\ref{3.1D}) shows that for large
 $L$ the RHS of (\ref{3.4dB}) reduces to the RHS of (\ref{3.4dB1}).
 \end{proof}
 
 It was noted in \cite{DM63} that choosing $F_L(X) = G_L(X) = \chi_{X \in [0,L]}$ gives
 $\hat{F}_L(k) = {1 \over i k} (1 - e^{i k L})$, which does not permit the passage from
 (\ref{3.4dB}) to (\ref{3.4dB1}). By considering the functional form of  (\ref{3.4dB}) in this
 case, it was shown instead that for large $L$ 
  \begin{equation}\label{3.4dB2}
 \lim_{N \to \infty} {\rm Cov}^{\rm CUE} \, \Big ( \sum_{l=1}^N \chi_{X_l \in [0,L]} \Big ) \mathop{\sim}\limits_{L \to \infty}
 {1 \over \pi^2} \log L + B_2,
  \end{equation} 
  where, with $C$ denoting Euler's constant,
 \begin{equation}\label{3.4dB3}    
 B_2 = { 1 \over \pi^2} C + {1 \over \pi^2} (1 + \log 2 \pi).
   \end{equation} 
   The essential step in obtaining the leading term is to consider the contribution to the integral
   (\ref{3.4dB}) in the range $|k| < 1$. Making use of (\ref{3.1D}) and the functional form of $\hat{F}_L(k)$ as noted above,
   then changing variables in the integrand $k L/2 \mapsto k$ gives for the leading large $L$ form
    \begin{equation}\label{3.4dB4} 
   {2 \over \pi^2} \int_0^{L/2} {\sin^2 k \over k} \, dk \sim {1 \over \pi^2} \log L,
   \end{equation}  
   in agreement with (\ref{3.4dB2}).
   
   \begin{remark}
  1.~An alternative functional form to the RHS of (\ref{3.4dB1}) is \cite{DM63}
    \begin{equation}\label{3.4dB4z}  
    - {1 \over 2 \pi^2}  \int_{-\infty}^\infty dX \int_{-\infty}^\infty dY \, F'(X) G'(Y) \log | X - Y|,
    \end{equation} 
    which at a formal level follows from the generalised Fourier transform
   \begin{equation}\label{3.4dB4A}     
    - \int_{-\infty}^\infty 
\log( |x| ) e^{i k x} \, dx = {\pi \over |k|}.
  \end{equation} 
  Furthermore, noting  ${\partial^2 \over \partial X \partial Y} \log | X - Y| = 1/(X-Y)^2$, integration by parts
  of the expression
   \begin{equation}\label{3.4dB4+}  
  {1 \over 4 \pi^2}  \int_{-\infty}^\infty dX \int_{-\infty}^\infty dY \, {(F(X) - F(Y)) (G(X) - G(Y)) \over ( X - Y)^2 } 
    \end{equation} 
    shows it is equal to (\ref{3.4dB4z}); cf.~the equality in (\ref{Pp6}).
  
  2.~In the case $F_L = G_L$ (\ref{3.4dB1}) specifies the large $L$ form of the bulk scaled variance of
 the linear statistic $\sum_{l=1}^N F_L(X_l)$. It was established by Soshnikov \cite{So00} that
  the distribution of the centred linear statistic $\sum_{l=1}^N F_L(X_l) - L \int_{-\infty}^\infty F(X) \, dX$
  converges to a zero mean Gaussian with this bulk scaled variance; see also \cite{BK99}.
 \end{remark}
 
 In the case of the bulk scaling of Gaudin's deformation of the CUE as specified by the
 probability density function
 (PDF) (\ref{q3}), taking the place of (\ref{3.1Ca})  is the functional form \cite{Ga66}
   \begin{multline}\label{q3n}
   \lim_{N \to \infty} \Big ( {N \over 2 \pi} \Big )^2 \rho_{(2),N}^T(2 \pi X/N, 2 \pi X'/N;\alpha) \Big |_{\alpha = e^{-2 \pi a / N}} : =
 \rho_{(2),\infty}^T(X, X';a)    \\ = - \Big ( {1 \over 2 \pi a} \Big )^2
 \bigg | \int_{\nu}^\infty {e^{-i \omega (X - X')/a} \over e^\omega + 1} \bigg |^2,
 \end{multline}
 where $\nu = - \log (e^{2 \pi a} - 1)$. Changing variables $\omega/a = \omega'$, we note that for $a \to \infty$ only the integration region from
 $\omega' = - 2\pi $ to $0$ contributes, and (\ref{3.1Ca}) is reclaimed.
 
 It was shown in \cite{Ga66} that the Fourier transform of (\ref{q3n}) can be computed explicity.
 Forming from this the structure function (recall the definition in (\ref{3.1D})) gives
   \begin{equation}\label{sc1}
   S_\infty(k;a) = 1 - {1 \over 2} {1 \over \sinh(|k| a/2)} \bigg (
   {e^{|k| a / 2} \over 2 \pi a} \log \Big ( 1 + e^{- |k| a} (e^{2 \pi a} - 1)\Big ) - e^{-|k| a/2} \bigg ).
  \end{equation} 
  Most significant in the context of fluctuation formulas with a scale parameter $L$ is that the
  small $k$ limit results in a constant   \cite{Ga66}
   \begin{equation}\label{sc2}
   S_\infty(k;a)  \mathop{\sim} \limits_{k \to 0} {1 \over 2 \pi a} (1 - e^{-2 \pi a} ),
    \end{equation} 
    in distinction  to the behaviour of the functional form in (\ref{3.1D}). Hence, the analogue of   (\ref{3.4dB1}) is now
   \begin{equation}\label{3.4dBX}
 \lim_{N \to \infty} {\rm Cov}^{{\rm CUE{}}_\alpha} \, \Big ( \sum_{l=1}^N F_L(X_l),  \sum_{l=1}^N G_L(X_l) \Big )  \bigg |_{\alpha = e^{-2 \pi a/N}} 
 \mathop{\sim} \limits_{L \to \infty} 
{L \over 2 \pi }   S_\infty(0;a) \int_{-\infty}^\infty \hat{F}(k) \hat{G}(-k) \, d k,
 \end{equation} 
 telling us in particular that the fluctuations are now of order $L$. This contrasts with the order unity fluctuation for the bulk scaled CUE seen
 in  (\ref{3.4dB1}), and the order $\log L$ of (\ref{3.4dB2}). It contrasts too with the order unity fluctuation for the global scaled limit
 fluctuation of the CUE${}_\alpha$ seen in (\ref{4.4e1}). With regards to this latter point, note that setting $\alpha = e^{-2 \pi a / N}$
 for $N$ large,
 as done in the bulk scaling limit (\ref{q3n}), would lead to all the coefficients in the formula of (\ref{4.4e1}) diverging.
 We observe too that taking $a \to \infty$ in (\ref{sc2}), when we know (\ref{q3n}) limits to the bulk scaled CUE result, gives
 that $ S_\infty(k;a) $ tends to zero as $k$ tends to zero, and so (\ref{3.4dBX}) breaks down.
 
 \begin{remark}
 1.~A behaviour analogous to (\ref{3.4dBX}) can be obtained in the context of the bulk scaled CUE, modified so that a fraction
 $(1 - \zeta)$, $0 < \zeta < 1$ of the eigenvalues have been deleted uniformly at random. We will denote this ensemble by
 $\widehat{\rm CUE}_\zeta$. Such a model was first considered in detail
 in \cite{BP04}; for applications to Odlyzko's data set for the Riemann zeros see
 \cite{FM15,BFM17}. Generally in this setting the density is multiplied by $\zeta$ and the two-point function by $\zeta^2$. Hence 
 in place of (\ref{3.1Ba}) we have
  \begin{multline}\label{3.1Bx}
\lim_{N \to \infty} {\rm Cov}^{\widehat{\rm CUE}_\zeta} \, \Big ( \sum_{l=1}^N F(X_l),  \sum_{l=1}^N G(X_l) \Big )  \\ =
 \int_{-\infty}^\infty dx   \int_{-\infty}^\infty dx' \, F(x) G(x') \Big ( \zeta^2 \rho_{(2),\infty}^T(X,X') +  \zeta \delta(X - X') \Big ).
 \end{multline} 
 Furthermore
 $$
 S_\infty^{\widehat{\rm CUE}_\zeta}(k) = \zeta^2 S_\infty^{\rm CUE}(k) + (\zeta - \zeta^2),
 $$
 telling us in particular that $ S_\infty^{\widehat{\rm CUE}_\zeta}(0) = \zeta - \zeta^2 \ne 0$. Thus the analogue of
 (\ref{3.4dBX}) holds in this setting. A detailed  analysis is given in \cite{BD17}, which includes 
 consideration of the critical setting specified by
 $\zeta = 1 - c/L$, $(c >0)$, for which the analogue of the RHS of  (\ref{3.4dBX}) is again of order unity but the corresponding
 distribution is no longer Gaussian.  \\
 2.~The structure function taking a non-zero value at $k=0$ is in (\ref{sc2}) is also a feature of the statistics of the real eigenvalues
 in the ensemble of $N \times N$ real Gaussian matrices. For such matrices the eigenvalues occur in complex conjugate pairs
 and moreover there are $O(\sqrt{N})$ real eigenvalues which to leading order have uniform density in the interval
 $(-\sqrt{N}, \sqrt{N})$ \cite{EKS94}. Results from \cite{FN07}, \cite{BS09} give that in the bulk scaled limit the truncated
 two-point correlation has the explicit form
  \begin{equation}\label{4.4dS}
  \rho_{(2), \infty}^T(x,x') = - {1 \over 2 \pi} e^{-(x - x')^2} + {1 \over 2} {1 \over \sqrt{2 \pi}} | x - x'| e^{-(x - x')^2/2}
  {\rm erfc} ( | x - x'|/\sqrt{2}).
   \end{equation} 
   Defining the structure function as in (\ref{3.1D}),  a computer algebra assisted calculation gives
   $S_\infty(0) = (\sqrt{2} - 1)/\sqrt{\pi}$. Hence the covariance of the scaled linear statistics $F_L, G_L$ is
   proportional to $L$ as specified by  (\ref{3.4dBX}).
 \end{remark}
 
 \subsection{Dyson Brownian motion for the CUE}
 Haar distributed matrices $U \in U(N)$ admit a generalisation involving a parameter $\tau$, relating to the heat
 equation on the corresponding symmetric space; see e.g.~\cite[\S 2]{LM10}. We will refer to this  ensemble
 as $U(N;\tau)$. As first determined by 
 Dyson \cite{Dy62b}, the eigenvalues $\{ e^{i \theta_j(\tau)} \}_{j=1}^N$ then execute a particular Brownian dynamics with
 corresponding PDF obeying the Fokker-Planck equation
  \begin{equation}\label{R3}  
  {\partial p_t \over \partial \tau} = \mathcal L p_t, \qquad   \mathcal L = \sum_{j=1}^N {\partial \over \partial \theta_j} \bigg ( {\partial W \over \partial \theta_j} +
{1 \over \beta}  {\partial \over \partial \theta_j} \bigg ),
\end{equation} 
with $\beta = 2$, $W = - \sum_{1 \le j < k \le N} \log |  e^{i \theta_k} -   e^{i \theta_j} |$ and subject to a prescribed initial condition.
Moreover, as a consequence of an observation of Sutherland \cite{Su71a} that provides a similarity transformation of $\mathcal L$
with $\beta = 2$ to a free quantum Hamiltonian, (\ref{R3}) can be exactly solved in a determinant form \cite{PS91}. And moreover, there are a number
of initial conditions for which the corresponding dynamical correlations can be written in a structured form.

One such initial condition is the equilibrium solution of (\ref{R3}),  $p_0 = p_t |_{t \to \infty} \propto e^{-\beta W}$. In particular, the truncated two-point
correlation for two different parameters, $\rho_{(1,1)}^T((x,\tau_x), (y,\tau_y))$ say, has the explicit functional form 
  \begin{equation}\label{R3a}  
\rho_{(1,1)}^T((x,\tau_x), (y,\tau_y)) = \Big ( {1 \over 2 \pi} \Big )^2 \bigg (
\sum_{|n| \le N/2} \Big ( {w \over z} \Big )^n e^{\gamma_n(\tau_y - \tau_x)} \bigg )
\bigg (
\sum_{|n| \ge N/2 + 1} \Big ( {z \over w} \Big )^n e^{-\gamma_n(\tau_y - \tau_x)} \bigg ).
\end{equation} 
Here $N$ is assumed even for convenience, $w = e^{i y}$, $z = e^{i z}$ and $\gamma_n = (n - (1/2)^2)/2$. In the setting of two distinct
parameters, replacing (\ref{3.1}) is the covariance formula
  \begin{equation}\label{3.1B}
  \Big \langle \sum_{j=1}^N f(x_j),  \sum_{j=1}^N g(y_j) \Big \rangle = \int_0^{2 \pi} dx  \, f(x)  \int_0^{2 \pi} dy \, g(y) \rho_{(1,1)}^T((x,\tau_x), (y,\tau_y)).
 \end{equation} 
 According to (\ref{R3a}) we have for the limiting Fourier series form of the truncated two-point
correlation 
  \begin{equation}\label{R3b}  
  \lim_{N \to \infty} \rho_{(1,1)}^T((x,\tau_x), (x,\tau_y))  \Big |_{\tau_y - \tau_x = t/N} =  \Big ( {1 \over 2 \pi} \Big )^2  \sum_{n = - \infty}^\infty
  \Big ( {z \over w} \Big )^n  \sum_{q=0}^{|n| - 1} e^{- ( |n| - 2q) t} 
 \end{equation} 
 (cf.~the $N \to \infty$ form of (\ref{3.4c})), and consequently
 \begin{equation}\label{R3c}  
  \lim_{N \to \infty} \bigg \langle \sum_{j=1}^N f(x_j),  \sum_{j=1}^N g(x_j) \bigg \rangle^{U(N;\tau)}   \Big |_{\tau_y - \tau_x = t/N}  = \sum_{n = - \infty}^\infty \Big (
   \sum_{q=0}^{|n| - 1} e^{- ( |n| - 2q) t} \Big ) f_n g_{-n},
  \end{equation}  
 which we see reduces to (\ref{3.4e}) for $t = 0$.
 
 Also of interest is the bulk scaling limit. In relation to the truncated two-point correlation it follows from (\ref{R3a}) that \cite{NF03}
 \begin{multline}\label{R3d}  
 \rho_{(1,1)}^{T, {\rm bulk}}((X,0), (Y,t)) := \lim_{N \to \infty} (2 \pi / N)^2  \rho_{(1,1)}^{T, {\rm bulk}}((2 \pi X/N,0), (2 \pi Y/N,4 \pi^2 t/N^2))  \\
  = \Big ( \int_0^1 e^{t (\pi u)^2/2} \cos \pi (Y - X)u \, du \Big )  \Big ( \int_1^\infty e^{-t (\pi u)^2/2} \cos \pi (Y - X)u \, du \Big ).
 \end{multline} 
 Defining the Fourier transform
 $$
 S(k;t) := \int_{-\infty}^\infty  \rho_{(1,1)}^{T, {\rm bulk}}((X,0), (Y,t))  e^{i (X - Y) k} \, dt
 $$
 we can calculate from (\ref{R3d}) that for small $|k|$ \cite{Sp87}, \cite[Eq.~(13.228)]{Fo10}
 \begin{equation}\label{R3e}  
  S(k;t)  \sim {|k| \over  2 \pi} e^{- \pi |k| t}
   \end{equation} 
   and hence for the parameter dependent extension of (\ref{3.4dB1}) we obtain
    \begin{equation}\label{3.4dB1P}
\lim_{L \to \infty } \lim_{N \to \infty} {\rm Cov}^{U(N;\tau)}  \, \Big ( \sum_{l=1}^N F_L(X_l),  \sum_{l=1}^N G_L(X_l) \Big )  \Big |_{t = L T}  =
{1 \over (2 \pi)^2} \int_{-\infty}^\infty \hat{F}(k) \hat{G}(-k)  |k|  e^{- \pi |k| T} \, d k.
 \end{equation}

 \subsection{Global scaling of Haar distributed real orthogonal random matrices}
 Closely related to the CUE is the ensemble of real orthogonal matrices --- further distinguished by the determinant equalling
 plus $1$ or minus $1$, and the parity of $N$ --- chosen with Haar measure. Choosing $N$ even and for the determinant
 to equal plus one for definiteness, for this ensemble we have (see e.g.~\cite[\S 7.2.7]{Fo10})
 \begin{multline}\label{4.4cO}
  \rho_{(2),N}^T(x,x') + \rho_{(1),N}(x) \delta(x - x') \\ = - \Big ( {1 \over \pi} +
  {2 \over \pi} \sum_{l=1}^{N/2} \cos l x \cos l x' \Big )^2 + 
  \delta (x-x')
 \Big ( {1 \over \pi} +
  {2 \over \pi} \sum_{l=1}^{N/2} \cos^2 l x  \Big ).
 \end{multline}  
 Direct calculation reveals that for $p,q$ non-negative integers less than or equal to $N/2$,
  \begin{equation}\label{4.4dO}
  \int_0^\pi dx   \int_0^\pi dx' \, \Big (   \rho_{(2),N}^T(x,x') + \rho_{(1),N}(x) \delta(x - x') \Big ) \cos p x  \cos q x'  = {p \over 4} \delta_{p,q}.
  \end{equation}
  Outside this range, the integral always vanishes if $p,q$ are of distinct parity; it gives the value $1/4$ for $p,q$ of
  the same parity and not equal, and the value ${\rm min} \, \{(p+1)/4, (N+2)/4 \}$ for  $p=q$.
 In keeping with the derivation of (\ref{3.4e}), this implies a simple expression for the limiting covariance.
 
  \begin{proposition}\label{p4.1}
  Let
  \begin{equation}\label{2.55}
  f(x) = f_0^{\rm c} + 2 \sum_{n=1}^\infty  f_n^{\rm c}  \cos(n x), \quad f_n^{\rm c} = {1 \over \pi} \int_0^\pi f(x) \cos(n x) \, dx,
  \end{equation} 
  and similarly for $g(x)$.  
  If $f$ and $g$ are differentiable on $[0,\pi]$ with $f', g'$ H\"older continuous
 of order $\alpha > 0$, then for $N$ even
\begin{equation}\label{4.4eO}
\lim_{N \to \infty} {\rm Cov}^{O^+(N/2)} \, \Big ( \sum_{l=1}^{N/2} f(x_l),  \sum_{l=1}^{N/2} g(x_l) \Big )   =
\sum_{n=1}^\infty    n    f_n^{\rm c}  g_{n}^{\rm c}.
 \end{equation} 
Furthermore, if $f = g = \chi_{[L_0, L_1]}$ ($0 < L_0 < L_1 <  \pi$) and $N_{[L_0, L_1]} = \sum_l  \chi_{x_l \in [L_0, L_1]}$, then
\begin{equation}\label{4.4fO}
\lim_{N \to \infty}  {1 \over \log N} {\rm Var}^{\rm O^+(N/2)} \,  ( N_{[L_0, L_1]}  ) =  {1 \over \pi^2}.
 \end{equation} 
 \end{proposition} 
 
 \begin{proof}
 The reasoning relating to (\ref{4.4eO}) has already been given.
  In relation to (\ref{4.4fO}), with $f = g = \chi_{[L_0,L_1]}$, $0 < L_0 < L_1 <  \pi$,
  we compute that for $n \ne 0$, $f_n^{\rm c} =  (\sin (L_1 n) - \sin(L_0 n)) /(\pi n)$. Substituting in
  the RHS of (\ref{4.4eO}), with the sum truncated at $n = O(N)$ as is justified by
  (\ref{4.4dO}) and surrounding text, gives (\ref{4.4f}).
  \end{proof}
 
 \begin{remark}
 1.~The centred characteristic function associated with the linear statistic $\sum_{l=1}^{N/2} f(x_l)$,
 $0 < x_l < \pi$, for the ensemble of random real orthogonal matrices $O^+(N/2)$ has the limiting
 Gaussian form
 \begin{equation}\label{4.4gO}
 \lim_{N \to \infty} \Big \langle \prod_{l=1}^{N/2} e^{i t ( f(x_l) - f_0^{\rm c})} \Big \rangle^{O^+(N/2)} =
 \exp \bigg ( - {t^2 \over 2} \sum_{n=1}^\infty n (f_n^{\rm c})^2 \bigg )
 \end{equation} 
 (cf.~(\ref{Pp4})). This was first established by Johansson \cite{Jo97} for $f$ polynomial;
 \cite{CGMY21} extends the validity to $f = e^{V(e^{ix})} = e^{V(e^{-ix})} $ for $V$ analytic 
 in a neighbourhood of the unit circle. \\
 2.~For $R \in {\rm SO}^+(N/2)$ we have ${\rm Tr} \, R^k = \sum_{p=1}^{N/2} \cos k p x$.
 This linear statistic has a property analogous to that of $| {\rm Tr} \, U^k |^2$ for $U \in {\rm CUE}$
 noted in Remark \ref{R1}.2. Thus the first $N/2$ moments coincide with those of $\sqrt{k}$ times
 a standard real Gaussian random variable \cite{DE01}, a fact which is related to the RHS of (\ref{4.4dO}) being 
 independent of $N$ for all non-negative integers $p,q \le N/2$.
 \end{remark}
 
 \subsection{The COE and CSE}
 Starting with matrices $U_N \in U(N)$ chosen with Haar measure, then forming symmetric unitary matrices $U_N^T U_N$ gives
 Dyson's circular orthogonal ensemble (COE). 
 A variation is to begin with matrices $U_{2N} \in U(2N)$ chosen with Haar measure.
 Forming   the 
 self dual quaternion unitary matrices   $Z_{2N}^{-1} U_{2N}^T Z_{2N} U_{2N}$,
 where $Z_{2N} = \mathbb I_{N} \otimes {\small \begin{bmatrix} 0 & 1 \\ -1 & 0 \end{bmatrix}}$,
gives Dyson's circular 
 symplectic ensemble (CSE).
 
 Define
 \begin{equation}\label{4.4hC}
 S_N(\theta) = {1 \over 2 \pi} {\sin(N \theta/2) \over \sin (\theta/2)}, \: D_N(\theta) = {d \over d \theta} S_N(\theta), \:
 I_N(\theta) = \int_0^\theta S_N(\theta') \, d \theta', \: J_N(\theta) = I_N(\theta) - {1 \over 2} {\rm sgn} (\theta).
 \end{equation} 
 In terms of these quantities the corresponding two-point correlation functions read \cite{Dy70,TW98}
 \begin{align}
 \rho_{(2),N}^{\rm COE}(\theta, \theta') & =  \Big ( (S_N(\theta - \theta'))^2 - D_{N}(\theta - \theta') J_{N}(\theta - \theta') \Big ), \label{2.54a}\\
 \rho_{(2),N}^{\rm CSE}(\theta, \theta')  & = {1 \over 4} \Big ( (S_{2N}(\theta - \theta'))^2 - D_{2N}(\theta - \theta') I_{2N}(\theta - \theta') \Big ).  \label{2.54b}
 \end{align}
 Starting from these expressions and defining the Fourier coefficients $m_l^{\rm COE}$ and $m_l^{\rm CSE}$ as in
 (\ref{3.4c}), we know from \cite{WF15} that 
  \begin{align}\label{2.56}
  m_l^{\rm COE} & =
  N - (N - |l|) \chi_{N - |l|>0} + {\rm min}(|l|,N) - 2l \bigg ( \sum_{s=M_-}^{M_+} {1 \over 2 s - 1}  \bigg ) \\
    m_l^{\rm CSE} & = \begin{cases} \displaystyle {| l | \over 2} + {|l| \over 2} \bigg ( {1 \over 2 N - 1} + {1 \over 2N - 3} +
    \cdots + {1 \over 2 N - (2 |l| - 1)} \bigg ), & |l| \le 2N - 2 \\
    N, & |l| > 2N - 2, \end{cases}
    \end{align}
where $M_-:= {1 \over 2} (N+1) + {\rm max}\, (0,|l|-N)+1$, $M_+:=  {1 \over 2} (N+1) + |l|$. For $l > 0$ the quantity
$ m_l^{\rm COE} $ monotonically increases to the value $N$, while $m_l^{\rm CSE} $ has a single maximum at
$l=N$, which to leading order in $N$ is equal to $(N/4) \log N $.

From the exact results (\ref{2.56}) we can deduce the analogues of (\ref{3.4e}) and (\ref{3.4f}).

\begin{proposition}\label{p2.10}
Label the COE and CSE by $\beta = 1$ and $\beta = 4$ respectively. We have
\begin{equation}\label{3.4eb}
\lim_{N \to \infty} {\rm Cov}^{(\beta)} \, \Big ( \sum_{l=1}^N f(x_l),  \sum_{l=1}^N g(x_l) \Big )   =
{2 \over \beta} \sum_{l=-\infty}^\infty    | l |    f_l g_{-l},
 \end{equation} 
while if $f = g =   \chi_{[0,L]}$ ($0 < L < 2 \pi$) and $N_L$ is specified by (\ref{1}) then
\begin{equation}\label{3.4fb}
\lim_{N \to \infty}  {1 \over \log N} {\rm Var}^{ (\beta)} \,  ( N_L  ) = {1 \over \beta \pi^2}.
 \end{equation} 
 \end{proposition} 
 
 \begin{proof}
 We consider (\ref{3.4eb}) only; the working required in relation to (\ref{3.4fb}) uses the same arguments with
 $f_l = g_l$ the explicit functional form noted in the proof of Proposition \ref{p2.1z}. We see from the above results
 that the Fourier coefficients have the functional form
 \begin{equation}\label{3.4fc}
 m_l^{(\beta)} = {2 l \over \beta} + r^{(\beta)}_{l,N}, \quad |l| \le N,
  \end{equation} 
  where $\lim_{N \to \infty} \sum_{l=1}^N f_l g_{-l}  r^{(\beta)}_{l,N} = 0$ provided $f_l g_{-l} = O(1/l^{2 + \epsilon})$,
  $\epsilon > 0$. We see too that outside this range, and with the same assumed decay of $f_l g_{-l}$, we have
  $\lim_{N \to \infty} \sum_{l=N+1}^\infty f_l g_{-l}  m_l^{(\beta)} = 0$. The limit formula  (\ref{3.4eb}) now follows.
  \end{proof}
  
  The bulk scaled limit is also of interest. For this we introduce the appropriate bulk scaling of the quantities
  (\ref{4.4hC}),
   \begin{equation}\label{4.4hB}
 S_\infty(X) =  {\sin \pi X \over \pi X}, \: D_\infty(X) = {d \over d X} S_\infty(X), \:
 I_\infty(X) = \int_0^X S_\infty(X') \, d X', \: J_\infty(X) = I_\infty(X) - {1 \over 2} {\rm sgn} (X).
 \end{equation} 
 We then see from (\ref{2.54a}), (\ref{2.54b})
  \begin{multline}\label{2.55a}
\lim_{N \to \infty} \Big ( {2 \pi \over N} \Big )^2 \rho_{(2),N}^{\rm COE}(2 \pi X/N, 2 \pi X' / N)  = :
\rho_{(2),\infty}^{\rm COE}(X,X') \\ = 
  (S_\infty(X -  X'))^2 - D_{\infty}(X - X') J_{\infty}(X - X') , 
 \end{multline}
  \begin{multline}\label{2.55b}
 \lim_{N \to \infty} \Big ( {2 \pi \over N} \Big )^2 \rho_{(2),N}^{\rm CSE}(2 \pi X/N, 2 \pi X' / N)  =: 
\rho_{(2),\infty}^{\rm CSE}(X,X')  \\
= 
  (S_\infty(2(X -  X')))^2 - D_{\infty}(2(X - X')) {1 \over 2} I_{\infty}(2(X - X')) ,
 \end{multline}
 as first deduced in the work of Dyson \cite{Dy62a} and Mehta--Dyson \cite{MD63} respectively,
 although the functional form (\ref{2.55b}) is also contained in \cite{Dy62a}. It results
 there from the computation of the bulk scaled two-point correlation function of every second
 eigenvalue in the COE; see \cite[\S 4.2.3]{Fo10} for more on the implied
 inter-relationship. The work \cite{Dy62a} also contains the computation of the Fourier transform
 of the corresponding truncated two-point correlations, which gives for the corresponding
 structure functions
   \begin{align}
   S^{\rm COE}_\infty(k)  & =
\left \{
\begin{array}{ll} {|k| \over  \pi} - {|k| \over 2 \pi} \log  \Big ( 1  + {|k|
\over  \pi} \Big ), &
  |k| \le 2 \pi, \\[.2cm]
2 - {|k| \over 2 \pi} \log {|k|/\pi + 1 \over |k|/\pi - 1}, &
|k| \ge 2 \pi. \end{array} \right. \label{2.56a}\\
S^{\rm CSE}_\infty(k)  & =
\left \{
\begin{array}{ll} {|k| \over 4 \pi} - {|k| \over 8 \pi} \log | 1 - {|k|
\over 2 \pi}|, &
  |k| \le 4 \pi, \\[.2cm]
1, &
|k| \ge 4 \pi. \end{array} \right.   \label{2.56b}
\end{align}
Consequently, as made explicit in \cite{DM63} in the case of the bulk scaled COE,
the leading small $|k|$ term from these functional forms implies that
the limiting variance formula
(\ref{3.4dB1}) as obtained for the bulk scaled CUE is modified only by a simple proportionality,
 \begin{equation}\label{3.4dC1}
\lim_{L \to \infty } \lim_{N \to \infty} {\rm Cov}^{(\beta)} \, \Big ( \sum_{l=1}^N F_L(X_l),  \sum_{l=1}^N G_L(X_l) \Big )   = {2 \over \beta}
{1 \over (2 \pi)^2} \int_{-\infty}^\infty \hat{F}(k) \hat{G}(-k)  |k| \, d k,
 \end{equation} 
 where the meaning of $\beta$ is as in Proposition \ref{p2.10}.
 
 In relation to the bulk scaled linear statistic $\sum_{l=1}^N \chi_{X_l \in [0,L]}$, following the strategy outlined in \cite{Dy62a}
 it has been noted in \cite[\S 14.5.1]{Fo10} that for the COE and CSE (and too the CUE upon identifying $\beta = 2$)
  \begin{equation}\label{3.4dC2}
 \lim_{N \to \infty}  {\rm Var} \Big ( \sum_{l=1}^N \chi_{X_l \in [0,L]} \Big ) \sim {2 \over \pi^2 \beta} \log L + B_\beta,
   \end{equation} 
   where, with $C$ denoting Euler's constant,
  \begin{equation}\label{3.4dC3} 
  B_\beta = {2 \over \pi^2 \beta} C + {2 \over \pi} \int_0^1 {1 \over y^2} \Big ( S_\infty^{(\beta)}(y) - {y \over \pi \beta}  \Big ) \, dy +
  {2 \over \pi} \int_1^\infty {1 \over y^2} S_\infty^{(\beta)}(y) \, dy.
    \end{equation} 
    Substituting (\ref{2.56a}) gives \cite{Dy62a}
    \begin{equation}\label{3.4dC4}   
B_1 = {2 \over \pi^2} C + {2 \over \pi^2} \Big ( 1 + \log 2 \pi \Big ) - {1 \over 4},
\end{equation} 
while substituting (\ref{2.56a}) leads to the formula \cite[Eq.~(16.1.4)]{Me04}
  \begin{equation}\label{3.4dC5}      
B_4 = {1 \over 2 \pi^2} C + {1 \over 2 \pi^2} \Big ( 1 + \log 4 \pi \Big ) + {1 \over 16}.
\end{equation}  

\begin{remark}
A recent result \cite{AGL21} gives that for a one-dimensional point process, in the limit $L \to \infty$
 \begin{equation}\label{3.4dCX}
 \lim_{N \to \infty} {\rm Var} \, \Big ( \sum_{l=1}^N  \chi_{X_l \in [0,L]} \Big )  \asymp \bigg ( L^2 \int_{|x| < c/L} S_\infty(x) \, dx +
 \int_{|x| > c/L} {S_\infty(x) \over x^2} \, dx   \bigg ),
 \end{equation}
 which is consistent with (\ref{3.4dC2}) and furthermore gives some insight into the structure of (\ref{3.4dC3}). In fact
 \cite{AGL21} gives an analogous asymptotic bound in the $d$-dimensional case for $( \sum_{l=1}^N  \chi_{| \mathbf r_l | < L} $.
 \end{remark}
  
  \subsection{Bulk scaling of the circular $\beta$-ensemble}\label{S2.7}
 Unitary random matrices with eigenvalue PDF proportional to 
  \begin{equation}\label{z.1}
  \prod_{1 \le j <k \le N} | e^{i \theta_k} -  e^{i \theta_j} |^\beta
\end{equation}  
are said to form the circular $\beta$-ensemble. The cases $\beta = 1, 2$ and 4
are realised by the COE, CUE and CSE respectively. For general $\beta > 0$ there is
a realisation in terms of certain unitary Hessenberg random  matrices \cite{KN04}.  
As emphasised by Dyson \cite{Dy62}, upon writing (\ref{z.1}) in the form
 \begin{equation}\label{z.1a}
e^{- \beta \sum_{1 \le j < k \le N} \phi( e^{i \theta_j}, e^{i \theta_k})}, \qquad \phi(z_j, z_k) = - \log |z_k - z_j|,
\end{equation}  
there is analogy with the equilibrium statistical mechanics of particles repelling pairwise via the
logarithmic potential and confined to a circle. The interpretation of $\beta$ is then as 
 the inverse temperature. 

For $\beta = p/q$ a positive rational number in reduced form, the bulk scaled
structure function $S_\infty(k;\beta)$ is known explicitly \cite{FJ97}. This  functional
form shows that the quantity
 \begin{equation}\label{z.2}
F(k;\beta) = {\pi \beta \over k} S(k;\beta)
\end{equation}  
in the range $0 < k < {\rm min} (2 \pi, \pi \beta)$ extends to an analytic function of $k$
about the origin with radius of convergence ${\rm min} (2 \pi, \pi \beta)$. The leading
terms of the corresponding power series in $k$ --- in this the coefficient of $k^j$ is
a polynomial of degree $j$ in $(2/\beta)$ which has particular palindromic properties ---
are known up to an including $j=10$ \cite{FJM00,Fo21a}, with the first two being
 \begin{equation}\label{z.3}
 f(k;\beta) = 1 + {1 \over 2 \pi} (1 - 2/\beta)k + \cdots
 \end{equation}  
Hence
 \begin{equation}\label{z.4}
 S(k;\beta) = {|k| \over \pi \beta}  + {1 \over 2 \pi^2 \beta} (1 - 2/\beta)k^2 + \cdots
  \end{equation}  
  Since the derivation of (\ref{3.4dC1}) is determined entirely by the leading term in this
  expansion, we see that this same expression, derived previously for $\beta = 1,2$ and $4$,
  holds for all $\beta > 0$.
  
  The derivation of (\ref{3.4dB2}) for $\beta = 2$ and (\ref{3.4dC2}) in the cases $\beta = 1, 4$,
  when used in conjunction with the knowledge (\ref{z.4}) affirms the formula (\ref{3.4dC2}) as
  valid for general $\beta > 0$. Moreover the recent work \cite{FL20}, using a $\beta$-generalisation
  of the Fisher-Hartwig conjecture from the theory of Toeplitz determinants \cite{FF04} has
  computed for the constant $B_\beta$ in (\ref{3.4dC2})
   \begin{equation}\label{z.5}
   B_\beta = {2 \over \pi^2 \beta} \bigg ( C + \log \beta + \sum_{q=1}^\infty \Big ( {2 \over \beta} \psi^{(1)}(2q/\beta) - {1 \over q} \Big ) \bigg ),
    \end{equation}  
  where $\psi^{(1)}(z) = {d^2 \over d z^2} \log \Gamma(z)$. The work \cite{SLMS21a} shows that
  the constant $B_\beta$ also occurs in the next order term of the global scaling of
  Var$\,(N_L)$ for $0< L < 2 \pi$, with the leading term being proportional to $\log N$.

  \subsection{Two-dimensional support}
  In the mid 1960's Ginibre introduced into random matrix theory the study of the
  eigenvalue statistics of, among other ensembles, $N \times N$ standard complex
  Gaussian random matrices \cite{Gi65}. For this ensemble, to be denoted GinUE, all
  the eigenvalues are in the complex plane. It was shown in \cite{Gi65} that the
  statistical state of the eigenvalues forms a determinantal point process, with
  the $N \to \infty$ bulk correlation kernel
   \begin{equation}\label{5.1k}
   K_\infty^{\rm GinUE}(w,z) = {1 \over \pi} e^{-(|w|^2 + |z|^2)/2} e^{w \overline{z}}
   \end{equation}  
   (in distinction to the derivation of (\ref{3.1Ca}), no scaling of the eigenvalues is
   required as part of the limit). Consequently the corresponding two-point correlation function
   has the simple functional form
    \begin{equation}\label{5.1a}
    \rho_{(2), \infty}^{\rm GinUE}(w,z) = {1 \over \pi^2} \Big ( 1 - \exp(-|w-z|^2 ) \Big ).
    \end{equation}
    This in turn implies that up to a constant the structure function, defined by the
    two-dimensional analogue of the integral in (\ref{3.1D}), is also a Gaussian
    \begin{equation}\label{5.1az}
    S_\infty^{\rm GinUE}(\mathbf k) =  1 - e^{-| \mathbf k |^2/ 4 \pi } .
    \end{equation} 
    
    The analogue of (\ref{3.4dB1}) can now readily be deduced.    
    \begin{proposition}
   Let $z_l = x_l + i y_l$. We have \cite[Exercises 15.4]{Fo10}
    \begin{equation}\label{5.4d}
\lim_{L \to \infty } \lim_{N \to \infty} {\rm Cov}^{\rm GinUE} \, \Big ( \sum_{l=1}^N f(x_l/L, y_l/L),  \sum_{l=1}^N g(x_l/L,y_l/L) \Big )   =
{1 \over (2 \pi)^2}  {1 \over 4 \pi}  \int_{\mathbb R^2}  \hat{f}(\mathbf k) \hat{g}(- \mathbf k)  |\mathbf k|^2 \, d \mathbf k,
 \end{equation} 
 valid provided the integral converges. 
     \end{proposition}
    
    \begin{proof}
    Starting with the two-dimensional analogue of
 (\ref{3.4dB}), 
    the essential point in the derivation of (\ref{5.4d})   is the small $| \mathbf k |$ form of $ S_\infty^{\rm GinUE}(\mathbf k) $. Thus, we read
 off from (\ref{5.1az}) that
 \begin{equation}\label{5.4f} 
  S_\infty^{\rm GinUE}(\mathbf k)  \sim { | \mathbf k |^2 \over  4 \pi }.
\end{equation} 

   \end{proof}

   As is consistent with (\ref{C.1}) set
  \begin{equation}\label{5.4ii} 
 C_\infty^{\rm GinUE}(\mathbf r - 
 \mathbf r') := \rho_{(2),\infty}^{\rm GinUE}(\mathbf r - 
 \mathbf r', \mathbf 0) - {1 \over \pi^2} + {1 \over \pi} \delta ( \mathbf r - 
 \mathbf r'),
  \end{equation}
 so that
 \begin{equation}\label{5.4j} 
 \int_{\mathbb R^2}  C_\infty^{\rm GinUE}(\mathbf r) e^{i \mathbf k \cdot \mathbf r} \, d \mathbf r = S_\infty^{\rm GinUE}(\mathbf k).
 \end{equation}
   For a general region $\Lambda \in \mathbb C$, the two-dimensional analogue of (\ref{3.4dB}) then gives
    \begin{equation}\label{5.4i} 
  \lim_{N \to \infty} {\rm Var}^{\rm GinUE} \Big ( \sum_{l=1}^N 
 \chi_{z_l \in \Lambda}  \Big ) = \int_{\mathbb R^2} d \mathbf r  \int_{\mathbb R^2} d \mathbf r' \, C_\infty^{\rm GinUE}(\mathbf r - 
 \mathbf r') \chi_{\mathbf r \in \Lambda}  \chi_{ \mathbf r' \in \Lambda}.
 \end{equation}
 Simple manipulation of (\ref{5.4i}) gives
 \begin{multline*}
  \lim_{N \to \infty} {\rm Var}^{\rm GinUE} \Big ( \sum_{l=1}^N 
 \chi_{z_l \in \Lambda}  \Big )   \\ = \int_{\mathbb R^2} d \mathbf r  \int_{\mathbb R^2} d \mathbf r' \, C_\infty^{\rm GinUE}(\mathbf r )
 (  \chi_{ \mathbf r  +  \mathbf r' \in \Lambda} - 1 )   \chi_{ \mathbf r' \in \Lambda} + | \Lambda |
  \int_{\mathbb R^2}  C_\infty^{\rm GinUE}(\mathbf r) \, d \mathbf r.
  \end{multline*}
  According to (\ref{5.4j}) the final term in this expression is equal to $S_\infty^{\rm GinUE}(\mathbf 0)$, which from
   (\ref{5.1a}) is equal to $0$. On the other hand, as the region $\Lambda$ is scaled to infinity in a self similar manner,
   $\Lambda \mapsto \lambda \Lambda$, with $\lambda \to \infty$,
   the quantity in the first term $  \int_{\mathbb R^2} d \mathbf r' \,  (  \chi_{ \mathbf r  +  \mathbf r' \in \Lambda} - 1 )   \chi_{ \mathbf r' \in \Lambda} $
   for fixed $\mathbf r$
   is to leading order proportional to the surface area of   $\lambda \Lambda$, $| \partial (\lambda \Lambda)| $ say \cite{MY80}.
   Hence in this limit
    \begin{equation}\label{5.4k} 
    \lim_{N \to \infty} {\rm Var}^{\rm GinUE} \Big ( \sum_{l=1}^N 
 \chi_{z_l \in \lambda \Lambda}  \Big )    \sim c_{\Lambda}  | \partial (\lambda \Lambda)|
   \end{equation}
   for some proportionality $c_{\Lambda}$. An illustration for $\Lambda$ a square centred at the origin
   and rotated at random is given in Figure \ref{Ffig2}. Generally in two or more dimensions point processes with the
   property that the variance of the number of particles in a region scales with the surface area of the region
    have been termed hyperuniform \cite{TS03, To16,  GL17}. 
     \begin{figure*}
\centering
\includegraphics[width=0.65\textwidth]{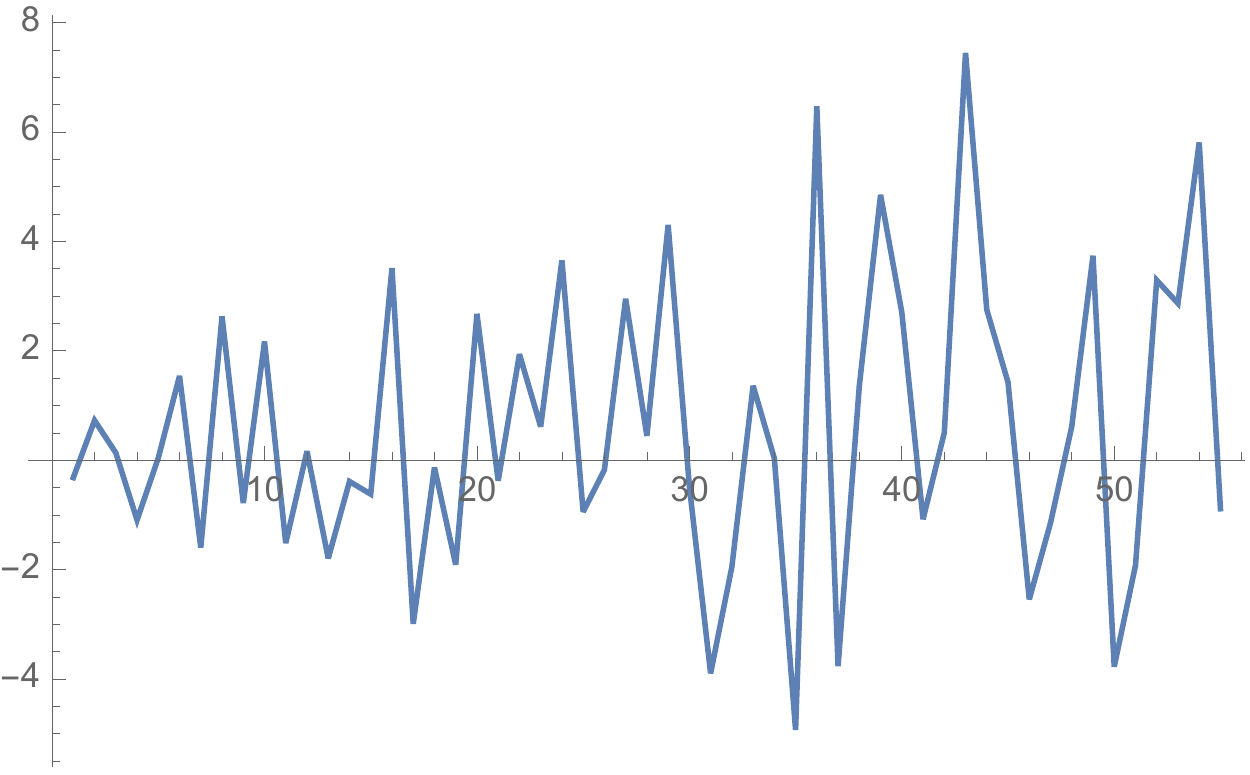}
\caption{Plot of values of the random variable $ ( \sum_{j=1}^N ( \chi_{|x_j| < L/2}  \chi_{|y_j| < L/2} )^{\circ} - L^2/\pi  )$ for $z_j = x_j + i y_j$
the eigenvalues of a single $N=1,600$ GinUE matrix, as a function of $L=1,2,\dots,56$; the values of the random variable have been joined
for visual clarity. The symbol $\circ$ indicates that for each $L$ the square has been rotated by some uniformly chosen angle. Note that the 
growth is approximately of the order of $\sqrt{L}$.}
\label{Ffig2}
\end{figure*}

   The quantity
     \begin{equation}\label{xr}
     \int_{\mathbb R^2} \chi_{\mathbf r + \mathbf r' \in \Lambda}  \chi_{ \mathbf r' \in \Lambda} \,  d \mathbf r' 
    \end{equation}
    relevant to a direct computation of (\ref{5.4i}) has been evaluated in \cite{TS03} for the case of $\lambda =    \lambda_R$ a disk of radius $R$
    centred at the origin. Then (\ref{xr}) is rotationally invariant and thus a function of $r/R$, where $r:= | \mathbf r |$, to be denoted
    $\alpha(r/R)$ say. We see from the definition that $\alpha(r/R)=0$ for $r \ge 2 R$. For $0 < r < 2 R$ the result of  \cite{TS03} gives
    $$
    \alpha(r/R) = {2 \over \pi} \Big ( {\rm Arcos} \, x - x (1 - x^2)^{1/2} \Big ), \quad x = r/ 2 R,
    $$
    which has the large $R$ form
    $$
      \alpha(r/R) =  1 - {2 \over \pi} {r \over R} + O\Big ( {r \over R} \Big )^2.
      $$
      This substituted in (\ref{5.4i}) implies that for large $R$
    \begin{equation}\label{xr1}   
   \lim_{N \to \infty} {\rm Var}^{\rm GinUE} \Big ( \sum_{l=1}^N 
 \chi_{z_l \in \Lambda_R}  \Big )  \sim    
    - 2 R \int_{\mathbb R^2} | \mathbf r |  C_\infty^{\rm GinUE}(\mathbf r) \, d \mathbf r = {R \over \sqrt{\pi}},
    \end{equation}
    where the equality follows upon recalling the definition    (\ref{5.4ii}) and the exact result
   (\ref{5.1a}). Note that this is consistent with (\ref{5.4k}).

   One feature of the GinUE eigenvalues is that to leading order their density is uniform in the
   disk $|z| < 1$. This feature is shared by zeros of the random polynomial \cite{Ha56,Ha96}
   $$
   p_N(z) = a_0 + a_1 z + \cdots + a_N z^N,
   $$
   where each coefficient $a_j$ is a zero mean complex random variable with variance $\sigma_j^2 = 1/ j!$.
   The  bulk large $N$ limiting
   form of the zeros two-point correlation function is known from  \cite{Ha96} to be given by
  \begin{equation}\label{13.a}
  \rho_{(2)}(z_1, z_2) = {1 \over \pi^2} f( | z_1 - z_2|^2/2),
  \end{equation}
  where
  \begin{equation}\label{13.b}    
  f(x) := { (\sinh^2 x + x^2) \cosh x - 2 x \sinh x \over \sinh^3 x} = {1 \over 2} {d^2 \over d x^2} ( x^2 \coth x );
   \end{equation}
   for the equality in (\ref{13.b}) see \cite{FH98}. 
   
   The asymptotic relation in (\ref{xr1}) applies equally as well to the present two-dimensional point
   process (to be denoted cGP), and gives \cite{FH98}
   \begin{equation}\label{13.c}    
   \lim_{R \to \infty} {1 \over | \partial \Lambda_R | } {\rm Var}^{\rm cGP} \Big ( \sum_{l=1}^N \chi_{z_l \in \Lambda_R} \Big ) =
   {1 \over 8 \pi^{3/2}} \zeta(3/2),
   \end{equation}
   where $\zeta(s)$ denotes the Riemann zeta function, and use has been made of (\ref{13.b}) to compute the integral in
   (\ref{xr1}).   
   
   With regards to the analogue of (\ref{5.4d}), we know the driving feature is the small $| \mathbf k |$ form
   of the structure function (\ref{5.4f}).  Defining $S_\infty^{\rm cGP}(\mathbf k)$ by the analogue of 
   (\ref{5.4j}) and expanding for small $| \mathbf k |$ shows
   $$
   S_\infty^{\rm cGP}(\mathbf k) = c_0 + c_2 | \mathbf k |^2 + c_4  | \mathbf k |^4 + \cdots, \qquad c_{2j} \propto \int_0^\infty r^{2j+1} C_\infty^{\rm cGP}(r) \, dr.
   $$
   It follows from (\ref{13.a}) and (\ref{13.b}) that $c_0 = c_2 = 0$. The first nonzero coefficient is $c_4$, which is readily computed giving
    \begin{equation}\label{13.cz}  
     S_\infty^{\rm cGP}(\mathbf k)  \sim c_4 | \mathbf k |^4, \qquad c_4 = \zeta(3)/8 \pi.
   \end{equation}     
 This implies that for $L \to \infty$ \cite{FH98}
   \begin{equation}\label{5.4d1}
\lim_{N \to \infty} {\rm Cov}^{\rm cGP} \, \Big ( \sum_{l=1}^N f(x_l/L, y_l/L),  \sum_{l=1}^N g(x_l/L,y_l/L) \Big )   \sim 
{c_4 \over (2 \pi)^2}  {1 \over L^2}  \int_{\mathbb R^2}  \hat{f}(\mathbf k) \hat{g}(- \mathbf k)  |\mathbf k|^4 \, d \mathbf k,
 \end{equation} 
 assuming the integral converges, or in words the covariance goes to zero at a rate proportional to $1/L^2$.
 
 \subsection{Summarising remarks and heuristics}
 \subsubsection{The two classes of large $N$ limits}
 The results of this Section have been based on the double integral formula for the covariance
 (\ref{3.1}). Starting from this generic formula, the aim has been to give its limiting form in two distinct
 large $N$ settings. One is a global scaling limit, in which for $N \to \infty$ the eigenvalue support is a
 finite integral. In the analysis of this Section, which has relied on explicit Fourier analysis of the two-point
 correlation function, analytic results for the global scaling limit of  (\ref{3.1}) were obtained for Dyson's 
 circular ensembles, a deformation of the CUE due to Gaudin, and the ensemble of real orthogonal matrices
 with Haar measure. The latter is distinct as translation invariance is broken. In all these cases it has been
 possible to reduce the double integral to a single integral with a simple integrand involving the Fourier
 transform of the linear statistics. The mechanism for this is that by direct calculation the Fourier 
 transform of the structure  function for these ensembles could be shown to have a simple form.
 For all the ensembles analysed in this limit, it was found that for large $N$ the covariance is of order
 unity for smooth linear statistics, This behaviour contrasts with that of a gas of noninteracting eigenvalues, 
 for which the covariance is proportional to $N$; recall (\ref{q2}). For the linear statistic counting the number of
 eigenvalues in an interval, the linear statistic is a step function and so not smooth. Exact calculation leads
 to the conclusion that the variance,  in the global scaling limit of the same ensembles analysed in the
 case of a smooth statistic, is then proportional to $\log N$.
 
 The other large $N$ setting of interest is to first compute what in statistical mechanics is referred to as the
 thermodynamic limit, and termed above as the bulk scaling limit.
 Thus the coordinates are scaled so that the mean density is of order unity and
 the limit $N \to \infty$ then performed. Next a length scale $L$ is introduced into the linear statistics so that
 they vary on this scale, and finally the large $L$ limit is considered. Direct analysis of this limit is simpler
 than for the global scaling limit. In addition to the ensembles already analysed, it is possible
 to study the covariance of two linear statistics for the eigenvalues of the circular $\beta$
 ensemble, for the real eigenvalues of the ensemble of $N \times N$ real Gaussian matrices, and the eigenvalues of
 complex Ginibre ensemble. In the case of real Gaussian matrices the covariance
 is proportional to $L$, which is a characteristic property of the structure function being nonzero at the origin.
 A feature of the eigenvalues of the complex Ginibre ensemble is that the variance of the counting function for
 the number of eigenvalues in a region scales with the length of the boundary of that region.
 
 \subsubsection{Consistency with log-gas predictions}\label{S2.10b}
 The log-gas analogy for the eigenvalue PDF for the circular $\beta$ ensemble
 (\ref{z.1})
 leads to predictions for both the smoothed  bulk and global scaled forms 
 of the quantity $C_{(2),N}(x,x')$, as required for the determination of the corresponding
 fluctuation formulas \cite{Be93,Ja95,Fo95}. The fluctuation formulas obtained using this
 heuristic are consistent with the exact results obtained in the case of the circular $\beta$ ensemble
 and moreover the working can be extended to apply to other random matrix ensembles. This is possible
 because of log-gas analogies for those random matrix ensembles too.
 
 When the eigenvalue PDF permits a Bolzmann factor interpretation $e^{-\beta U}$, it is possible to take the
 viewpoint that a linear statistic $U_u := \sum_{j=1}^N u(x_j)$ is a perturbing external one body potential, so that the
 perturbed Boltzmann factor becomes $e^{-\beta (U_u + U)}$. Expanding the factor $e^{-\beta U_u}$ to first order in
 $u$, $e^{-\beta U_u} \approx 1 - \beta \sum_{j=1}^N u(x_j)$  we can check from the definitions that
  \begin{equation}\label{C.2}
 q_u(x') :=  \langle n_{(1),N}(x') \rangle_u -   \langle n_{(1),N}(x') \rangle_{u=0} = - \beta \int_I u(x) C_{(2),N}(x,x') \, dx,
 \end{equation} 
 where $C_{(2),N}(x,x')$ is the quantity (\ref{C.1}) computed in the absence of $U_u$ and $n_{(1),N}(x') := \sum_{j=1}^N \delta(x' - x_j)$.
 The key hypothesis is that for large $N$, the LHS of (\ref{C.2}) is determined by the macroscopic electrostatics implied by the
 pair potential $\phi(z,z')$ in (\ref{z.1a}), and thus satisfies
 the integral equation
  \begin{equation}\label{C.3}
  - \int_0^{2 \pi} \log | \sin(x-x')/2 | q_u(x') \, dx' = u(x) + C.
 \end{equation} 
 Here the constant $C$ is determined by the particle conservation condition
     \begin{equation}\label{C.3a}
     \int_0^{2 \pi} q_u(x') \, dx' = 0.
  \end{equation}
  The functional form of $q_u(x)$ can readily be determined \cite[Prop.~14.3.4]{Fo10}.
  
  \begin{proposition}
  With $u(x) = \sum_{p=-\infty}^\infty u_p e^{i p \theta}$ given, the solution $q_u(x)$ of the integral
  equation (\ref{C.3}), subject to the constraint (\ref{C.3a}), is given by the Fourier series
  \begin{equation}\label{C.3b}  
  q_u(x) = - {1 \over \pi} \sum_{p=-\infty}^\infty |p| u_p e^{i p \theta}.
    \end{equation}
    \end{proposition}
    
    \begin{proof}
    Substituting the Fourier series implicit in (\ref{Pp7}) for $ \log | \sin(x-x')/2 |$ in (\ref{C.3}) together with the 
    Fourier series of $u(x)$ gives
    $$
    \sum_{p=-\infty}^\infty \alpha_p e^{i p \theta} \int_0^{2 \pi} q_u(x') e^{-i p x'} \, dx' = 
    \sum_{p=-\infty}^\infty |p| u_p e^{i p \theta} + C,
    $$
    where $\alpha_p$ is defined in   (\ref{Pp7a}). Equating coefficients of $e^{i p \theta}$ and requiring
  (\ref{C.3a})  gives the value of the
    Fourier coefficients of $q_u(x)$ and (\ref{C.3b}) follows.
    \end{proof}
    
    Now it follows from (\ref{C.2}), the definition of $C_{(2),N}(x,x')$ (\ref{C.1}) and 
    (\ref{3.1}) that
     \begin{equation}\label{C.3c}   
     {\rm Cov} \, \Big ( \sum_{l=1}^N f(x_l), \sum_{l=1}^N g(x_l) \Big ) = - {1 \over \beta} \int_I f(x) q_g(x) \, dx.
  \end{equation}
  In the case of the circular-$\beta$ ensemble, assuming the validity of the hypothesis that for
  large $N$, $q_u$ is determined by   (\ref{C.3}) and we see by substituting (\ref{C.3b}) that the
  fluctuation formula (\ref{3.4eb})  results, now predicted to be valid for general $\beta > 0$.
  
  We know from Proposition \ref{p4.1} that the limiting covariance formula in the case of the independent
  eigenvalues for random $O^+(N/2)$ matrices takes the simple form (\ref{4.4e}) involving the cosine transform.
  The joint PDF for the independent eigenvalues for Haar distributed $O^+(N/2)$ matrices is proportional to
  (see e.g.~\cite[Eq.~(2.62)]{Fo10})
    \begin{equation}\label{C.3cx}  
  \prod_{1 \le j < k \le {N/2}} | \cos x_k - \cos x_j |^\beta, \qquad 0 \le x_l \le \pi,
  \end{equation}
  with $\beta = 2$. Writing this in Boltzmann factor form the corresponding pair potential is $\phi^{\rm c}(x_j, x_k) = -
  \log | \cos x_k - \cos x_j |$ (here the superscipt ``c" indicates the involvement of cosine). Hence from
  the viewpoint of macroscopic electrostatics the task is to solve the integral equation
     \begin{equation}\label{C.3d}  
     - \int_0^\pi \log | \cos x - \cos x'| q_u^{\rm c}(x') \, dx' = u(x) + C, \quad {\rm subject \: to} \quad \int_0^{\pi} q_u^{\rm c}(x') \, dx' =0.
 \end{equation}     
 Knowledge of the expansion (see e.g.~\cite[Exercises 1.4 q.4]{Fo10})
 $$
 \log (2 | \cos x - \cos t |) = - \sum_{n=1}^\infty {2 \over n} \cos n x \cos nt
 $$
shows that with $u(x)$ given in terms of its cosine expansion as in (\ref{2.55}), the
solution of (\ref{C.3d}) is given by
   \begin{equation}\label{C.3cy}  
q_u^{\rm c}(x) = - {2 \over \pi} \sum_{p=1}^\infty p u_p^{\rm c} \cos p x.
 \end{equation}
Substituting this in the RHS of (\ref{C.3c}) with $\beta = 2$ reclaims (\ref{4.4e1}) and moreover predicts
that its generalisation for $\beta > 0$ in the sense of (\ref{C.3cx}) is to multiply the RHS therein
by $2/\beta$. 
 
\begin{remark}
Gaudin's eigenvalue PDF (\ref{q3}) can be written in Boltzmann factor form involving a pair potential.
However this pair potential, as seen in the second expression of (\ref{q3}), is not long range in an 
appropriate scaling limit. Due to this, it is not expected that the hypothesis of an analogue of
(\ref{C.3}) will be valid. Indeed, assuming it is leads to a result for the covariance which is in
contradiction to the exact result (\ref{4.4e1}).
\end{remark}

\section{Other structures leading to explicit formulas}\label{S3}
\subsection{The Gaussian $\beta$-ensemble}
Under the change of variables $y_j = \cos x_j$ and with $N$ replaced by $N/2$ (this for
convenience) the PDF (\ref{C.3cx}) becomes proportional to 
  \begin{equation}\label{4.1}
  \prod_{l=1}^N ( 1 - y_l^2)^{-1/2} \prod_{1 \le j < k \le N} | y_k - y_j|^\beta, \qquad | y_l| < 1.
   \end{equation}     
   This is an example of a Jacobi $\beta$-ensemble (see e.g.~\cite[\S 3.11]{Fo10} and Section \ref{S3.2}
   below).
 It follows from the derivation of (\ref{C.3cy}) that the macroscopic
 log-gas viewpoint predicts for the corresponding fluctuation formula,
  \begin{equation}\label{4.1a} 
  \lim_{N \to \infty} {\rm Cov}^{\rm J}  \Big ( \sum_{j=1}^N f(x_j),  \sum_{j=1}^N g(x_j) \Big )  = {2 \over \beta}
  \sum_{n=1}^\infty n f_n^{\rm c} g_n^{\rm c}  
  \end{equation} 
  independent of the details of the one body term $ \prod_{l=1}^N ( 1 - y_l^2)^{-1/2}$. The important point is
  that the eigenvalue support is the interval $(-1,1)$, or more generally a single interval $(a,b)$,
  and that the underlying pair potential in the Boltzmann factor analogy is
  $- \log | y_j - y_k|$.
  
  The leading order eigenvalue support being a single interval is shared by a number of
  random matrix ensembles with an eigenvalue PDF of the form
    \begin{equation}\label{4.1c} 
    \prod_{l=1}^N e^{-\beta N V(x_l)} \prod_{1 \le j < k \le N} | x_k - x_j |^\beta,
  \end{equation}   
  and thus also having an underlying logarithmic pair potential.
  Our interest in this section is in the Gaussian $\beta$-ensemble, specified by setting $V(x) = x^2$ in (\ref{4.1c}).
  Even though for finite $N$ the eigenvalues may be located anywhere on the
 real line,  as $N \to \infty$ their support is the  single interval, $(-1,1)$ say. The log-gas argument
 of subsection \ref{S2.10b}
 then predicts an identical expression to (\ref{4.1a}) for the covariance,
 \begin{align}\label{4.1d} 
  \lim_{N \to \infty} {\rm Cov}^{\rm G} \,& \Big ( \sum_{j=1}^N f(x_j),  \sum_{j=1}^N g(x_j) \Big )  = {2 \over \beta}
  \sum_{n=1}^\infty n f_n^{\rm c} g_n^{\rm c}   \nonumber  \\
  & =  {2 \over \beta} {1 \over 8 \pi^2} \int_{-\pi}^\pi d\theta  \int_{-\pi}^\pi d\phi \, { (f(\cos \theta) -  f(\cos \phi))  (g(\cos \theta) -  g(\cos \phi)) \over
  | e^{i \theta} - e^{i \phi} |^2}  \nonumber  \\
   &= {2 \over \beta} {1 \over 4 \pi^2} \int_{-1}^1 dx  \int_{-1}^1 dy \, {(f(x) - f(y)) (g(x) - g(y)) \over (x - y)^2}
  {1 - x y \over (1 - x^2)^{1/2} (1 - y^2)^{1/2} }.  
\end{align}
Here the second equality can be seen to imply the first upon using the identity 
(\ref{Pp7a}) and integrating by parts. In the second equality the integration domain can be reduced to
$[0,\pi]^2$ by replacing the denominator by
$$
{1 \over  | e^{i \theta} - e^{i \phi} |^2}  + {1 \over  | e^{i \theta} - e^{-i \phi} |^2} 
$$
and multiplying by 2. A simple change of variables then gives the third equality \cite{La18}.
The reference  \cite{CD01} gives a different perspective on the first equality starting from the
third equality.

The identity
  \begin{equation}\label{3.15a}
  { (1 - xy) \over (1 - x^2)^{1/2} (1 - y^2)^{1/2} (x - y)^2} = {1 \over (1 - x^2)^{1/2}} {\partial^2 \over \partial x \partial y} \Big ( (1 - y^2)^{1/2} \log | x - y| \Big )
   \end{equation}
  allows for the rewrite of (\ref{4.1a}) in the case $f=g$ \cite{Jo98, CL98}
   \begin{equation}\label{4.1aH}
    \lim_{N \to \infty} {\rm Var}^{\rm G} \, \Big ( \sum_{j=1}^N f(x_j) \Big ) = {2 \over \beta} {1 \over  \pi^2} \int_{-1}^1 dy \, {f(y) \over \sqrt{1 - y^2}}
    \int_{-1}^1 dx \, {f'(x)  \sqrt{1 - x^2} \over x - y}. 
    \end{equation}
    One can check too that the LHS of (\ref{3.15a}) can be written as \cite{BY05}
      \begin{equation}\label{3.15b}
      - {1 \over 2} {\partial^2 \over \partial x \partial y} \log \bigg ( {1 - xy + \sqrt{(1 - x^2) (1 - y^2)} \over 1 - xy - \sqrt{(1 - x^2) (1 - y^2)} } \bigg ).
  \end{equation}
  This substituted in (\ref{4.1a}) gives, upon integration by parts,
  \begin{align}\label{4.1aG+} 
  \lim_{N \to \infty} {\rm Cov}^{\rm G} \, & \Big ( \sum_{j=1}^N f(x_j),  \sum_{j=1}^N g(x_j) \Big )  \nonumber  \\
   & = {2 \over \beta} {1 \over 4 \pi^2} \int_{-1}^1 dx \, f'(x)  \int_{-1}^1 dy \, g'(y) \, \log \bigg ( {1 - xy + \sqrt{(1 - x^2) (1 - y^2)} \over 1 - xy - \sqrt{(1 - x^2) (1 - y^2)} } \bigg ) \nonumber \\
    & = {2 \over \beta} {1 \over 4 \pi^2} \oint_{|z| = 1 \atop y > 0} dz  \oint_{|w|=1 \atop v > 0}  dw \,  f'(x) g'(u)
  \log \Big | {1 - z w \over 1 - z \bar{w}} \Big | {yv \over z w},
  \end{align} 
  where in the final expression  $z = x + i y$, $w = u + i v$; see    \cite{Pa13} for details relating to the second equality,
   which provides a link with the Gaussian free field \cite{Bo14}.  
According to (\ref{3.2}) a corollary of the third form in (\ref{4.1a}) is that for the Gaussian (and Jacobi) $\beta$-ensembles,
   \begin{equation}\label{4.1b} 
   \lim_{N \to \infty} \rho_{(2), N}^T(x,y) \doteq - {1 \over \beta} {1 \over \pi^2} { (1 - xy) \over (1 - x^2)^{1/2} (1 - y^2)^{1/2} (x - y)^2},
   \end{equation}
   supported on $|x|, |y| < 1$. Here the symbol $ \doteq $ is used to indicate the limiting functional form has been smoothed with
   respect to test functions; as seen in (\ref{Pp5}) the pointwise limit is not expected to exist.

The density of eigenvalues for the Gaussian $\beta$-ensemble in the large $N$ limit  on
their support $(-1,1)$ is the Wigner semi-circle law
$\rho_{(1)}^{\rm W}(x) = {2 \over \pi} (1 - x^2)^{1/2}$; see e.g.~\cite[\S 1.4]{Fo10}. 
 The case
$\beta = 2$ --- referred to as the Gaussian unitary ensemble (GUE) --- is realised by the 
complex Hermitian random matrices $H^{\rm c} = {1 \over 2} (X + X^\dagger)$, where $X$ is an $N \times N$ complex standard
Gaussian matrix; scaling these matrices by $1/\sqrt{2N}$ gives rise to (\ref{4.1c}) with $\beta = 2$.
Analogously, the case $\beta = 1$ --- known as the Gaussian orthogonal ensemble (GOE) ---
is realised by the real symmetric random matrices $H^{\rm r} = {1 \over 2} (X + X^T)$, where
$X$ is an $N \times N$ real standard Gaussian matrix. Scaling these matrices by $1/\sqrt{2N}$ as for the
GUE gives (\ref{4.1c}) with $\beta = 1$ as the eigenvalue PDF. A realisation in terms of Gaussian
random matrices is known in the case $\beta = 4$ too (see e.g.~\cite[\S 1.3.2]{Fo10}), which is
referred to as the Gaussian symplectic ensemble (GSE). Together the values $\beta = 1, 2$ and
4 are referred to as the classical cases, with their underlying symmetries isolated by Dyson \cite{Dy62c}


For the GOE case $\beta = 1$ the functional form (\ref{4.1b}) was first derived  in the 1978 work of
French, Mello and Pandey \cite{FMP78}. This was done through an analysis of the covariance formula for the pair
of linear statistics $\sum_{j=1}^N x_j^p = {\rm Tr} \, H^p$, $\sum_{j=1}^N x_j^q = {\rm Tr} \, H^q$. As such the strategy
used was a generalisation of the method of moments as introduced by Wigner to study the eigenvalue density
\cite{Wi55}. Soon after Pandey \cite{Pa81} realised that this polynomial covariance could usefully be encoded by considering
instead
  \begin{equation}\label{4.1dz} 
  {\rm Cov} \, \Big ( \sum_{j=1}^N {1 \over x - x_j},  \sum_{j=1}^N {1 \over y - x_j} \Big ) =
   {\rm Cov} \, \Big ( {\rm Tr} (x \mathbb I_N - H)^{-1},  {\rm Tr} (y \mathbb I_N - H)^{-1} \Big ).
     \end{equation}  
     Studying the mean value of one such linear statistic, i.e.
 \begin{equation}\label{4.1e} 
 \bigg \langle  \sum_{j=1}^N {1 \over x - x_j} \bigg \rangle,
 \end{equation}     
corresponds to the Steiltjes transform of the eigenvalue
density, the analysis of which was introduced into random matrix theory by Pastur \cite{Pa72}; see also the
text \cite{PS11}. A number of derivations of (\ref{4.1b}) additional to those those of \cite{FMP78,Pa81} have
been listed in the recent work \cite{Sa21}.

Here we will focus on a method of derivation of (\ref{4.1b}) based on the loop equation formalism, first made
use of in this context in relation to the GUE \cite{AM90}, and generalised to the Gaussian $\beta$-ensemble
for general $\beta > 0$ in \cite{BMS11,MMPS12,WF14}.   Denote the covariance (\ref{4.1dz}) by $\overline{W}_2^{\rm G}(x,y;N,\kappa)$
and the mean (\ref{4.1e}) by $\overline{W}_1(x;N,\kappa)$, where $\kappa := \beta / 2$.
An integration by parts procedure gives that these quantities for the Gaussian $\beta$-ensemble are related
by the first loop equation
\begin{equation}\label{L1}
\Big ( { \kappa - 1 \over 2 N}   {\partial \over \partial x_1} - 2 \kappa  x_1 \Big ) \overline{W}_1^{\rm G}(x_1;N,\kappa) + 2 N \kappa
+ {\kappa \over 2 N} \Big ( \overline{W}_2^{\rm G}(x_1, x_1;N,\kappa) + (\overline{W}_1^{\rm G}(x_1;N,\kappa) )^2 \Big ) = 0.
\end{equation}

To progress further, the $1/N$ expansions
\begin{align}\label{L2}
 \overline{W}_1^{\rm G}(x;N,\kappa)  & = N W_{1,0}^{\rm G}(x;\kappa) + W_{1,1}^{\rm G}(x;\kappa) + {1 \over N} W_{2,1}^{\rm G}(x;\kappa) + \cdots,  \nonumber \\
  \overline{W}_2^{\rm G}(x,y;N,\kappa)   & =  W_{2,0}^{\rm G}(x,y;\kappa) + {1 \over N} W_{2,1}^{\rm G}(x,y;\kappa) + \cdots,
\end{align}
rigorously justified in \cite{BG12},  are introduced into (\ref{L1}). Equating like powers of $N$ gives a quadratic equation for
$ W_{1,0}^{\rm G}(x;\kappa)$ with solution 
\begin{equation}\label{L2a}
W_{1,0}^{\rm G}(x;\kappa) =  2 (x - \sqrt{x^2 - 1})
\end{equation}
independent of $\kappa$.
 With this established,
equating terms independent of $N$ gives a linear equation for $ W_{1,1}^{\rm G}(x;\kappa)$ with solution
\begin{equation}\label{L2az} 
 W_{1,1}^{\rm G}(x;\kappa) = {1 \over 2} \Big ( 1 - {1 \over \kappa}   \Big )
 \Big ( {1 \over \sqrt{x^2 - 1}} - {x \over x^2 - 1} \Big ).
\end{equation}
However, at order $1/N$ in (\ref{L2}), two unknown quantities $  W_{1,2}^{\rm G}(x;\kappa)$ and $W_{2,0}^{\rm G}(x,y;\kappa)$
are involved. To separate these unknowns the second equation of the loop hierarchy is needed. As well as involving
$ \overline{W}_1^{\rm G}$ and $ \overline{W}_2^{\rm G}$, this second equation involves the three point quantity $ \overline{W}_3^{\rm G}=  
\overline{W}_3^{\rm G}(x_1, x_2, x_3; N, \kappa)$; see e.g.~\cite{WF14} for its precise definition. For large $N$, continuing
the pattern from (\ref{L2}), $ \overline{W}_3^{\rm G} = O(1/N)$ and does not contribute to leading order in the second loop
equation. In fact the only unknown to leading order is $W_{2,0}^{\rm G}(x,y;\kappa)$, with the equation linear in this
quantity and having solution
\begin{equation}\label{L2b} 
 W_{2,0}^{\rm G}(x,y;\kappa) = {2 \over \beta} \bigg ( {xy - 1 \over 2 (x - y)^2 \sqrt{(x^2 - 1)(y^2 - 1)}} - {1 \over 2 ( x - y)^2} \bigg ).
\end{equation} 
Now generally for $I$ an open interval, $\alpha(t)$ continuous on $I$ and
$$
f(z) = \int_I {\alpha(t) \over z - t} \, dt, \quad z \notin I,
$$
the inverse formula for the Stieltjes transform gives
$$
\alpha(t) = {1 \over \pi} \lim_{\epsilon \to 0^+} {\rm Im} \, f(t - i \epsilon), \qquad t \in I.
$$
Applying this with respect to both the $x$ and $y$ variable in (\ref{L2b}) the functional form
in (\ref{4.1b}) is obtained.

By noting that for $\gamma$ a simple contour enclosing $(-1,1)$ in the complex plane, and $f$ analytic on and
within $\gamma$,  we have by Cauchy's integral formula the representation
\begin{equation}\label{CI}
f(x) = {1 \over 2 \pi i} \oint {f(z) \over z - x} \, dz
\end{equation} 
shows
\begin{equation}\label{L2r} 
\Big ( {1 \over 2 \pi i } \Big )^2  \oint dw \, f(w)  \oint dz \, g(z) \, {\rm Cov} \, \Big (
\sum_{j=1}^N {1 \over w - x_j}, \sum_{j=1}^N {1 \over z - x_j} \Big ) = {\rm Cov} \Big ( \sum_{j=1}^N f(x_j),  \sum_{j=1}^N g(x_j) \Big ).
\end{equation} 
This fact, together with certain large deviation bounds, enables the rigorous deduction of (\ref{4.1a}) for $f,g$ analytic in a
neighbourhood of $(-1,1)$ from knowledge of (\ref{L2b}); see e.g.~\cite{BG12}. Replacing the Cauchy integral
formula by the Helffer-Sj\"ostrand formula \cite{Er11,ORS11,NY16,HK17} allows for the conditions on $f$ and $g$ to be
further weakened.

\begin{remark}\label{R3.1}
1.~Consider the linear statistics $\sum_{j=1}^N x_j^p$, ($p=1,2$). As noted in \cite{BF97f}, in the case
of the Gaussian $\beta$-ensemble the corresponding characteristic functions (\ref{Pp}) are simple to evaluate,
\begin{equation}\label{L2cz}
\hat{P}_{N, f = x}(t) = e^{-t^2/4 \beta}, \qquad \hat{P}_{N, f = x^2}(t) =  ( 1 - i t/ \beta N)^{-(1/2)(N + \beta N (N - 1)/2)}.
\end{equation} 
Recalling (\ref{Pp4}), it follows
\begin{equation}\label{L2d}
\lim_{N \to \infty} {\rm Var}^{\rm G}\Big ( \sum_{j=1}^N x_j \Big ) = {1 \over 2 \beta}, \qquad
\lim_{N \to \infty} {\rm Var}^{\rm G}\Big ( \sum_{j=1}^N x_j^2 \Big ) = {1 \over 4 \beta},
\end{equation}
as is consistent with (\ref{4.1a}). We note too that in the limit $N \to \infty$ the second characteristic function in (\ref{L2c})
becomes a Gaussian like the first upon the recentring $\hat{P}_{N,f=x^2}(t) \mapsto \hat{P}_{N,f=x^2}(t) e^{-i t\langle
\sum_{j=1}^N x_j^2 \rangle}$. That the rescaled limiting distribution of the polynomial linear statistics $\sum_{j=1}^N x_j^k$
($k \in \mathbb Z^+$) is a Gaussian with variance as implied by (\ref{4.1a}) was first established by Johansson \cite{Jo98}.
In \cite{BG12} this result was extended to a wider class of linear statistics using a loop equation analysis; see also \cite{BLS18},  \cite{LLW19}
and \cite{BMP22}.
 \\
2.~With $f = x^{k_1}$, $g=x^{k_2}$ it is a known corollary of (\ref{4.1a}) (see e.g.~\cite{MN04,DE05,BH16,CMV15}) that
\begin{align}\label{Gm}
2^{2k_1 + 2 k_2} \lim_{N \to \infty} {\rm Cov}^{\rm G} \Big ( \sum_{j=1}^N x_j^{2k_1},   \sum_{j=1}^N x_j^{2k_2} \Big )  & = {2 \over \beta}
{(2k_1)! (2 k_2)! \over (k_1!)^2 (k_2!)^2 } {k_1 k_2 \over k_1 + k_2}  \nonumber \\
2^{2k_1 + 2 k_2+2} \lim_{N \to \infty} {\rm Cov}^{\rm G} \Big ( \sum_{j=1}^N x_j^{2k_1+1},   \sum_{j=1}^N x_j^{2k_2+1} \Big )  & = {2 \over \beta}
{(2k_1+1)! (2 k_2+1)! \over (k_1!)^2 (k_2!)^2 } {k_1 k_2 \over k_1 + k_2+1} .
\end{align}
It is pointed out in \cite{MST09,CMV15} that from a particular combinatorial viewpoint, equivalent enumeration formulas 
are contained in the work of Tutte \cite{Tu62}. \\
 3.~(Variance for number of particles in an interval) As in Proposition \ref{p2.1z}, an example of a linear
    statistic for which (\ref{4.1a}) breaks down in $f=g= \chi_{(a,b)}$, for $(a,b) \subset (-1,1)$.
    Let $N_{(a,b)} := \sum_{j=1}^N \chi_{x_j \in (a,b)}$. It is proved in \cite{BS12} that for $(a,b) = (0,1)$
     \begin{equation}\label{G6d} 
     \lim_{N \to \infty} {1 \over \log N} {\rm Var}^{\rm G}(  N_{(a,b)}) = {1 \over \pi^2 \beta}
 \end{equation}      
(cf.~(\ref{3.4f})). This was conjectured to hold true in the general case, a fact which has been
established in the special cases $\beta = 1,2$ and 4 \cite{Sh15}.
\end{remark}

  A half  line scaling is possible. Suppose in (\ref{4.1a}) that $f(-1+x) = F(X/L)$, $g(-1+x) = G(X/L)$, where $F(X), G(X)$ are assumed to decay
at infinity. In the second equality of (\ref{4.1a}) change variables $-1 + x = X/L$, $-1 + y = Y/L$. This shows
\begin{multline}\label{4.1aX} 
 \lim_{L \to \infty}  \lim_{N \to \infty} {\rm Cov}^{\rm G} \, \Big ( \sum_{j=1}^N f(x_j),  \sum_{j=1}^N g(x_j) \Big )  \bigg |_{f(-1+x) = F(X/L) \atop g(-1+x) = G(X/L)} \\
   = {1 \over \beta} {1 \over 4 \pi^2} \int_{0}^\infty dX  \int_{0}^\infty dY \, {(F(X) - F(Y)) (G(X) - G(Y)) \over (X - Y)^2}
  {X + Y \over (XY)^{1/2} }.
  \end{multline}
  The identity,
  $$
  {1 \over 2} {1 \over \sqrt{XY}} {X + Y \over (X - Y)^2} = {\partial^2 \over \partial X  \partial Y} \log \bigg | {\sqrt{X} - \sqrt{Y}  \over \sqrt{X} + \sqrt{Y} } \bigg |
  $$
  substituted in (\ref{4.1aX}) gives, upon integration by parts, the rewrite of the RHS of (\ref{4.1aX}),
  \begin{equation}\label{4.1aY} 
  - {1 \over \beta  \pi^2}  \int_{0}^\infty dX \, F'(X) \int_{0}^\infty dY \, G'(Y)  \log \bigg | {\sqrt{X} - \sqrt{Y}  \over \sqrt{X} + \sqrt{Y} } \bigg |.
 \end{equation}
 Furthermore, making use of the Fourier transform
 $$
  - \log \Big | \tanh {\pi y \over 4 \pi} \Big | = {1 \over 2} \int_{-\infty}^\infty {\tanh \pi x \over x} e^{ixy} \, dx,
  $$
  shows that an alternative form to (\ref{4.1aY}) is \cite{Be93,CV14}
 \begin{equation}\label{4.1aZ}   
 {1 \over \beta  \pi^2} \int_{-\infty}^\infty  \hat{F}^{\rm e}(k)   \hat{G}^{\rm e}(-k)    k \tanh(\pi k) \, dk, \qquad \hat{a}^{\rm e}(k) := \int_{-\infty}^\infty e^{i k x} a(e^x) \, dx.
  \end{equation} 
  The change of variables $X = u^2, Y= v^2$ shows that (\ref{4.1aX}) permits the further rewrites \cite{BW99,MC20}
  \begin{equation}\label{4.1aZ1} 
  {1 \over \beta} {1 \over (2 \pi)^2} \int_{-\infty}^\infty du  \int_{-\infty}^\infty dv \, {(F(u^2) - F(v^2)) (G(u^2) - G(v^2)) \over (u - v)^2}
 = {1 \over \beta} {1 \over (2 \pi)^2 } \int_{-\infty}^\infty |k|   \hat{F}^{\rm s}(k)   \hat{G}^{\rm s}(-k)  \, dk,
  \end{equation} 
  where $\hat{h}^{\rm s}(k) := \int_{-\infty}^\infty h(x^2) e^{i k x} \, dk$.
  
  Both the GOE and GUE allow for extensions to involve what historically has been termed an 
  external source. With $G$ a GOE matrix ($\beta = 1$) or GUE matrix ($\beta = 2$) the corresponding ensembles with an external  source
   (to be denoted G${}^\boxplus$) are specified by the sum
    \begin{equation}\label{G6a}
    A + G,
\end{equation} 
where it is assumed that $A$ is real symmetric ($\beta = 1$) or complex Hermitian    ($\beta = 2$).
Suppose furthermore that as $N \to \infty$ the eigenvalue density of ${1 \over \sqrt{2N}} A$ has a compactly
supported limiting density with corresponding measure $d \mu(x)$. Let $\tilde{m}(z)$ be the solution
of the Pastur equation
  \begin{equation}\label{G6b} 
\tilde{m}(z) = \int_{-\infty}^\infty {1 \over t - 2z -   \tilde{m}(z)} \, d \mu(t)
  \end{equation} 
  which has positive imaginary part for $z$ in the upper half complex plane. The quantity $\tilde{m}(z)$ then
  corresponds to the Stieltjes transform of the limiting scaled eigenvalue density of (\ref{G6a}); see e.g.~\cite{PS11}.
  We have from
  \cite{DF19,JL19,LSX21}
  \begin{equation}\label{G6c}   
  W_{2,0}^{\rm G^\boxplus}(x,y) = - {2 \over \beta} {\partial^2 \over \partial x  \partial y} \log \bigg ( 1 -  
  \int_{-\infty}^\infty {  d \mu(t)  \over (t - 2x -   \tilde{m}(z))  (t - 2y -   \tilde{m}(z)) } \bigg ).
    \end{equation} 
    Note that in the special case that $d \mu(t) = \delta (t) dt$ corresponding to $A=0$ in (\ref{G6a}),
    it follows from (\ref{G6b}) that $1/(2z + \tilde{m}(z)) = -  \tilde{m}(z)$. Using this in (\ref{G6c})
    in this same special case reclaims the form of $ W_{2,0}^{\rm G}(x,y) $ given in (\ref{3.36}) below.
  A formula of a different type for the variance of a polynomial linear statistic in the case of the GUE with
  a source has been given in \cite{DFK21}, which comes about through its relation to multiple
  orthogonal polynomials.  Another point of interest is that the external source model
  (\ref{G6a}) can be viewed in terms of Dyson Brownian motion; see e.g.~\cite[\S 11.1]{Fo10}
  for $\beta = 1,2$ and 4, and \cite{Fo13} for general $\beta > 0$.
  A study of the corresponding fluctuation formulas for linear statistics from this perspective
  can be found in \cite{Be08}. 
     
Linear statistics ranging over a restriction of the full set of eigenvalues are of interest
\cite{BPZ13,GMT17,Ga21}. With the eigenvalues ordered $x_1 > x_2 > \cdots > x_N$,
considered is the fluctuation of $\sum_{j=1}^K f(x_j)$, where $K/N \to \gamma$ as $N \to \infty$.
Define $c = c(\gamma)$ by
  \begin{equation}\label{G6dz}  
  \gamma = {2 \over \pi} \int_c^1 \sqrt{1 - x^2} \, dx,
  \end{equation} 
  so that the fraction of eigenvalues in the interval $(c,1)$ of the Wigner semi-circle is $\gamma$.
  With $f_c(x) = (f(x) - f(c)) \chi_{1 > x > c}$ it is proved in \cite{BPZ13} that for the GUE
  \begin{equation}\label{G6e}    
  \lim_{N \to \infty} {\rm Var}^{\rm GUE} \Big ( \sum_{j=1}^K f(x_j) \Big ) = \lim_{N \to \infty}
  {\rm Var}^{\rm GUE} \Big ( \sum_{j=1}^N f_c(x_j) \Big ),
  \end{equation} 
  with the RHS in turn being given by any of the formulas (\ref{4.1a}) and (\ref{4.1aG+}).
  A case of particular interest is $f(x) = x^2$.   Use of the first formula in (\ref{4.1aG+}) and
  evaluating the integral via computer algebra shows
   \begin{multline}\label{G6f} 
  \lim_{N \to \infty} {\rm Var}^{\rm GUE} \Big ( \sum_{j=1}^K x_j^2 \Big ) =   \\
  {1 \over \pi^2} {1 \over 8} \Big ( 3 c^4 - 4 c^3 \sqrt{1 - c^2}
  {\rm Arccos}\, c + c^2 (-7 + 2 c \sqrt{1 - c^2}  {\rm Arccos}\, c ) + 4 + ( {\rm Arccos}\, c)^2 \Big ).
  \end{multline}
  Up to a simple scaling of $c$, this formula was first derived in the recent paper \cite[Eq.~(70)]{Ga21},
  from a large deviations viewpoint. Setting $c=0$ and using (\ref{G6dz}) and (\ref{G6e}) implies
   \begin{equation}\label{G6h}  
 {\rm Var}^{\rm GUE} \Big ( \sum_{j=1}^N x_j^2 \chi_{x_j > 0} \Big )   = {1 \over 16} \Big ( 1 + {16 \over \pi^2} \Big ).
  \end{equation} 
  Up to a simple scale factor, this result in the GOE case is known from \cite[\S 6.2]{MT14}, where it
  was shown to have relevance to the distribution of intrinsic volumes for the cone of
  positive semidefinite matrices.
  
  Also of interest are linear statistics associated with submatrices. For an $N \times N$ matrix $H$ and
  $I \subset \{1,\dots,N\}$, $| I | \le N$, denote by $H(I)$ the $| I | \times | I |$ Hermitian matrix formed by the 
  intersection of the rows and columns labelled by $I$ of $H$. 
 Let $H$ be chosen from the GOE ($\beta = 1$) or the GUE ($\beta = 2$). Motivated by the relevance of submatrices of
  random Hermitian matrices to models admitting a stepped surfaces interpretation --- see the overview
  \cite{Bo15} --- and moreover relating to the Gaussian free field, Borodin \cite{Bo14}, using methods
  from \cite[Ch.~2]{AGZ09}, took up the
  problem of computing the covariance for the pair of linear statistics ${\rm Tr}(H(I_p))^{k_p},
 {\rm Tr}(H(I_q))^{k_q}$,  under the assumption that
 $$
 \lim_{N \to \infty} {| I_p| \over N}  =: b_p > 0, \quad  \lim_{N \to \infty} {| I_q| \over N} =: b_q > 0, \quad
 \lim_{N \to \infty}  {| I_p \cap I_q | \over N} =: c_{pq} > 0.
 $$
 Let $T_n(x)$ denote the $n$-th Chebyshev polynomial of the first kind, $T_n(\cos \theta) = \cos n \theta$.
 Let $\tilde{H}(I_p)$ denote $H(I_p)$ scale so that its limiting eigenvalue support is the interval $(-1,1)$.
 From \cite{Bo14} we have
   \begin{equation}\label{G6i}    
   \lim_{N \to \infty} {\rm Cov} \Big ( {\rm Tr} (T_{k_p} ( \tilde{H}(I_p)),  {\rm Tr} (T_{k_q} ( \tilde{H}(I_q)) \Big ) =
   \delta_{k_p, k_q} {k_p \over 2 \beta} \Big ( {c_{pq} \over \sqrt{b_p b_q} }\Big )^{k_p}.
 \end{equation} 
 Note that this is consistent with the first line of (\ref{3.2}) in the case $b_p = b_q = c_{pq} = 1$.
 This theme, with emphasis placed on the form of the covariance written in a form relating to
 the Gaussian free field --- recall the final expression in (\ref{4.1aG+}) --- has been followed up
 in \cite{BG15, DP18, Du18, BG18,CE20}, amongst other works.
 
 As our final point specifically in relation to the Gaussian $\beta$-ensemble, we will review fluctuation
 formulas associated with a particular high temperature limit. The latter is specified by replacing
 $e^{-\beta N V(x_l)}$ in (\ref{4.1c}) by $e^{-x_l^2/2}$, setting $\beta = 2 \alpha/N$ and taking $N \to \infty$.
 It is known that the corresponding normalised density, $\rho_{(1)}(x;\alpha)$ say, has the exact functional
 form \cite{ABG12} (see \cite{FM21} for a derivation via loop equations)
 $$
 \rho_{(1)}(x;\alpha) = {1 \over \sqrt{2 \pi} \Gamma(1 + \alpha)} {1 \over | D_{-\alpha}(ix) |^2},
 $$
 where $D_\mu(z)$ denotes the parabolic cylinder function. Introduce now the orthogonal polynomials with
 respect to $ \rho_{(1)}(x;\alpha)$. These are the so called associated Hermite polynomials
 $\{ p_n^{\rm H}(x;\alpha) \}$, which can be generate through the three term recurrence
 $$
 p_{n+1}^{\rm H}(x;\alpha) = x  p_n^{\rm H}(x;\alpha) - (n + \alpha)  p_{n-1}^{\rm H}(x;\alpha),
 $$
 with $p_0^{\rm H}(x;\alpha) = 1$, $p_1^{\rm H}(x;\alpha) = x$. The case $\alpha = 0$ corresponds to the classical Hermite
 polynomials. In relation to a fluctuation formula, the recent work of Nakano, Trihn and Trinh \cite{NTT21}
 has shown that with $P_n^{\rm H}(x;\alpha) := \int^x p_n^{\rm H}(t;\alpha) \, dt$,
   \begin{equation}\label{G6p}   
   \lim_{N \to \infty} {1 \over N} {\rm Cov}^{\rm G}(P_m^{\rm H}, P_n^{\rm H}) \Big |_{w(x) = e^{-x^2/2} \atop \beta = 2 \alpha/N} =
   \delta_{m,n} {(\alpha + 1) \cdots (\alpha + n) \over n + 1}.
  \end{equation}

\subsection{The Laguerre and Jacobi $\beta$-ensembles}\label{S3.2}
The potential $V(x) = - (\alpha \log x - x)/2$,
$x \in \mathbb R^+$ substituted in (\ref{4.1c}) corresponds to the Laguerre $\beta$-ensemble.
In the limit $N \to \infty$ the normalised density limits to the Mar\v{c}enko--Pastur functional
form
$$
{ \sqrt{(x - c)( d - x)} \over 2 \pi (1 + \alpha) x},
$$
where $(c,d)$ is the interval of support with $c = (1 - \sqrt{1 + \alpha})^2$, $d = (1 + \sqrt{1 + \alpha})^2$. Although this is distinct from the
Wigner semi-circle functional form for the normalised eigenvalue density as holds for the
Gaussian $\beta$-ensemble, a loop equation analysis \cite{FRW17} gives for $W_{2,0}^{\rm L}$ a functional form
which includes the corresponding result (\ref{L2b}) for the Gaussian $\beta$-ensemble. Thus it is
found
\begin{equation}\label{L2c} 
 W_{2,0}^{\rm L}(x,y;\kappa = \beta/2) = {2 \over \beta} \bigg ( {xy - (c+d)(x+y)/2+cd  \over 2 (x - y)^2 \sqrt{((x-c)(x-d)(y-c)(y-d))}} - {1 \over 2 ( x - y)^2} \bigg ),
\end{equation}
supported on $x,y \in (c,d)$ (to reclaim (\ref{L2b}) set $c=-1$, $d=1$) as first identified in \cite{AM90}.
Consequently, by applying the inverse Stieltjes transform,
 \begin{equation}\label{r.1bX} 
   \lim_{N \to \infty} \rho_{(2), N}^T(x,y) \doteq - {1 \over \beta} {1 \over \pi^2} { (-cd + (c+d)(x+y)/2  - xy) \over (c - x) (d - x)
   (c - y) (d - y)}.
   \end{equation}
   Substituting (\ref{r.1bX}) in (\ref{3.2}) gives one particular functional form of the limiting covariance.
   More revealing is to change variables in the linear statistic $f(x)$ by writing
   $f(\alpha_1 + \alpha_2 \cos \theta)$ with $\alpha_1 = (c+d)/2$, $\alpha_2 = (d - c)/2$, and
   similarly for $g(x)$. In this new variable, performing a cosine expansion as on the
   RHS of (\ref{2.55}) then gives the simplified expression
  \begin{equation}\label{r.1bY}   
  \lim_{N \to \infty} {\rm Cov}^{\rm L} \Big ( \sum_{j=1}^N f(x_j),  \sum_{j=1}^N g(x_j) \Big ) = {2 \over \beta}
  \sum_{n=1}^\infty    n    f_n^{\rm c}  g_{n}^{\rm c};
 \end{equation} 
 cf.~the first expression in (\ref{4.1c}). 

\begin{remark}\label{R3.2}
1.~As observed in \cite{BF97f}, for the Laguerre $\beta$-ensemble it is simple to
compute the characteristic function for the linear statistic $\sum_{j=1}^N x_j$,
  \begin{equation}\label{r.1bZ} 
  \hat{P}_{N, f=x}^{\rm L}(t) = (1 - 2 it/N \beta)^{-N(1 + N \beta/2) - \beta N (N - 1)/2}.
 \end{equation}  
 This implies
   \begin{equation}\label{s.1bZx}   
   \lim_{N \to \infty} {\rm Var}^{\rm L} \Big ( \sum_{j=1}^N x_j \Big ) = {2 \over  \beta} ( \alpha + 1),
 \end{equation}  
 which is readily checked to be consistent with (\ref{r.1bY}). 
 Upon the recentring of replacing $  \hat{P}_{N, f=x}^{\rm L}(t)$ by  $\hat{P}_{N, f=x}^{\rm L}(t) e^{-i t \sum_{j=1}^N x_j}$,
  we see that
 the $N \to \infty$ form of (\ref{r.1bZ}) is a Gaussian. For linear statistics analytic in the neighbourhood of the eigenvalue
 support, the loop equation analysis of \cite{BG12} gives that the limiting recentred distribution is a Gaussian with variance
 determined by (\ref{r.1bY}).
 \\
 2.~Denote the limiting covariance for the monomial linear statistics $\sum_{j=1}^N x_j^{k_1}$, $\sum_{j=1}^N x_j^{k_2}$
in the Laguerre (Gaussian) cases by $\mu_{k_1, k_2}^{\rm L}$  ($\mu_{k_1, k_2}^{\rm G}$). Thus for $k_1, k_2$ of the same
parity $\mu_{k_1, k_2}^{\rm G}$ is given by (\ref{Gm}), while for $k_1, k_2$ of different parity $\mu_{k_1, k_2}^{\rm G} = 0$.
It is shown in  \cite{CMV15} that the covariance formula as obtained by substituting (\ref{r.1bX}) in (\ref{3.2}) implies
  \begin{equation}\label{s.1bW}  
  \mu_{k_1, k_2}^{\rm L} = \alpha_1^{k_1 + k_2} \sum_{p=0}^{k_1}   \sum_{q=0}^{k_2}  \binom{k_1}{p}   \binom{k_2}{q}  \Big ( {\alpha_2 \over \alpha_1} \Big )^{p+q} \mu_{p, q}^{\rm G} .
  \end{equation}   
  Also, in the special case $\alpha = 0$ so that $c=0$, $d=4$, $\alpha_1 = \alpha_2 = 2$ there is the simplification
   \cite{CMV15}
     \begin{equation}\label{s.1bX}
  \mu_{k_1, k_2}^{\rm L} \Big |_{\alpha = 0} = {1 \over \beta} 2^{k_1 + k_2 + 2} {1 \over k_1 + k_2} \binom{2k_1 - 1}{k_1}        \binom{2k_2 - 1}{k_2}.
  \end{equation}
  \end{remark}

  The Laguerre $\beta$-ensemble has been specified in terms of the eigenvalue
 PDF (\ref{4.1c}) with potential $V(x) = - (\alpha \log x + x)/2$, $x > 0$. It is a standard result in random matrix
 theory (see e.g.~\cite[\S 3.2]{Fo10}) that with $\alpha N = (n - N) + 1 - 2/\beta$ this eigenvalue PDF is realised
 by Wishart matrices $W = {1 \over N} X^\dagger X$, where $X$ is an $n \times N$ ($n \ge N$) standard real ($\beta = 1$)
 or complex ($\beta = 2$) Gaussian random matrix. 
 In a statistical setting the matrix $X$ 
  is the centred data matrix and  $W$ is proportional to the sample covariance. 
  More generally the Wishart class involves the centred data matrix having the form $X \Sigma^{1/2}$, where $X$ is as above
   and $\Sigma$ is an $N \times N$ positive definite matrix. We will use the symbol ${\rm L}^\boxtimes$ to indicate this setting.
   Let $1+\alpha_\infty = \lim_{M,N\to \infty} M/N$ and suppose $\Sigma$
   has a limiting eigenvalue density with compact support specified by the measure $d \nu(x)$.
   Specify $\tilde{m}(z)$ --- the Stieltjes transform of the limiting eigenvalue density of ${1 \over N} \Sigma^{1/2} X^\dagger X \Sigma^{1/2}$
   --- as the solution of 
   $$
   \alpha_\infty - z \tilde{m}(z) = (1 + \alpha_\infty) \int_{-\infty}^\infty {1 \over 1 + t \tilde{m}(z)} d \nu(t),
   $$
   which has ${\rm Im} \,  \tilde{m}(z) > 0$ for $z$ in the upper half complex plane; see e.g.~the text \cite{PS11} in relation to this result.
   Results of Bai and Silverstein \cite{BS04}, and further developed in \cite{LP09,Sh11,NY16,LSX21}, give
   \begin{equation}\label{s.1bY} 
   W_{2,0}^{{\rm L}^\boxtimes} = {1 \over \beta} {\partial^2 \over \partial x \partial y} \log \bigg (  { \tilde{m}(x) -  \tilde{m}(y) \over x - y} \bigg )
   \end{equation}
   and, with $I$ denoting the interval of support of the density,
   \begin{equation}\label{s.1bZ}  
   {\rm Cov}^{{\rm L}^\boxtimes}(f,g) = {1 \over \beta \pi^2} \int_Idx \int_I dy \,  f'(x) g'(y) \log \bigg | {   \tilde{m}(x)  - \overline{ \tilde{m}(y) } \over  \tilde{m}(x)  -  \tilde{m}(y) } \bigg |. 
   \end{equation} 
   
As for the Gaussian $\beta$-ensemble, the Laguerre $\beta$-ensemble admits a scaled high temperature limit.
This is specified by replacing 
$e^{-\beta N V(x_l)}$ in (\ref{4.1c}) by $x_l e^{- x_l}$, setting $\beta = 2 \alpha/N$ and taking $N \to \infty$.
 The normalised density, $\rho_{(1)}(x;\alpha_1,\alpha)$ say,  is then given in terms of the Whittaker function
 $W_{\zeta,\mu}(z)$ by  \cite{ABMV13} 
 $$
 \rho_{(1)}(x;\alpha_1\alpha) = {1 \over \Gamma(\alpha+1) \Gamma(\alpha + \alpha_1)}{1 \over | W_{-\alpha-\alpha_1/2,(1 + \alpha_1)/2}(-x) |^2}, \quad x > 0.
 $$
 The PDF for the corresponding mean $x \rho_{(1)}(x;\alpha_1,\alpha)/ (\alpha_1 + \alpha)$ can be recognised as the
 weight function for the associated Laguerre polynomials
 $\{ p_n^{\rm L}(x;\alpha_1,\alpha) \}$, defined by  the three term recurrence
 $$
 p_{n+1}^{\rm L} = x  p_n^{\rm L} - ( \alpha +2 \alpha_1 + 2n + 1)  p_{n}^{\rm L} - (\alpha + \alpha_1 + n) ( \alpha_1 + n) p_{n-1}^{\rm L},
 $$
 with $p_0^{\rm L}= 1$, $p_1^{\rm L} = x - (\alpha + 2 \alpha_1 + 1)$.
 A companion result to (\ref{G6p}) obtain in  \cite{NTT21} gives that
with $P_n^{\rm L}(x;\alpha_1,\alpha) := \int^x p_n^{\rm L}(t;\alpha_1,\alpha) \, dt$,
   \begin{equation}\label{L6p}   
   \lim_{N \to \infty} {1 \over N} {\rm Cov}^{\rm L}(P_m^{\rm L}, P_n^{\rm L}) \Big |_{w(x) = x^{\alpha_1}e^{-x} \atop \beta = 2 \alpha/N} =
   \delta_{m,n}  {\alpha + \alpha_1 \over n + 1} \prod_{j=1}^n ( \alpha_1 + j ) ( \alpha + \alpha_1 + j).
  \end{equation}

 We turn now to the consideration of the Jacobi $\beta$-ensemble. An example of a Jacobi $\beta$-ensemble is seen in 
 (\ref{4.1}). In this example the support is $(-1,1)$. Note that when written in the form (\ref{4.1c}) this gives a $V(x)$ proportional
 to $1/N$, which is sub-leading. We will consider instead a version of the Jacobi $\beta$-ensemble supported on $(0,1)$ with
 $V(x)$ of order unity and given explicitly by $V(x) = \gamma_1 \log x + \gamma_2 \log (1 - x)$.
 The corresponding $N \to \infty$ normalised density is known to be given by a functional form first identified by
 Wachter \cite{Wa78},
   \begin{equation}\label{Wa1}
   (\gamma_1 + \gamma_2+2) {\sqrt{(x - c^{\rm J}) ( d^{\rm J} - x)} \over 2 \pi x (1 - x)}
  \end{equation}   
  supported on $(c^{\rm J}, d^{\rm J})$ with these endpoints specified by
   \begin{equation}\label{Wa2}   
   \Bigg ( \sqrt{{\gamma_1 +1 \over \gamma_1 + \gamma_2 + 2} \Big ( 1 - {1 \over \gamma_1 + \gamma_2 + 2} \Big )} \pm
 \sqrt{{ 1 \over \gamma_1 + \gamma_2+2} \Big ( 1 - {\gamma_1 + 1 \over \gamma_1 + \gamma_2+2} \Big )}     \Bigg )^2.
  \end{equation}   
  
  In the present context the main point is that a loop equation analysis \cite{FRW17} confirms that the limiting
  functional form (\ref{L2c}) for the two-point quantity $W_{2,0}$ again holds true, as expected for all $\beta$-ensembles
  with density supported on a single interval. Hence the covariance is given by (\ref{r.1bY}). In the case of monomial linear statistics, the formula
  (\ref{s.1bW}) again holds, with  $\alpha_1 = (c^{\rm J}+d^{\rm J})/2$, $\alpha_2 = (d^{\rm J} - c^{\rm J})/2$.
 In the special case
  $\gamma_1 = \gamma_2 = 0$ the support is $(c^{\rm J},d^{\rm J}) = (0,1)$ and the covariance formula for 
 monomial linear statistics simplifies to be related to (\ref{s.1bX})     \cite{CMV15},
   \begin{equation}\label{Wa3}   
   \mu_{k_1, k_2}^{\rm J} \Big |_{\gamma_1 = \gamma_2 = 0} = 2^{-k_1 - k_2 }   \mu_{k_1, k_2}^{\rm L} \Big |_{\alpha = 0}.
   \end{equation}  
   
   \begin{remark}
   1.~For the Jacobi $\beta$-ensemble specified by (\ref{4.1}) with $\prod_{l=1}^N (1 - y_l^2)^{-1/2}$ replaced
   by $\prod_{l=1}^N (1 - y_l)^{\lambda_1}  (1 + y_l)^{\lambda_2}$, where $\lambda_1, \lambda_2 > - 1$ and
   fixed, the analogue of (\ref{G6d}) has been proved by Killip \cite{Ki08}, provided the interval $(a,b)$ has
   $a=-1$ or $b=1$. \\
   2.~A loop equation analysis has been applied to various discretisation of the classical $\beta$-ensembles
   \cite{BGG17,DK19}. In the so-called one cut regime, the universal form of $W_{2,0}$ as given
   by the RHS of (\ref{L2c}) is recovered. \\
   3.~All convex potentials $V(x)$ in (\ref{4.1c}) are known to lead to a one cut regime for the corresponding
   eigenvalue density. However, without this assumption, the eigenvalue density may consist of several
   intervals and the fluctuation formula for a linear statistic typically involves quasi-periodic terms \cite{Sh13}.
   \end{remark}

   \subsection{Wigner matrices}
   In the paragraph below (\ref{4.1c}) the GOE and GUE were defined in terms of real symmetric and complex Hermitian
   matrices, with elements on, and elements above, the diagonal independently and identically distributed as particular
   zero mean Gaussians. If the requirement of a Gaussian distribution is weakened to some other zero mean, finite variance
   distributions, the GOE and GUE generalise to what is termed the real symmetric and complex Hermitian Wigner
   matrices. Specifically, following \cite{BY05} it is assumed the variances are such $\langle |x_{ij}|^2\rangle = 1$ $(i < j)$ and
   $\langle x_{ii}^2\rangle = \sigma^2$. In this setting, the celebrated Wigner semi-circle law (see e.g.~\cite{PS11}) is equivalent to
   the result that, after scaling the matrices by $1/\sqrt{2N}$, the limit of (\ref{4.1e}) which we denote by $W_{1,0}^{\rm W}(x)$ is
   again given by (\ref{L2a}) and thus
     \begin{equation}\label{Wa4}   
W_{1,0}^{\rm W}(x) = 2 ( x - \sqrt{x^2 - 1}).
 \end{equation}  
 Note the independence on $\sigma^2$ and higher moments of the distribution of the entries. Generalising results obtained
 earlier by D'Anna and Zee \cite{DZ96} and Khorunzhy et al.~\cite{KKP96}, Bai and Yao \cite{BY05} have computed the
 scaled limit of the two-point quantity (\ref{4.1dz}) for Wigner matrices. As is consistent with the usage in (\ref{L2}), we denote
 this limiting quantity by $W_{2,0}^{\rm W}(x,y;\kappa)$, where $\kappa = 1/2$ (real case), $\kappa = 1$ (complex case).

 \begin{proposition}
 Let $\langle |x_{ij}|^4\rangle = 1$ $(i < j)$     be finite and independent of $(i,j)$. In the complex case, with
 $x_{ij} = x_{ij}^{\rm r} + i  x_{ij}^{\rm i}$ require that $\langle |x_{ij}^{\rm r} |^2\rangle = \langle |x_{ij}^{\rm i} |^2\rangle$.
 Subject only to a further technical condition
 on the decay of the tails of the distribution, one has
   \begin{multline}\label{Wa5}   
W_{2,0}^{\rm W}(x,y;\kappa) =     
     \Big ( {d \over dx} W_{1,0}^{\rm W}(x) \Big )   \Big ( {d \over dy} W_{1,0}^{\rm W}(y) \Big )  \\
     \times \bigg ( \sigma^2 - 1/\kappa + 2 \tilde{\beta}  W_{1,0}^{\rm W}(x)   W_{1,0}^{\rm W}(y)  +
     { (1 / \kappa) \over (1 -  W_{1,0}^{\rm W}(x)   W_{1,0}^{\rm W}(y) )^2 } \bigg ),
\end{multline}
where $\tilde{\beta} = \langle (| x_{12} |^2 - 1)^2 \rangle - 1/\kappa$.
\end{proposition}

\begin{remark}
1.~In the special case of the GOE we have $1/\kappa = \sigma^2 = 2$ and $\tilde{\beta} = 0$, while for the GUE we have
$1/\kappa = \sigma^2 = 1$ and $\tilde{\beta} = 0$. This implies that as an alternative to (\ref{L2b}), for $W_{2,0}^{\rm G}(x,y;\kappa)$
we have the functional forms
   \begin{align}\label{3.36}
W_{2,0}^{\rm G}(x,y;\kappa)     & = {2 \over \beta} \Big ( {d \over dx} W_{1,0}^{\rm W}(x;\kappa) \Big ) \Big ( {d \over dy} W_{1,0}^{\rm W}(y;\kappa) \Big ) 
{1 \over (1 -  W_{1,0}^{\rm W}(x;\kappa)  W_{1,0}^{\rm W}(y;\kappa) )^2} \nonumber \\
& = - {2 \over \beta} {\partial^2 \over \partial x  \partial y} \log (1 -   W_{1,0}^{\rm W}(x;\kappa)  W_{1,0}^{\rm W}(y;\kappa) ).
\end{align}
The second of these is known earlier from the work of Br\'ezin et al.~\cite{BHZ95}. \\
2.~The rewrite of one of the terms in (\ref{Wa5}) implied by (\ref{3.36}) can be extended to the remaining terms
   \begin{multline}\label{3.36a}
   \Big ( {d \over dx} W_{1,0}^{\rm W}(x) \Big )   \Big ( {d \over dy} W_{1,0}^{\rm W}(y) \Big )  \Big ( \sigma^2 - 1/\kappa + 2 \tilde{\beta}  W_{1,0}^{\rm W}(x)   W_{1,0}^{\rm W}(y) \Big )
   \\
   = {\partial^2 \over \partial x \partial y} \bigg (  ( \sigma^2 - 1/\kappa )  W_{1,0}^{\rm W}(x)   W_{1,0}^{\rm W}(y)  + {\tilde{\beta} \over 2} \Big ( 1 - 2 x  W_{1,0}^{\rm W}(x) \Big )
   \Big (  1 -  2 y  W_{1,0}^{\rm W}(y) \Big ) \bigg ),
\end{multline}
where use has been made of the quadratic equation satisfied by $W_{1,0}^{\rm W}$. An extension of the functional form (\ref{3.15b}), obtained by taking the
inverse Stieltjes transform of $W_{2,0}^{\rm W}(x,y;\kappa)$ so rewritten with respect to both $x$ and $y$, is the result
  \begin{multline}\label{3.36b} 
  \lim_{N \to \infty} \rho_{(2)}^{T, \rm W}(x,y)  \doteq  {1 \over 4 \pi^2} {\partial^2 \over \partial x \partial y}   \bigg ( {2 \over \beta}
   \log \bigg ( {1 - xy + \sqrt{(1 - x^2) (1 - y^2)} \over 1 - xy - \sqrt{(1 - x^2) (1 - y^2)} } \bigg )  \\ + ( \sigma^2 - (2/\beta) + 2 \tilde{\beta} x y) 4
   (1 - x^2)^{1/2}  (1 - y^2)^{1/2} \bigg ).
   \end{multline}
   Notice the separation of variables in the terms which differ from the Gaussian result.
   This substituted in the limiting form of (\ref{3.2}) implies \cite{BY05}
     \begin{equation}\label{3.36c} 
    \lim_{L \to \infty}  \lim_{N \to \infty} {\rm Cov}^{\rm W} \, \Big ( \sum_{j=1}^N f(x_j),  \sum_{j=1}^N g(x_j) \Big ) = {1 \over 4 \pi^2} \int_{-1}^1 dx \, f'(x)  \int_{-1}^1 dy \, g'(y) V(x,y),
    \end{equation} 
    where $V(x,y)$ is the functional form in  (\ref{3.36b}) after the partial derivatives.
    Moreover, with $f_n^{\rm c}, g_n^{\rm c}$ as in (\ref{4.1a}) it is shown in \cite{BY05} that the RHS of (\ref{3.36c}) can be rewritten in
    the form
     \begin{equation}\label{3.36d}  
     (\sigma^2 - (2/\beta)) f_1^{\rm c}  g_1^{\rm c} +2 \tilde{\beta}  f_2^{\rm c}  g_2^{\rm c}  + {2 \over \beta} \sum_{l=1}^N l f_l^{\rm c}  g_l^{\rm c}.
    \end{equation}    
    We note too the work of Shcherbina \cite{Sh11} for an independent evaluation of the limit in (\ref{3.36c}) in the real case for $f=g$, which gives a
    functional form generalising the final equality in (\ref{4.1d}). A comprehensive study of conditions on $f$ for which this formula is valid has recently
    been given by Landon and Sosoe \cite{LS22}; this work also reviews earlier work along these lines as part of the Introduction section.
    In the complex case Bao and Xie \cite{BX16} remove the requirement of \cite{BY05}
    that $\langle |x_{ij}^{\rm r} |^2\rangle = \langle |x_{ij}^{\rm i} |^2\rangle$; the covariance formula now depends on the
    parameter $\Phi := \langle |x_{ij}^{\rm r} |^2\rangle  -  \langle |x_{ij}^{\rm i} |^2\rangle$,
    \begin{equation}\label{3.36d1}  
     \sigma^2  f_1^{\rm c}  g_1^{\rm c} +2 ( \langle |x_{ij} |^4\rangle -\Phi^2-2)
      f_2^{\rm c}  g_2^{\rm c}  +  \sum_{l=2}^N l (1 + \Phi^l) f_l^{\rm c}  g_l^{\rm c}.
    \end{equation}     
    Note that when $\Phi = 1$, which corresponds to real Wigner matrices, this is consistent with
     (\ref{3.36d}) for $\beta = 1$.
     \\
    3.~In relation to the proof of the central limit theorem associated with a linear statistic for Wigner matrices,
    for which (\ref{3.36c}) implies the variance,
    the recent work \cite{BH21} highlights the strategy introduced in \cite{KKP96} as being particularly influential.
    Denote $G = (x \mathbb I - H)^{-1}$ as the resolvent of the Wigner matrix $H$, so that ${\rm Tr} \, G$ is equal to
    (\ref{4.1e}). The corresponding matrix elements then satisfy the simple identity
    $$
    G_{jm} = - x^{-1} \delta_{j,m} + x^{-1} \sum_{k=1}^N G_{jk} H_{km}.
    $$
   To average over the distribution of the entries of $H$, in the Gaussian case use can be made of the identity
   $$
   \langle G_{jk} H_{km} \rangle =   \langle  H_{km}^2 \rangle   \Big \langle  {\partial \over \partial H_{km}}G_{jk}  \Big \rangle
  $$
  as follows from   
   \begin{equation}\label{3.36e}  
 \langle \xi f(\xi) \rangle =  \langle \xi^2  \rangle \langle f'(\xi) \rangle.
  \end{equation}   
  As a replacement to (\ref{3.36e}) in the case of distributions outside the Gaussian class, it is proposed in
  \cite{KKP96} to make use of the particular cumulant expansion
  \begin{equation}\label{3.36f}   
   \langle \xi f(\xi) \rangle =  \sum_{l=0}^p {\kappa_{l+1} \over l!}    \langle f^{(p)}(\xi) \rangle  + R_{p+1}.
 \end{equation}   
 Here $\{\kappa_{l+1} \}_{l=0,1,\dots}$ refers to the cumulants of the distribution of $\xi$, $f^{(l)}$ denotes
 the $l$-th derivative of $f$, and $R_{p+1}$ is a remainder term which can be bounded in terms of $f^{(l+1)}$.
 In     \cite{BH21} (\ref{3.36f}) is attributed to Barbour \cite{Ba86}. \\
 4.~The results for the covariance of linear statistics for the Gaussian external source model of Remark
 \ref{R3.1}.5, and the Gaussian sample covariance matrices of Remark  \ref{R3.2}.3, have been
 generalised to Wigner matrices. In fact the references cited to arrive at (\ref{G6c}) and
 (\ref{s.1bY}) are formulated in this more general setting. \\
 5.~Let the independent upper triangular diagonals of an Hermitian matrix be labelled $d=1$ (main diagonal), $d=2$ (first
 diagonal above the main diagonal), etc. A band Hermitian matrix has all such independent diagonals $d > d^*$ for some
 $d^*$ with all entries equal to zero. In the case that the entries in the diagonals $d=1,\dots,d^*$ are random and as for Wigner matrices,
 and $d^*$ is dependent on $N$ such that $d^* \to \infty$, $d^*/N \to 0$ as $N \to \infty$ a generalisation of the covariance
 formula (\ref{3.36c}) has been derived in \cite{Sh15a, JSS16}.
 \end{remark} 
 
 \subsection{Singular values of random matrix products}
 The singular values of an $n \times N$ $(n \ge N)$ matrix $X$ are the eigenvalues of $X X^\dagger$. In the case $X = G_1 G_2$ with
 each $G_i$ an independent GinUE matrix, a loop equation analysis of the square singular values has been carried out in
 \cite{DF20}. The already known result \cite{PZ11} that the limiting resolvent satisfies the cubic equation
   \begin{equation}\label{3.56a} 
 x^2 (W_{1,0}^{G^2}(x))^3 - x W_{1,0}^{G^2}(x) + 1 = 0
   \end{equation}  
 was recovered, and the limiting second order resolvent $W_{2,0}^{G^2}$ was expressed in terms of $W_{1,0}^{G^2}$.
 Analogous to the results (\ref{Gm}) and (\ref{s.1bX}) it was found from this that the corresponding limiting covariance of the monomial
 statistics $\sum_{j=1}^N x^{k_1}_j$, $\sum_{j=1}^N x^{k_2}_j$, to be denoted $\mu_{k_1, k_2}^{G^2}$ say, has the explicit evaluation
   \begin{equation}\label{3.37}   
\mu_{k_1, k_2}^{G^2} = {2 k_1 k_2 \over 3 (k_1 + k_2)} \binom{3k_1}{k_1}     \binom{3k_2}{k_2}.
  \end{equation}  
  In the case of a product of $M$ independent GinUE matrices, Gorin and Sun \cite{GS22} have given a double
  integral formula for ${\rm Cov} (p_j(x), p_k(x))$ which however appears to be difficult to evaluate. This is similarly
  true of the formula for the variance of a more general, not necessarily polynomial, linear statistics in
  the case of the product of two real Wigner matrices given in \cite{GNT17}.
  
 In the case of product of complex, rectangular Ginibre matrices
  the work of Lambert \cite{La18}
  does provide an easy to evaluate single contour integral formula for the variance of
  a polynomial linear statistics.  This work is based on
  special properties of the biorthogonal functions underpinning
  integrability of the singular values of the products
  \cite{AIK13}. To state the result, let $G_j$ be an $N_j \times N_{j-1}$ rectangular GinUE
  matrix and consider the squared singular values of the product $W_N=G_M G_{M-1} \cdots G_1$
  where $N_0 = 1$ and $N_j = N + \eta_j$ with $\eta_j \ge 0$. Divide the squared singular
  values by $\prod_{j=1}^M N_j$. Reading off from \cite[Th.~4.2]{La18}, the following fluctuation
  formula holds true.
  
  \begin{proposition}\label{P3.6}
  In terms of the above notation, suppose $N/N_j \to \gamma_j \in [0,1]$ and $N \to \infty$. Then the
  variance of the polynomial linear statistic of the squared singular values  $\sum_{j=1}^N p(x_j)$
   is given by
   \begin{equation}\label{3.37a}   
   \sum_{k=1}^\infty k C_k C_{-k},
    \end{equation}  
    where
  \begin{equation}\label{3.37b}   
   C_k = {1 \over 2 \pi i} \oint      p \Big ( z^{-M} \prod_{l=0}^M (z + \gamma_l) \Big ) z^{-k} {dz \over z}.
   \end{equation}
 Moreover, the limiting distribution of this linear statistic is a Gaussian.
 \end{proposition}
 
 Specifically, in the case of a product of two square GinUE matrices, and with $p(x) = x^l$
 we have that (\ref{3.37a}) reduces to
  \begin{equation}\label{3.37c}   
 \sum_{k=1}^l  k \binom{3l}{2l+k} \binom{3l}{2l-k} = {l \over 3} \binom{3l}{l}^2,
  \end{equation}  
 where the value of the sum has been obtained using computer algebra. This is in agreement with the case $k_1 = k_2$
 of (\ref{3.37}).
 
 Closely related to the squared singular values of random complex GinUE matrices is the Laguerre Muttalib--Borodin model
 eigenvalue PDF in the variables $x_l \mapsto x_l^{1/M}$, proportional to
   \begin{equation}\label{3.37d} 
   \prod_{l=1}^N x_l^c e^{- x_l^{1/M} }\prod_{1 \le j < k \le N} (x_k - x_j) (x_k^{1/M} - x_j^{1/M}), \qquad x_l > 0.
  \end{equation}  
  For example, when $M = 2$ the scaled limiting resolvent satisfies (\ref{3.56a}) \cite{FW17}. After scaling $x_l \mapsto x_l/N$, it
  has been shown in \cite{La18} that the formula of    Proposition \ref{P3.6} remains true with $\gamma_0 = \gamma_1 = \cdots = \gamma_M = 1$.

 \subsection{Global scaling of Ginibre matrices and generalisations}
 The eigenvalue PDF for GinUE is proportional to
  \begin{equation}\label{5.1}
  \prod_{l=1}^N e^{-\beta | z_l|^2/2} \prod_{1 \le j < k \le N} | z_k - z_j|^\beta
  \end{equation}   
  with $\beta = 2$; see e.g.~\cite[\S 15.1.1]{Fo10}. As in \S \ref{S2.7} the parameter $\beta$ has the interpretation of inverse temperature
  in a particular equilibrium statistical mechanics analogy. The latter relates to a system of $N$ particles in two-dimensions with
  potential energy
  $$
  U = {1 \over 2} \sum_{l=1}^N | z_l|^2/2 - \sum_{1 \le j < k \le N}  \log | z_k - z_j|,
  $$
which up to an additive constant is realised by the two-dimensional one-component plasma model (2dOCP) of $N$ log-potential
unit charges in the presence of a disk of radius $\sqrt{N}$  containing a uniform smeared out neutralising background.

The global scaling limit corresponds to the replacement $z_l \mapsto \sqrt{N} z_l$. Then, for all $\beta > 0$, the
density is the uniform distribution on the unit disk, as can be established by potential theoretic considerations
\cite{Se16,Ch21}. From a random matrix viewpoint, this latter feature is an example of the circular law \cite{BC12}, which tells
us that for non-Hermitian random matrices with identically distributed, zero mean and finite variance entries, the eigenvalue
density in the global scaling limit is uniform inside a disk.

In relation to fluctuation formulas, with
 \begin{equation}\label{5.2}
 C_N^{\rm 2dOCP}(\mathbf r - \mathbf r') = \rho_{(2),N}^{\rm 2dOCP}(\mathbf r - \mathbf r', \mathbf 0 ) - {N^2 \over \pi^2} + {N \over \pi} \delta(\mathbf r - \mathbf r')
  \end{equation} 
  (cf.~(\ref{5.4ii}); note that $N / \pi$ corresponds to the density inside the unit disk, assuming global scaling) write
 \begin{equation}\label{5.2a}  
 S_\infty^{\rm 2dOCP}(\mathbf k) = \int_{\mathbb R^2} \Big ( \lim_{N \to \infty} C_N^{\rm 2dOCP}(\mathbf r, \mathbf 0) \Big ) e^{i \mathbf k \cdot \mathbf r} \,
 d \mathbf r.
  \end{equation} 
 Linear response arguments \cite{Ja95} predict that
  \begin{equation}\label{5.2b}     
  S_\infty^{\rm 2dOCP}(\mathbf k) = {|\mathbf k|^2 \over 2 \pi \beta}
  \end{equation} 
  (note the consistency with (\ref{5.4f}) in the case $\beta = 2$), or equivalently
    \begin{equation}\label{5.2c}  
  \lim_{N \to \infty} C_N^{\rm 2dOCP}(\mathbf r, \mathbf r')  = - {1 \over 2 \pi \beta} \nabla^2 \delta (\mathbf r - \mathbf r' ).
   \end{equation}    
   The validity of (\ref{5.2c}) is restricted to $\mathbf r,  \mathbf r'$ strictly inside of the unit disk. On the boundary,
   different linear response arguments predict \cite{Ja95}, \cite[\S 15.4.3]{Fo10} predict
   \begin{equation}\label{5.2d}  
    \lim_{N \to \infty} C_N^{\rm 2dOCP}(\mathbf r, \mathbf r')  = - {1 \over 2 \pi^2 \beta} \bigg ( {\partial^2 \over \partial \theta_1   \partial \theta_2}
    \log \Big | \sin(\theta_1 - \theta_2)/2 \Big | \bigg ) \delta(r_1 - 1)  \delta(r_2 - 1).
   \end{equation} 
   The two results (\ref{5.2c}) and    (\ref{5.2d}) together, substituted in (\ref{3.1}), predict \cite{Fo99}
  \begin{equation}\label{5.2e} 
  \lim_{N \to \infty} {\rm Cov}^{\rm 2dOCP} \Big ( \sum_{j=1}^N f(\mathbf r_j),  \sum_{j=1}^N g(\mathbf r_j) \Big ) =
  {1 \over 2 \pi \beta} \int_{ |\mathbf r < 1}   \nabla f \cdot  \nabla g \, dx dy +
  {1 \over \beta} \sum_{n=-\infty}^\infty |n| f_n g_{-n},
   \end{equation}   
where $f_n, g_n$ are the angular Fourier components of $f(\mathbf r) |_{|\mathbf r| = 1},   g(\mathbf r) |_{|\mathbf r| = 1}$ (cf.~(\ref{3.4d})).  

In the case that $f(\mathbf r) = f( | \mathbf r |)$ or $g(\mathbf r) = g( | \mathbf r |)$ the second term in (\ref{5.2e}) vanishes.
Further setting $f=g$ in the GinUE case $\beta = 2$ it is simple to compute the limiting characteristic function for the linear
statistic $\sum_{j=1}^N f( | \mathbf r_j |)$ using the Vandermonde determinant form of $\prod_{1 \le j < k \le N} (z_k - z_j)$;
see e.g.~\cite[Eq.~(1.173)]{Fo10}. This calculation shows \cite{Fo99}
  \begin{equation}\label{3.69a} 
 \lim_{N \to \infty} {\rm Var}^{\rm GinUE} \Big ( \sum_{j=1}^N f(|\mathbf r_j|) \Big ) = {1 \over 2} \int_0^1 r ( f'(|\mathbf r|) )^2 \, dr,
  \end{equation}   
 which is consistent with the appropriate specialisation of  (\ref{5.2e}). Proofs of (\ref{5.2e}), and the underlying Gaussian
 fluctuation formula, have been given in \cite{RV07,AHM15} for $\beta = 2$ and in \cite{LS18} for general $\beta > 0$.
 
 A generalisation of (\ref{5.1}) is the PDF proportional to 
  \begin{equation}\label{5.1x}
  \prod_{l=1}^N \exp \Big ( - {\beta \over 2} \sum_{j=1}^N \Big ( {x_j^2 \over 1 + \tau} + {y_j^2 \over 1 - \tau} \Big ) \Big ) \prod_{1 \le j < k \le N} | z_k - z_j|^\beta, \quad 0 \le \tau < 1.
  \end{equation}   
  After scaling $z_j \mapsto \sqrt{N} z_j$ the leading order support is an ellipse with semi-axes $A= 1 + \tau$, $B = 1 - \tau$
  \cite{DGIL94,FJ96}. In the case $\beta = 2$ (\ref{5.1x}) corresponds to the eigenvalue PDF of the complex nonsymmetric
  random matrices $J = H + i v A$ \cite{FKS97}. Both $H$ and $A$ are Gaussian Hermitian random matrices, and 
  with $X = H, A$ and $\tau = (1 - v^2)/(1 + v^2)$, have joint PDFs for their
  elements proportional to $\exp \Big ( - {1 \over 1 + \tau} {\rm Tr} X^2 \Big )$.  Parametrising the boundary of the ellipse by
   \begin{equation}\label{5.1y}
   x + i y = \cosh (\xi_b + i \eta), \qquad 0 \le \eta < 2 \pi, \: \: \tanh \xi_b = (1 - \tau)/(1 + \tau)
   \end{equation}    
   the only modification of (\ref{5.2e}) required is to replace $f_n,g_n$ therein by the angular Fourier components
 in $\eta$ of $f(\mathbf r), g(\mathbf r)$ on the boundary (\ref{5.1y}) \cite{Fo99}. For the particular linear statistic
 $f(\mathbf r) = c_{10} x + c_{01} y$, completing the square gives for the characteristic function
   \begin{equation}\label{5.1z} 
   \hat{P}_{N,f}(t) =  e^{- t^2 ( c_{10}^2 (1 + \tau) + c_{01}^2 (1 - \tau))/(2\beta)}.
   \end{equation}
  Hence the variance is equal to $(c_{10}^2 (1 + \tau) + c_{01}^2 (1 - \tau) )/(2 \beta)$, which indeed is consistent
  with the specified modification of (\ref{5.2e}). In the case $\beta = 2$ a derivation of the elliptic analogue of
  (\ref{5.2d}) is possible \cite{FJ96,AC21,ADM22}.

 Generalising GinUE matrices to $N \times N$ complex matrices having general i.i.d.~complex entries $z_{ij}$ with
 $\langle z_{ij} \rangle = \langle z_{ij}^2 \rangle = 0$, $ \langle |z_{ij} |^2 \rangle = 1$ gives the analogue of the
 Wigner class of complex Hermitian matrices. It is shown in the work of Cipolloni, Erd\"os and Schr\"oder \cite{CES21}
 that the covariance formula (\ref{5.2e}) with $\beta = 2$ is then extended to include the additional term
  \begin{equation}\label{5.2f} 
  \kappa_4 \Big ( {1 \over \pi} \int_{|z|<1} f(\mathbf r) \, dx dy - {1 \over 2 \pi} \int_0^{2 \pi} f(\theta) \, d \theta \Big )
  \Big ( {1 \over \pi} \int_{|z|<1} g(\mathbf r) \, dx dy - {1 \over 2 \pi} \int_0^{2 \pi} g(\theta) \, d \theta \Big ),
  \end{equation}
  where $\kappa_4 :=  \langle |z_{ij} |^4 \rangle - 2$ is the fourth cumulant of the distribution of the entries. Moreoever,
  the work \cite{CES21} places particular emphasis on the class of test functions for which the covariance formula
  has a rigorous proof, and provides an extensive list, and discussion, of previous literature.
  
  Modifying GinUE to have standard real rather than standard complex entries gives what we will refer to as the Ginibre
  orthogonal ensemble (GinOE). Unlike the circumstance for the GOE and GUE, where (\ref{4.1c}) with $V(x) = x^2$ is the
  functional form for the eigenvalue PDF of both, depending on the value of $\beta$, this is not the case for GinOE
  in relation to (\ref{5.1}). In fact the eigenvalue PDF for GinOE is not absolutely continuous, and is naturally broken
  into sectors, depending on the number of real eigenvalues \cite{Ed97}. Furthermore, the complex eigenvalues
  for GinOE must come in complex conjugate pairs. Despite these differences, as first found in \cite{Ko16,OR16}, the
  fluctuation formula (\ref{5.2e}) with $\beta = 1$ does give the correct form of the covariance, where it is being assumed
  that both $f$ and $g$ are symmetric about the real axis. In the real analogue of the more general setting discussed
  in the previous paragraph the same term (\ref{5.2f}) is to be added \cite{CES21a}.
  
  \begin{remark}
  1.~The fluctuation formula (\ref{5.2e}) with $\beta = 2$ has been shown to remain valid in the case of the eigenvalues of products of complex
  Ginibre matrices in \cite{KOV20}. \\
  2.~The recent work \cite{Ch21a} gives a generalisation of (\ref{3.69a}) to the case that $f$ is discontinuous.
  \end{remark}


\subsection*{Acknowledgements}
	This research is part of the program of study supported
	by the Australian Research Council Discovery Project grant DP210102887.
I thank L.~Erd\"os for correspondence in relation to \cite{CES21}.


\begin{thebibliography}{10}

\bibitem{AGL21}
K.~Adhikari, S.~Ghosh and J.L.~Lebowitz,
\emph{Fluctuation and entropy in spectrally constrained random fields},
Commun. Math. Phys. \textbf{386}, (2021) 749--780.



\bibitem{ADM22}
G. Akemann, M.~Duits and L.D.~Molag,
\emph{The elliptic Ginibre ensemble: a unifying approach tolLocal and global statistics
for higher dimensions}, arXiv:2203.00287.


\bibitem{AIK13}
G. Akemann, J. R. Ipsen, and M. Kieburg, 
\emph{Products of rectangular random matrices: singular values and progressive scattering.} 
Phys. Rev. E \textbf{88}, (2013) 052118.

  \bibitem{ABG12}
R.~Allez, J.P.~Bouchard and A.~Guionnet, \emph{Invariant beta ensembles and the Gauss-Wigner
crossover}, Phys. Rev. Lett. \textbf{109} (2012), 09412.

\bibitem{ABMV13}
R. Allez, J.-P. Bouchaud, S. N. Majumdar, P. Vivo, \emph{Invariant $\beta$-Wishart ensembles, crossover densities and asymptotic corrections to the Marchenko-Pastur law},
 J. Phys. \textbf{46}, 015001 (2013).



\bibitem{AM90}
J.~Ambj{\o}rn and Yu.M. Makeenko, \emph{Properties of loop equations for the
  {H}ermitian matrix model and for two-dimensional quantum gravity}, Mod.
  Phys. Lett. \textbf{5} (1990), 1753--1763.
  
\bibitem{AC21}  
Y. Ameur, and J. Cronvall, \emph{Szeg\"o type asymptotics for the reproducing kernel in spaces of full-plane weighted polynomials},
arXiv:2107.11148v2. 
  
  
\bibitem{AHM15}  
  Y. Ameur, H. Hedenmalm, and N. Makarov, \emph{Random normal matrices and Ward identities},
  Ann. Probab., \textbf{43} (2015), 1157--1201.
  
  
  \bibitem{AGZ09}
G.W. Anderson, A.~Guionnet, and O.~Zeitouni, \emph{An introduction to random
  matrices}, Cambridge University Press, Cambridge, 2009.
  
  
 \bibitem{BF97f}
T.H.~Baker and P.J.~Forrester, \emph{Finite-$N$ fluctuation formulas for random matrices}, J. Stat.
  Phys. \textbf{88} (1997), 1371--1386. 
  
  
 \bibitem{BS04}  
  Z.D.~Bai and J.W.~Silverstein,  \emph{CLT for linear spectral statistics of large dimensional sample covariance matrix},
   Ann. Probab.  \textbf{32}, 553--605.
  
 \bibitem{BY05}  
  Z.D.~Bai and J.~Yao,  \emph{On the convergence of the spectral empirical process of Wigner matrices},  Bernoulli \textbf{11} (2005), 1059--1092.
  
 \bibitem{BH21}  
  Z. Bao and Y. He, \emph{Quantitative CLT for linear eigenvalue statistics of Wigner matrices},
arXiv:2103.05402.

 \bibitem{BPZ13}
Z.~Bao,  G.~Pan and W.~Zhou,  \emph{Central limit theorem for partial linear
eigenvalue statistics of Wigner matrices} J. Stat. Phys. \textbf{150} (2013), 88--129.





 \bibitem{BS12}
Z. Bao and Z. Su, \emph{Local Semicircle law and Gaussian fluctuation for Hermite $\beta$-ensemble},
Scientia Sinica Math. \textbf{42} (2012), 1017--1030.

 \bibitem{BX16}
Z. Bao and J. Xie, \emph{CLT for linear spectral statistics of Hermitian Wigner matrices with general moment
conditions}, Theory Probab. Appl., \textbf{60} (2016), 187--206.


 \bibitem{Ba86}
A.D.~Barbour,  \emph{Asymptotic expansions based on smooth functions in the central limit theorem} Prob. Th.
Related Fields \textbf{72} (1986), 289--303.
  

\bibitem{Ba12}
E. Basor,  \emph{A Brief History of the Strong Szeg\"o Limit Theorem}, Oper. Theor. Advan. Appl.
\textbf{222} (2012), 73--83.

\bibitem{BW99}
E. Basor and H. Widom, \emph{Determinants of Airy operators and applications to random
matrices},  J. Statist. Phys. \textbf{96} (1999), 1--20.


\bibitem{Be93}
C.W.J. Beenakker, \emph{Random-matrix theory of mesoscopic fluctuations in
  conductors and superconductors}, Phys. Rev. B \textbf{47} (1993),
  15763--15775.
  
\bibitem{BLS18}  
  F. Bekerman, T. Lebl\'e, and S. Serfaty.  \emph{CLT for fluctuations of $\beta$-ensembles with general
potential},  Electron. J. Probab., \textbf{23} (2018), 115.

\bibitem{Be08}
M. Bender,  \emph{Global fluctuations in general $\beta$ Dyson's Brownian motion}, Stochastic Process. Appl. \textbf{118} (2008), 1022--1042.
  

\bibitem{BD17}
T. Berggren and M. Duits, \emph{Mesoscopic fluctuations for the thinned circular unitary
ensemble}, Math. Phys. Anal. Geom.  \textbf{20} (2017), 19.

\bibitem{BGS84}
O.~Bohigas, M.J. Giannoni, and C.~Schmit, \emph{Characterization of chaotic
  quantum spectra and universality of level fluctuation laws}, Phys. Rev. Lett.
  \textbf{52} (1984), 1--4.
  
  \bibitem{BP04} O.~Bohigas and M.P.~Pato, \emph{Missing levels in correlated spectra},
  Phys.~Lett.~B \textbf{595} (2004), 171--176.
  
  
  \bibitem{BC12}
  C. Bordenave and D. Chafa\"i,  \emph{Around the circular law},  Probability Surveys  \textbf{9} (2012), 1--89.
  
  \bibitem{BFM17}
F.~Bornemann, P.J. Forrester, and A.~Mays, \emph{Finite size effects for
  spacing distributions in random matrix theory: circular ensembles and riemann
  zeros}, Stud. Appl. Math. \textbf{138} (2017), 401--437.
 
   \bibitem{Bo14} 
  A.~Borodin, \emph{CLT for spectra of submatrices of Wigner random matrices},
Moscow Mat. J. \textbf{14} (2014), 29--38.  

  \bibitem{Bo15} 
 A. Borodin. Gaussian free fields in $\beta$-ensembles and random surfaces. Lecture: Clay
Mathematics Institute, online resource, 2015.

  \bibitem{BG15}
A. Borodin and V. Gorin, \emph{General $\beta$-Jacobi corners process and the Gaussian free field},  Comm.Pure Appl. Math, \textbf{68} (2015),
1774--1844.


 \bibitem{BGG17}
A.~Borodin,  V.~Gorin and A.~Guionnet,  \emph{Gaussian asymptotics of discrete $\beta$-ensembles}, Publ. Math.
IHES \textbf{125} (2017), 1--78.

  
  \bibitem{BS09}
A.~Borodin and C.D. Sinclair, \emph{The {G}inibre ensemble of real random
  matrices and its scaling limit}, Commun. Math. Phys. \textbf{291} (2009),
  177.
  

  \bibitem{BG12}
G.~Borot and A.~Guionnet, \emph{Asymptotic expansion of $\beta$ matrix models
  in the one-cut regime}, Commun. Math. Phys. \textbf{317} (2013), 447--483.
  
  
 \bibitem{BMP22} 
      P.~Bourgade, K.~Mody and M.~Pain, \emph{Optimal local law and central limit theorem for 
$\beta$-ensembles},  
Comm.  Math. Phys. \textbf{390} (2022), 1017--1079.
  
  
  \bibitem{BK99}  
  A. Boutet de Monvel and A. Khorunzhy, \emph{Asymptotic distribution of smoothed
eigenvalue density. I. Gaussian random matrices},  Random Oper. Stochastic
Equations, \textbf{7} (1999), 1--22.
  
  \bibitem{BH16}
   E.~Br\'ezin and S. Hikami,  \emph{Random matrix theory with an external source} Springer briefs in
   mathematical physics vol.~19,  2016.
   
    \bibitem{BHZ95}
   E. Br\'ezin, S. Hikami, and A. Zee,  \emph{Universal correlations for deterministic plus random
hamiltonians}, Phys. Rev. E 51 (1995), 5442.
  
     \bibitem{BMS11}
A.~Brini, M.~Mari\~{n}o, and S.~Stevan, \emph{The uses of the refined matrix
  model recursion}, J. Math. Phys. \textbf{52} (2011), 35--51.
  
\bibitem{BG18}  
  A. Bufetov, V. Gorin,  \emph{Fluctuations of particle systems determined by Schur
generating functions}, Adv. Math., \textbf{338}  (2018), 702--781.

\bibitem{CD01}
T.~Cabanal-Duvillard, \emph{Fluctuations de la loi empirique de grandes matrices al'eatoires},
 Ann. Inst. H. Poincar\'e Probab. Statist.  \textbf{37} (2001), 373--402.
 
\bibitem{Ch21a} 
 C. Charlier, Asymptotics of determinants with a rotation-invariant weight and discontinuities along circles,
arXiv:2109.03660.
 
 \bibitem{Ch21}
 D.~Chafa\"i, \emph{Aspects of Coulomb gases}, arXiv:2108.10653.
 
 
\bibitem{CL98} 
 Y.~Chen and N.~Lawrence,  \emph{On the linear statistics of Hermitian random matrices}, J.
Phys. A \textbf{34} (1998), 1141--1152.

\bibitem{CE20}
G.~Cipolloni, L.~Erd\"os, \emph{Fluctuations for differences of linear eigenvalue statistics for sample covariance matrices},
Rand. Mat. Theory  Appl. \textbf{9} (2020), 2050006. 

\bibitem{CES21}
G.~Cipolloni, L.~Erd\"os and D.~Schr\"oder
\emph{Central limit theorem for linear eigenvalue statistics of
non-Hermitian random matrices}, 
Commun. Pure Applied Math.,  (2021), https://doi.org/10.1002/cpa.22028. 


\bibitem{CES21a}
G. Cipolloni, L. Erd\"os and D. Schr\"oder, \emph{Fluctuation around the circular law for
random matrices with real entries}, Electron. J. Prob., \textbf{26} (2021), 1--61.
  
 
 \bibitem{CGMY21}
 T.~Claeys, G.~Glesner, A.~Minakov, M.~Yang,  \emph{Asymptotics for averages over classical orthogonal ensembles}, 
 International Mathematics Research Notices (2021) rnaa354, https://doi.org/10.1093/imrn/rnaa354 
 
 
 
 \bibitem{CMV15} 
 F.D.~Cunden,  F.~Mezzadri and P.~Vivo, \emph{A unified fluctuation formula for one-cut $\beta$- ensembles of random
matrices}, J. Phys. A  \textbf{48}  (2015), 315204.

 \bibitem{CV14}
F.D. Cunden and P. Vivo
\emph{Universal covariance formula for linear statistics on random matrices}
Phys. Rev. Lett., \textbf{113} (2014), 070202.
  
   \bibitem{DF19}
S.~Dallaporta and  M.~Fevrier, \emph{Fluctuations of linear spectral statistics of deformed Wigner matrices}, 
arXiv :1903.11324. 

 \bibitem{DF20}
 S.~Dartois, P.J.~Forrester,
\emph{Schwinger--Dyson and loop equations for a product of square Ginibre random matrices},
J.~Phys.  A   \textbf{53}, 175201.


  
  \bibitem{DE01}
  P.~Diaconis and S.~Evans, \emph{Linear Functionals of Eigenvalues of Random Matrices}, Trans.
Amer. Math. Soc. \textbf{353} (2001), 2615--2633.

\bibitem{DGIL94}
P.~Di Francesco, M.~Gaudin, C.~Itzykson, and F.~Lesage, \emph{Laughlin's wave
  functions, {C}oulomb gases and expansions of the discriminant}, Int. J. Mod.
  Phys. A \textbf{9} (1994), 4257--4351.


 \bibitem{DK19}
E. Dimitrov and A. Knizel,  \emph{Log-gases on quadratic lattices via discrete loop equations and
q-boxed plane partition}, J. Funct. Anal., \textbf{276} (2019), 3067--3169.


\bibitem{DFK21}
M.~Duits, B.~Fahs and R.~Kozhan, \emph{Global fluctuations for multiple orthogonal polynomial ensembles},
J. Funct.  Analysis
\textbf{281} (2021), 109062.

  \bibitem{DE05}
I.~Dumitriu and A.~Edelman, \emph{Global spectrum fluctuations for the $\beta$-{H}ermite and
  $\beta$-{L}aguerre ensembles via matrix models}, J. Math. Phys. \textbf{47}
  (2006), 063302.

\bibitem{Dy62}
F.J. Dyson, \emph{Statistical theory of energy levels of complex systems {I}},
  J. Math. Phys. \textbf{3} (1962), 140--156.
  
  
  
  \bibitem{Dy62a}
F.J.~Dyson, \emph{Statistical theory of energy levels of complex systems {III}},
  J. Math. Phys. \textbf{3} (1962), 166--175.
  
  \bibitem{Dy62b}
F.J. Dyson, \emph{A {B}rownian motion model for the eigenvalues of a random
  matrix}, J. Math. Phys. \textbf{3} (1962), 1191--1198.
  
  \bibitem{Dy62c}
F.J.~Dyson, \emph{The three fold way. {Algebraic} structure of symmetry groups and
  ensembles in quantum mechanics}, J. Math. Phys. \textbf{3} (1962),
  1199--1215.
  
   \bibitem{DM63}
  F.J.~Dyson and M.L.~Mehta, \emph{Statistical theory of energy levels of complex
  systems. IV}, J.~Math.~Phys. 4 \textbf{4} (1963), 701--712.
  
  \bibitem{Dy70}
F.J.~Dyson, \emph{Correlations between eigenvalues of a random matrix},
Commun. Math. Phys. \textbf{29} (1970), 235--250.

 \bibitem{DZ96}
J. D'Anna and A. Zee,  \emph{Correlations between eigenvalues of large random matrices
with independent entries}, Phys Rev E \textbf{53} (1996), 1399.



\bibitem{Du18}
M. Duits, \emph{On global fluctuations for non-colliding processes}, Ann. Probab.
\textbf{46} (2018), 1279--1350.

\bibitem{DP18}
I. Dumitriu and E. Paquette,  \emph{Spectra of overlapping wishart matrices and the Gaussian
free field},  Rand. Mat. Th. Appl.,  \textbf{7} (2018), 1850003.



\bibitem{Ed97}
A.~Edelman, \emph{The probability that a random real {G}aussian matrix has $k$
  real eigenvalues, related distributions, and the circular law}, J.
  Multivariate. Anal. \textbf{60} (1997), 203--232


\bibitem{EKS94}
A. Edelman, E. Kostlan, and M. Shub, \emph{How many eigenvalues of a random matrix are real?}
J. Amer. Math. Soc. \textbf{7} (1994), 247.


\bibitem{Er11}
L. Erd\"os,  \emph{Universality of Wigner random matrices: a survey of recent results}, Russian Math.
Surveys \textbf{66} (2011), 67--198.
  
  
  
      \bibitem{Fo93a}
P.J.~Forrester, \emph{Statistical properties of the eigenvalue motion
of Hermitian matrices}, Phys.~Lett.~A {\bf 173} (1993), 355--359.

 \bibitem{Fo95}
P.J.~Forrester, \emph{Global fluctuation formulas and universal correlations for random
matrices and log-gas systems at infinite density}, Nucl. Phys. B \textbf{435} (1995), 421.
  
\bibitem{Fo99}
P.J.~Forrester, \emph{Fluctuation formula for complex random matrices}, J. Phys. A
  \textbf{32} (1999), L159--L163.
  
  
    \bibitem{Fo10}
P.J. Forrester, \emph{Log-gases and random matrices}, Princeton University Press,
  Princeton, NJ, 2010.
  
  \bibitem{Fo13}
  P.J.~Forrester,  \emph{The averaged characteristic polynomial for the Gaussian and chiral Gaussian
ensemble with a source}, J. Phys. A  \textbf{46} (2013), 345204.
  
   \bibitem{Fo18}
P.J.~Forrester, \emph{Meet Andr\'eief, Bordeaux 1886, and Andreev, Kharkov 1882--83},
Random Matrices Theory Appl. \textbf{8} (2019), 1930001.


 \bibitem{Fo21a}
P.J.~Forrester, \emph{Differential identities for the structure function of some random matrix ensembles},
 J. Stat. Phys. \textbf{183} (2021), 33.


 \bibitem{Fo21b}
 P.J.~Forrester, \emph{Circulant $L$-ensembles in the thermodynamic limit},
 J. Phys. A \textbf{54} (2021), 444003.
 
  \bibitem{FF04}
 P.J. Forrester and N.E. Frankel,  \emph{Applications and
generalizations of Fisher-Hartwig asymptotics}, J. Math. Phys. \textbf{45} (2004),
2003--2028
 
 \bibitem{FH98}
P.J. Forrester and G.~Honner, \emph{Exact statistical properties of the zeros
  of complex random polynomials}, J. Phys. A \textbf{32} (1999), 2961--2981.
  
  \bibitem{FJ96}
P.J. Forrester and B.~Jancovici, \emph{Two-dimensional one-component plasma in
  a quadrupolar field}, Int. J. Mod. Phys. A \textbf{11} (1996), 941--949.
  
  
  \bibitem{FJ97}
P.J. Forrester and B.~Jancovici, \emph{Exact and asymptotic formulas for overdamped {Brownian}
  dynamics}, Physica A \textbf{238} (1997), 405--424.
  
  \bibitem{FJM00}
P.J. Forrester, B.~Jancovici, and D.S. McAnally, \emph{Analytic properties of
  the structure function for the one-dimensional one-component log-gas}, J.
  Stat. Phys. \textbf{102} (2000), 737--780.


 
 
 \bibitem{FM15}
P.J. Forrester and A.~Mays, \emph{Finite-size corrections in random matrix
  theory and {O}dlykzko's dataset for the {R}iemann zeros}, Proc. Roy. Soc.
  A \textbf{471} (2015), 20150436.
  
  
\bibitem{FM21}
  P.J. Forrester and G.~Mazzuca, \emph{The classical $\beta$-ensembles with $\beta$ proportional to 
  $1/N$: from loop equations to Dyson's disordered chain}, J. Math. Phys. 62,  (2021) 073505.
  
  
  
  \bibitem{FN07}
P.J. Forrester and T.~Nagao, \emph{Eigenvalue statistics of the real {G}inibre
  ensemble}, Phys. Rev. Lett. \textbf{99} (2007), 050603.
  
      \bibitem{FRW17}
P.J. Forrester, A.A. Rahman, and N.S. Witte, \emph{Large $N$ expansions for the Laguerre and Jacobi $\beta$ ensembles from the loop equations}, J. Math. Phys. \textbf{58}
  (2017), 113303.
  
  
  \bibitem{FW17} 
  P.J. Forrester and D. Wang,  \emph{Muttalib-Borodin ensembles in random matrix theory --- realisations and correlation functions}, {\it Elec. J. Probab.} \textbf{22} (2017), 54
  
  \bibitem{FMP78}
J.B. French, P.A. Mello, and A.~Pandey, \emph{Statistical properties of
  many-particle spectra. {II}. {Two}-point correlations and fluctuations}, Ann.
  Phys. \textbf{113} (1978), 277--293.
  

\bibitem{FKS97}
Y.V. Fyodorov, B.A. Khoruzhenko, and H.-J. Sommers, \emph{Almost-{H}ermitian
  random matrices: crossover from {Wigner-Dyson} to {Ginibre} eigenvalue
  statistics}, Phys. Rev. Lett. \textbf{79} (1997), 557--560.  
  
  
  \bibitem{FL20} 
  Y.V.~Fyodorov and P. Le Doussal,  \emph{Statistics of extremes in eigenvalue-counting staircases},
Phys. Rev. Lett. \textbf{124} (2020), 210602. 
  

  \bibitem{Ga66}   
  M. Gaudin, \emph{Une famille \`a un param\`etre d'ensembles unitaires}, Nucl.
Phys. \textbf{85} (1966), 545--575.

\bibitem{GL17}  
  S. Ghosh and J.L. Lebowitz,  \emph{Fluctuations, large deviations and rigidity in hyperuniform
systems: a brief survey}, Indian J. Pure Appl. Math., \textbf{48} (2017), 609--631.

\bibitem{Gi65}
J.~Ginibre, \emph{Statistical ensembles of complex, quaternion, and real
  matrices}, J. Math. Phys. \textbf{6} (1965), 440.

  \bibitem{GNT17}
F. G\"otze, A. Naumov, and A. Tikhomirov, \emph{Distribution of linear statistics of singular values of the product of random matrices},
 Bernoulli \textbf{23}, (2017) 3067--3113.



 \bibitem{GS22}
 V.~Gorin and Y.~Sun, \emph{Gaussian fluctuations for products of random matrices},
American J. Math.
\textbf{144} 2022, 287--393.

\bibitem{Ga21}
A. Grabsch,  \emph{General truncated linear statistics for the top
eigenvalues of random matrices}, arXiv:2111.09004.

\bibitem{GMT17}
A. Grabsch, S. N. Majumdar and C. Texier, \emph{Truncated linear statistics associated with the
top eigenvalues of random matrices}, J. Stat. Phys.  \textbf{167},  234--259.


  


\bibitem{Ha56}
J.M. Hammersley, \emph{The zeros of random polynomials}, Proceedings of the
  Third Berkeley Symposium on Probability and Statistics (J.~Neyman, ed.),
  vol.~2, Univ. California Press, Berekeley, CA, 1956, pp.~89--111.
  
  \bibitem{Ha96}
J.~H. Hannay, \emph{Chaotic analytic zero points: exact statistics for those of
  a random spin state}, J. Phys. A \textbf{29} (1996), L101--L105.
  
 \bibitem{HK17}
  Y.~He and  A.~Knowles, \emph{Mesoscopic eigenvalue statistics of Wigner matrices}, Ann. Appl. Probab. \textbf{27} (2017), 1510 --1550.
  
  
  
 \bibitem{ILV88} 
  M.E.H.~Ismail, J.~Letessier and G.~Valent,  \emph{Linear birth and death models and associated 
  Laguerre and Meixner polynomials}, J. Approx. Theory  \textbf{55} (1988), 337--348.
  
 \bibitem{JSS16} 
  I. Jana, K. Saha, and A. Soshnikov,  \emph{Fluctuations of linear eigenvalue statistics of random
band matrices}, Theory Probab. Appl. \textbf{60} (2016), 407--443. 
  
  
  \bibitem{Ja95}
B.~Jancovici, \emph{Classical {Coulomb} systems: screening and correlations
  revisited}, J. Stat. Phys. \textbf{80} (1995), 445--459.
  
 \bibitem{JL19}  
  H.C. Ji, J.O. Lee. Gaussian fluctuations for linear spectral statistics of deformed Wigner matrices, Random Matrices :
Theory and Applications \textbf{9} (2020), 2050011


\bibitem{Jo88}
K.~Johansson, \emph{On {S}zeg\"o's formula for {Toeplitz} determinants and
  generalizations}, Bull. Sc. math., $2^{\rm e}$ s\'erie \textbf{112} (1988),
  257--304.
  
 \bibitem{Jo97}
  K.~Johansson,   \emph{On random matrices from the compact classical groups}, Ann. Math. 
\textbf{145} (1997), 519--545. 

\bibitem{Jo98}
 K.~Johansson,, \emph{On fluctuation of eigenvalues of random {Hermitian} matrices},
  Duke Math. J. \textbf{91} (1998), 151--204.

  
   \bibitem{Ha00}
F.~Haake, \emph{Quantum signatures of chaos}, 2nd ed., Springer, Berlin, 2000.




 \bibitem{KKP96}
A.~Khorunzhy, B.~Khoruzhenko and L.~Pastur,  \emph{Random matrices with 
independent entries: asymptotic properties of the Green function}, J. Math. Phys.
\textbf{37} (1996), 5033--5060.


\bibitem{Ki08}
R.~Killip, \emph{Gaussian fluctuations for $\beta$ ensembles}, Int. Math. Res.
  Not. \textbf{2008} (2008), rnn007. 

\bibitem{KN04}
R.~Killip and I.~Nenciu, \emph{Matrix models for circular ensembles}, Int.
  Math. Res. Not. \textbf{50} (2004), 2665--2701.
  
  \bibitem{Ko16}
P.~Kopel, \emph{Hermitian and non-Hermitian random matrix theory},
PhD thesis, University of California, Davis (2016).

  
  
\bibitem{KOV20}  
  P.~Kopel, S.~O'Rourke, and V.~Vu, \emph{Random matrix products: Universality and least
singular values}, Ann. Probab. \textbf{48} (2020), 1372--1410.
  
 \bibitem{La18}
 G.~Lambert, \emph{Limit theorems for biorthogonal ensembles and related combinatorial identities},
 Adv.~Math. \textbf{329}, (2018) 590--648.  
  
  
\bibitem{LLW19}  
  G. Lambert, M. Ledoux and C. Webb,  \emph{Quantitative normal approximation of linear statistics of 
$\beta$-ensembles}, Ann. Prob.  \textbf{47} (2019), 2619--2685.

\bibitem{LS22}
B.~Landon and P.~Sosoe, \emph{Almost-optimal bulk regularity conditions in
the CLT for Wigner matrices}, arXiv:2204.03419.


\bibitem{LS18}
T. Lebl\'e and S. Serfaty, \emph{Fluctuations of two dimensional Coulomb gases}, Geom. Funct. Anal.,
\textbf{28} (2018), 443--508.

\bibitem{LM10}
T. L\'evy, M. Ma\"ida, \emph{Central limit theorem for the heat kernel measure on the unitary group}, J. Funct. Anal.,
\textbf{259} (2010), 3163--3204.
  
  
\bibitem{LP09}  
  A.~Lytova,  and L.~Pastur, \emph{Central limit theorem for linear eigenvalue statistics of random matrices with independent entries},  Ann.Probab. \textbf{37}
  (2009), 
1778--1840.

\bibitem{LSX21}
 Y. Li. K. Schnelli. Y. Xu, \emph{Central limit theorem for mesoscopic eigenvalue statistics of deformed Wigner matrices and sample covariance matrices},
  Ann. Inst. H. Poincar\'e  Probab. Statist. \textbf{57} (2021), 506 --546.


 \bibitem{MY80}
Ph. Martin, T. Yalcin, \emph{The charge fluctuations in classical Coulomb systems},
J. Stat. Phys., \textbf{22} (1980), 435--463.

 \bibitem{Me04}
M.L. Mehta, \emph{Random matrices}, 3rd ed., Elsevier, San Diego, 2004.


\bibitem{MD63}
M.L. Mehta and F.J. Dyson, \emph{Statistical theory of the energy levels of
  complex systems. {V}}, J. Math. Phys. \textbf{4} (1963), 713--719.
  
  
\bibitem{MC20}  
  C. Min and Y. Chen, \emph{Linear statistics of random matrix ensembles at the spectrum edge associated with the Airy
kernel}, Nucl. Phys. B \textbf{950} (2020), 114836.
  
  
\bibitem{MN04}  
  J. Mingo and A. Nica,  \emph{Annular noncrossing permutations and partitions, and second-order
asymptotics for random matrices}, Int. Math. Res. Not., \textbf{2004} (2004), 1413--1460.

\bibitem{MST09}
J.A.~Mingo, R.~Speicher, and E.~Tan, \emph{Second order cumulants
of products}, Trans. Amer. Math. Soc., \textbf{361} (2009), 4751--4781.
  
  
   \bibitem{MMPS12}  
  A.D. Mironov, A.Yu. Morozov, A.V. Popolitov, and Sh.R.~Shakirov, \emph{Resolvents and Seiberg-Witten representation for a Gaussian $\beta$-ensemble}, Theor. Math. Phys., \textbf{171} (2012), 505--522.
  
  
 \bibitem{MT14}  
  M. B. McCoy and J. A. Tropp, \emph{From Steiner formulas for cones to concentration of intrinsic volumes},
Discrete Comput. Geom. \textbf{51} (2014), 926--963.

 \bibitem{NF03}
T.~Nagao and P.J.~Forrester, \emph{Dynamical correlations for circular ensembles of random matrices},
Nucl. Phys. B \textbf{660} (2003), 557--578.

  
\bibitem{NY16}
J. Najim, J. Yao, \emph{Gaussian fluctuations for linear spectral statistics of large random covariance matrices},  Ann. Appl.
Probab. \textbf{26} (2016), 1837--1887.  

\bibitem{NTT21}
F. Nakano, H. D. Trinh, K. D. Trinh,  \emph{Limit theorems for moment processes of beta
Dyson's Brownian motions and beta Laguerre processes}, arXiv:2103.09980.


\bibitem{OR16}
S. O'Rourke and D. Renfrew, \emph{Central limit theorem for linear eigenvalue statistics of elliptic random
matrices}, J. Theoret. Probab.  \textbf{29} (2016), 1121--1191. 



\bibitem{ORS11}
S. O'Rourke, D. Renfrew and A. Soshnikov, \emph{On fluctuations of matrix entries of regular
functions of Wigner matrices with non-identically distributed entries}, J. Theor. Probab.
\textbf{26} (2013), 750--780.


  
  
\bibitem{Pa81}  
A. ~Pandey,  \emph{Statistical properties of many-particle spectra. IV. new ensembles by Stieltjes transform
methods}, Ann. Phys. \textbf{134} (1981), 110--127.

\bibitem{PS91}
A. Pandey and P. Shukla,  \emph{Eigenvalue correlations in the circular ensembles}, J. Phys. A  \textbf{24} (1991),
3907--3926.

\bibitem{Pa13}
E.~Paquette, \emph{Eigenvalue Fluctuations of random matrices beyond the Gaussian
universality class}, PhD thesis, University of Washington, 2013.

\bibitem{Pa72}
L.A. Pastur, \emph{On the spectrum of random matrices}, Teor. Mat. Fiz.  \textbf{10},
(1972), 102--112. 

   \bibitem{PS11}
L.~Pastur and M.~Shcherbina, \emph{Eigenvalue distribution of large random
  matrices}, American Mathematical Society, Providence, RI, 2011.


 \bibitem{Pe83}
  P.~Pechukas,  \emph{Distribution of energy eigenvalues in the irregular spectrum}, Phys. Rev. Lett., \textbf{51} (1983), 943--946.
  
    \bibitem{PZ11}
K.~Penson and K.~Zyczkowski, \emph{Product of {G}inibre matrices:
  {F}uss-{C}atalan and {R}aney distributions}, Phys. Rev. E \textbf{83} (2011),
  061118.

  
  \bibitem{Po65}
C.E. Porter, \emph{Statistical theories of spectra: fluctuations}, Academic
  Press, New York, 1965.
  
  
  

  
\bibitem{RV07}  
  B. Rider and B. Vir\'ag,  \emph{The noise in the circular law and the Gaussian free field}, Int.
Math. Res. Not. 2007 \textbf{2007}, rnm006.
  
  
  \bibitem{Sa21}
 A.J.~Sargeant,  \emph{Numerical simulation of GUE two-point correlation and cluster functions}, Braz. J Phys. \textbf{51} (2021), 308--315.
 
 \bibitem{Se16}
 S.~Serfaty, \emph{Large systems with Coulomb interactions: variational study and statistical mechanics},
 Portugaliae Mathematica \textbf{73} (2016), 247--278.
  
 
  \bibitem{Sh11}
 M.~Shcherbina, \emph{Central limit theorem for linear eigenvalue statistics of the Wigner and sample covariance
random matrices},  J. Math. Physics, Analysis, Geometry, \textbf{7}  (2011), 176--192.

 \bibitem{Sh13}
M. Shcherbina, \emph{Fluctuations of linear eigenvalue statistics of $\beta$ matrix models in the multi-cut
regime},  J. Stat. Phys. \textbf{151} (2013), 1004--1034.

 \bibitem{Sh15}
M.~Shcherbina, \emph{Fluctuations of the eigenvalue number in the fixed interval for $\beta$-models with $\beta = 1,2,4$},
pp.~131--146 in ``Theory and Applications in Mathematical Physics", Ed.~E.~Agliari et al., World Scientific, 2015.

 \bibitem{Sh15a}
M. Shcherbina,  \emph{On fluctuations of eigenvalues of random band matrices}, J. Stat. Phys. \textbf{161}
(2015), 73--90.


 \bibitem{SLMS21}
N.R.~Smith, P. Le Doussal, S.N.~Majumdar, and G.~Schehr, \emph{Counting statistics for non-interacting
fermions in a d-dimensional potential}, Phys. Rev. E \textbf{103} (2021), L030105.

 \bibitem{SLMS21a}
N. R. Smith, P. Le Doussal, S. N. Majumdar, and G. Schehr,  \emph{Full counting statistics for interacting
trapped fermions}, SciPost Physics  \textbf{11} (2021), 110.
  
  \bibitem{So00}
  A. Soshnikov,  \emph{The central limit theorem for local linear statistics in classical compact groups and
related combinatorial identities}, Ann. Probab.  \textbf{28} (2000), 1353--1370.

 \bibitem{Sp87}
H.~Spohn,  \emph{Interacting Brownian particles: a study of Dyson's model}, In:
Hydrodynamic Behavior and Interacting Particle Systems, ed. by G.C. Pa-
panicolaou, IMA Volumes in Mathematics \textbf{9} , Springer-Verlag (1987) 151--179.


\bibitem{Su71a}
B.~Sutherland, \emph{Exact results for a quantum many body problem in one
  dimension}, Phys. Rev. A \textbf{4} (1971), 2019--2021.


 \bibitem{To16}
S. Torquato, \emph{Hyperuniformity and its generalizations}, Phys. Rev. E  \textbf{94} (2016),
022122.

 \bibitem{TS03}
S. Torquato and F. Stillinger,  \emph{Local density fluctuations, hyperuniformity, and
order metrics}, Phys. Rev. E \textbf{68} (2003),  041113 .

\bibitem{TW98}
C.A.~Tracy and H.~Widom, \emph{Correlation functions, cluster functions and spacing
  distributions in random matrices}, J. Stat. Phys. \textbf{92} (1998),
  809--835.
  
  
\bibitem{Tu62}  
  W.T.~Tutte, \emph{A census of slicings},  Can. J. Math. \textbf{14} (1962), 708--22.
  
   \bibitem{Wa78}
K.W. Wachter, \emph{The strong limits of random matrix spectra for sample
  matrices of independent elements}, Annal. Prob. \textbf{6} (1978), 1--18.
  
  \bibitem{Wi55} 
 E.P.~Wigner, \emph{Characteristic vectors of bordered matrices with infinite dimensions}, Ann. Math. \textbf{62} (1955), 548--564.  
 
 \bibitem{WF14}
N.S. Witte and P.J. Forrester, \emph{Moments of the {G}aussian $\beta$ ensembles
  and the large {$N$} expansion of the densities}, J. Math. Phys. \textbf{55}
  (2014), 083302.
  
   \bibitem{WF15}
N.S. Witte and P.J. Forrester, \emph{Loop equation analysis of the circular ensembles}, JHEP
  \textbf{2015} (2015), 173.
  
  
  
  \bibitem{Yu86}
  T. Yukawa, \emph{Lax form of the quantum eigenvalue problem},
  Phys. Lett. A \textbf{116} (1986), 227--230.
  
     \end{thebibliography}
\nopagebreak

\providecommand{\bysame}{\leavevmode\hbox to3em{\hrulefill}\thinspace}
\providecommand{\MR}{\relax\ifhmode\unskip\space\fi MR }
\providecommand{\MRhref}[2]{%
  \href{http://www.ams.org/mathscinet-getitem?mr=#1}{#2}
}
\providecommand{\href}[2]{#2}

\end{document}